\documentclass[%
 amsmath,amssymb,nofootinbib,
 reprint,%
]{revtex4-1}

\usepackage{graphicx}
\usepackage{dcolumn}
\usepackage{bm}
\usepackage{tree-dvips}
\usepackage[english]{babel}
\usepackage{array}
\usepackage{eqnarray,amsmath}
\usepackage{bbold}
\usepackage{esvect}
\usepackage{braket}
\usepackage{wrapfig}
\usepackage{graphicx}
\usepackage{caption}
\usepackage{float}
\usepackage{subcaption}\usepackage{makeidx}
\usepackage{lipsum}
\usepackage{fancyhdr}
\usepackage{ragged2e}
\usepackage{amsmath}
\usepackage{amssymb}
\usepackage{subcaption}
\usepackage{amsthm}
\usepackage{enumitem}
\usepackage{pagecolor}
\usepackage{tikz}
\usepackage[utf8]{inputenc}
\usepackage[toc,page]{appendix}
\usepackage{sidecap}
\usepackage{calrsfs}
\usepackage[colorlinks=true,citecolor=red,linkcolor=brown,urlcolor=blue]{hyperref}

\usetikzlibrary{decorations.pathreplacing,angles,quotes}

\newtheorem{theorem}{Theorem}
\newtheorem{lemma}{Lemma}

\newtheorem{mydef}{Definition}

\newtheorem{sub-section}{}

\usetikzlibrary{shapes.misc, positioning}

\tikzset{cross/.style={cross out, draw=black, minimum size=2*(#1-\pgflinewidth), inner sep=0pt, outer sep=0pt},
cross/.default={1pt}}

\DeclareMathAlphabet{\pazocal}{OMS}{zplm}{m}{n}
\newcommand{\E}{\pazocal{E}}

\newcommand{\F}{\pazocal{F}}

\newcommand{\h}{\pazocal{H}}

\newcommand{\p}{\pazocal{P}}

\begin{document}

\title[Sample title]{Reducing resources for verification of quantum computations}

\author{Samuele Ferracin}\email{S.Ferracin@warwick.ac.uk}
\author{Theodoros Kapourniotis}\email{T.Kapourniotis@warwick.ac.uk}
\author{Animesh Datta}\email{animesh.datta@warwick.ac.uk}
\affiliation{Department of Physics, University of Warwick, Coventry CV4 7AL, United Kingdom}

\date{\today}

\begin{abstract} 
{We present two verification protocols where the correctness of a ``target'' computation is checked by means of ``trap'' computations that can be efficiently simulated on a classical computer. Our protocols rely on a minimal set of noise-free operations (preparation of eight single-qubit states or measurement of four observables, both on a single plane of the Bloch sphere) and achieve linear overhead. To the best of our knowledge, our protocols are the least demanding techniques able to achieve linear overhead. They represent a step towards further reducing the quantum requirements for verification.}

\end{abstract}

\pacs{Valid PACS appear here}
\maketitle
\section{Introduction}

Scalable quantum computers are expected to decidedly widen our computing paradigm, providing polynomial or even exponential speed-ups in the solution of certain classes of problems \cite{NC00}. While the realization  of such devices remains out of reach, the development of small-sized and imperfect quantum computers appears to be an achievable goal within the next few years \cite{W16,B&al17}. In principle, these ``first-generation'' quantum computers should be capable of outperforming classical computers in some tasks and consequently provide evidence of the so-called ``{supremacy}'' of quantum computers over classical ones \cite{AC16}. However, the faulty functioning of their inner components represents a major issue. In this scenario, a fundamental question arises: how can we get confidence about the correct functioning of a quantum computer, if the correctness of the outcome cannot be checked easily? That is to say, \textit{is it possible to verify the correctness of the outcome of a quantum computation}? 

Eventually, quantum computers may not need to be verified at all. Advances in error correction and fault-tolerance might increase their reliability up to a point where verifying their functioning becomes redundant. Nevertheless, the above question has a deep value, which goes far beyond the mere scope of verifying a quantum computation. It refers to the possibility of using classical logic to verify that a given quantum phenomenon behaves as predicted in the high complexity regime. Addressing this question would provide a deeper understanding of some aspects of computational complexity theory, such as the relationship between the classes BQP (informally, the class of problems that can be solved by a quantum computer in a polynomial time) and interactive proofs \cite{AV12,ABE08}. Also, verification techniques will be needed by ``delegated computations'', where private users access a quantum computer remotely and malevolent third-parties might try to tamper the computation \citep{BFK09}. For all these reasons, much effort has been made in the last decade in an area that can be called ``{quantum verification}'' \cite{F85,AS06,BFK09,FK12,RUV12,HM15,GKW15,M16,F16,GKK17}. 

In the previous works, it has been shown that it is possible to verify an arbitrary quantum computation in scenarios where some limited assumptions are made. Based on the assumptions they rely on, the existing protocols can be divided into two broad classes. The first is the class of protocols where a given operation performed during the computation (such as state preparation \cite{FK12,B15,KD17,KW17} or measurement \cite{HM15,MF16,HKSE17}) is ``{trusted}'', and consequently treated as an ideal and noise-free \textit{resource}. On the other hand, the protocols in the second class consider computations run on two ``untrusted'' entangled servers that cannot communicate with each other \cite{RUV12,GKW15,HPF15,M16,HH16,NV16}. The possibility of verifying a quantum computation without any kind of assumption remains open.

In what follows, we will mainly be concerned with the first class of verification protocols. In those works, the problem of verification is illustrated in terms of an interactive game between Alice (the ``verifier'') and Bob (the ``prover'') of the following kind. Alice, endowed with some restricted quantum power, wants to run a universal quantum computation on Bob's quantum computer. However, Alice has no guarantee that Bob will follow her instructions. Thus, she wants to run the computation in such a way that if Bob is {dishonest}, she expects to be aware of it.

Although games of this kind are interesting for studying the security of delegated computations, some of their main concerns (such as blindness, namely Alice's ability to encrypt her own instructions) might appear more germane to cryptography than to verification, given that it seems fairly unlikely that quantum computers will actively conspire against us. Nevertheless, the reason why the cryptographic approach is adopted is two-fold. On the one hand, we do not know how to solve the  problem of verification in a different way: so far, it has not been possible to find alternative approaches that are also scalable and unconditionally secure - as an example, a verification technique based on the statistics of the outcomes is presented in Ref. \cite{B&al16}; however, its validity relies on assumptions on the noise affecting the experimental setup \cite{BFNV18}. On the other hand, the cryptographic approach is perfectly consistent with the spirit of verification. To some extent, Bob represents \textit{all that can go wrong} along a computation: his attempts to harm the computation correspond to all of the possible sources of errors able to corrupt the computation itself.

As mentioned above, verifying quantum computations with cryptographic techniques requires that the verifier has access to a number of ideal resources. Our work represents a step forward towards minimising the resources required by the existing verification protocols. This advance is clearly valuable from both the experimental and the theoretical perspective: devising least demanding protocols makes verification easier to implement experimentally and, at the same time, helps making progresses towards understanding what is the smallest set of noise-free operations needed by the verifier to decide with confidence on the correctness of an arbitrary quantum computation. 

In what follows, we present the following results:
\begin{itemize}
\item[(i) ]In Section \ref{sec:verifSP} we present Protocol \hyperlink{pr:pr1}{1}, {a verification protocol where the correctness of a quantum computation is tested by checking the outcome of several other classically efficiently simulable computations}. We prove its effectiveness, modulo a trust assumption regarding the preparation of single qubits in the states $\{|+\rangle_{\theta}=(\ket{0}+e^{i\theta}\ket{1})/\sqrt{2}\}$, $\theta\in\{0,\pi/4,..,7\pi/4\}$.

\item[(ii)]In Section \ref{sec:verifM}, we show that the above protocol can be adapted to a scenario where the trust assumption is made on the measurement in the set of bases $\{|\pm\rangle_{\phi}\langle\pm|\}=\{R_Z(\phi)|\pm\rangle\langle\pm|R^\dagger_Z(\phi)\}$, $\phi\in\{0,\pi/4,..,7\pi/4\}$ (or equivalently, on the measurement of the observables $X\textrm{, }Y\textrm{, }(X\pm Y)/\sqrt{2}$), provided that the qubits can be reused after the measurement. Based on this, we present Protocol \hyperlink{pr:pr2}{2}.
\end{itemize}
In the language of cryptographic protocols:
\begin{itemize}
\item[(i) ]Supposing that Alice can prepare single qubits in the set of states $\{|+\rangle_{\theta}\}$, $\theta\in\{0,\pi/4,..,7\pi/4\}$, we show how she can verify Bob's behaviour by {hiding} her ``target'' computation among a given number of ``trap'' computations, whose outcomes are easy to compute on a classical computer and can thus be compared to the obtained ones.
\item[(ii)]We show how the above mentioned protocol can be adapted to the case of Alice making local measurement in the set of bases $\{|\pm\rangle_{\phi}\langle\pm|\}$, $\phi\in\{0,\pi/4,..,7\pi/4\}$ (or equivalently, of Alice measuring the observables $X\textrm{, }Y\textrm{, }(X\pm Y)/\sqrt{2}$) and resending already measured qubits to Bob.
\end{itemize}
Compared to prior works, our protocols rely on a set of resources which is minimal and that is restricted to operations contained on a single plane of the Bloch sphere. In more detail, Protocol \hyperlink{pr:pr1}1 requires noise-free preparation of eight types of single-qubit states (as opposed to \cite{FK12,B15,KD17,KW17}, which requires perfect preparation of ten types of states), while Protocol \hyperlink{pr:pr2}2 requires noise-free measurement of four observables (as opposed to \cite{HM15}, which requires trusted measurements of five observables). On the other 


\newpage
\begin{figure}[H]
\centering
\begin{tikzpicture}[scale=1.3, every node/.style={scale=1.1}]

\draw [->,thick] (0.15,3.1) -- (2.7,3.1);
\node at (2.7,3.3) {$y$};

\filldraw[fill=green!20!white, draw=black, opacity=0.6] (0.0-0.25, 0.0-0.2) rectangle (0.0+2.25, 0.0+0.7);
\filldraw[fill=green!20!white, draw=black, opacity=0.6] (0.0-0.25, 1.0-0.2) rectangle (0.0+2.25, 1.0+0.7);
\filldraw[fill=green!20!white, draw=black, opacity=0.6] (0.0-0.25, 2.0-0.2) rectangle (0.0+2.25, 2.0+0.7);

\filldraw[fill=green!20!white, draw=black, opacity=0.6] (2.0-0.25, 0.5-0.2) rectangle (2.0+2.25, 0.5+0.7);
\filldraw[fill=green!20!white, draw=black, opacity=0.6] (2.0-0.25, 1.5-0.2) rectangle (2.0+2.25, 1.5+0.7);

\filldraw[fill=green!20!white, draw=black, opacity=0.6] (4.0-0.25, 0.0-0.2) rectangle (4.0+2.25, 0.0+0.7);
\filldraw[fill=green!20!white, draw=black, opacity=0.6] (4.0-0.25, 1.0-0.2) rectangle (4.0+2.25, 1.0+0.7);
\filldraw[fill=green!20!white, draw=black, opacity=0.6] (4.0-0.25, 2.0-0.2) rectangle (4.0+2.25, 2.0+0.7);

\draw [red,line width=0.75mm,dashed] (0.0-0.25,-0.5) -- (0.0-0.25,3.0);
\draw [red,line width=0.75mm,dashed] (0.0+1.75,-0.5) -- (0.0+1.75,3.0);
\draw [red,line width=0.75mm,dashed] (2.0+1.75,-0.5) -- (2.0+1.75,3.0);
\draw [red,line width=0.75mm,dashed] (4.0+1.75,-0.5) -- (4.0+1.75,3.0);

\draw (0.0,0.0) -- (6.0,0.0);
\draw (0.0,0.5) -- (6.0,0.5);
\draw (0.0,1.0) -- (6.0,1.0);
\draw (0.0,1.5) -- (6.0,1.5);
\draw (0.0,2.0) -- (6.0,2.0);
\draw (0.0,2.5) -- (6.0,2.5);

\draw (1.0,0.0) -- (1.0,0.5);
\draw (1.0,1.0) -- (1.0,1.5);
\draw (1.0,2.0) -- (1.0,2.5);
\draw (2.0,0.0) -- (2.0,0.5);
\draw (2.0,1.0) -- (2.0,1.5);
\draw (2.0,2.0) -- (2.0,2.5);

\draw (3.0,0.5) -- (3.0,1.0);
\draw (3.0,1.5) -- (3.0,2.0);
\draw (4.0,0.5) -- (4.0,1.0);
\draw (4.0,1.5) -- (4.0,2.0);

\draw (5.0,0.0) -- (5.0,0.5);
\draw (5.0,1.0) -- (5.0,1.5);
\draw (5.0,2.0) -- (5.0,2.5);
\draw (6.0,0.0) -- (6.0,0.5);
\draw (6.0,1.0) -- (6.0,1.5);
\draw (6.0,2.0) -- (6.0,2.5);

\draw [fill=white] (0.0,2.5) circle [radius=0.05];
\draw [fill=white] (0.5,2.5) circle [radius=0.05];
\draw [fill=white] (1.0,2.5) circle [radius=0.05];
\draw [fill=white] (1.5,2.5) circle [radius=0.05];
\draw [fill=white] (2.0,2.5) circle [radius=0.05];
\draw [fill=white] (2.5,2.5) circle [radius=0.05];
\draw [fill=white] (3.0,2.5) circle [radius=0.05];
\draw [fill=white] (3.5,2.5) circle [radius=0.05];
\draw [fill=white] (4.0,2.5) circle [radius=0.05];
\draw [fill=white] (4.5,2.5) circle [radius=0.05];
\draw [fill=white] (5.0,2.5) circle [radius=0.05];
\draw [fill=white] (5.5,2.5) circle [radius=0.05];
\draw [fill=white] (6.0,2.5) circle [radius=0.05];

\draw [fill=white] (0.0,2.0) circle [radius=0.05];
\draw [fill=white] (0.5,2.0) circle [radius=0.05];
\draw [fill=white] (1.0,2.0) circle [radius=0.05];
\draw [fill=white] (1.5,2.0) circle [radius=0.05];
\draw [fill=white] (2.0,2.0) circle [radius=0.05];
\draw [fill=white] (2.5,2.0) circle [radius=0.05];
\draw [fill=white] (3.0,2.0) circle [radius=0.05];
\draw [fill=white] (3.5,2.0) circle [radius=0.05];
\draw [fill=white] (4.0,2.0) circle [radius=0.05];
\draw [fill=white] (4.5,2.0) circle [radius=0.05];
\draw [fill=white] (5.0,2.0) circle [radius=0.05];
\draw [fill=white] (5.5,2.0) circle [radius=0.05];
\draw [fill=white] (6.0,2.0) circle [radius=0.05];

\draw [fill=white] (0.0,1.5) circle [radius=0.05];
\draw [fill=white] (0.5,1.5) circle [radius=0.05];
\draw [fill=white] (1.0,1.5) circle [radius=0.05];
\draw [fill=white] (1.5,1.5) circle [radius=0.05];
\draw [fill=white] (2.0,1.5) circle [radius=0.05];
\draw [fill=white] (2.5,1.5) circle [radius=0.05];
\draw [fill=white] (3.0,1.5) circle [radius=0.05];
\draw [fill=white] (3.5,1.5) circle [radius=0.05];
\draw [fill=white] (4.0,1.5) circle [radius=0.05];
\draw [fill=white] (4.5,1.5) circle [radius=0.05];
\draw [fill=white] (5.0,1.5) circle [radius=0.05];
\draw [fill=white] (5.5,1.5) circle [radius=0.05];
\draw [fill=white] (6.0,1.5) circle [radius=0.05];

\draw [fill=white] (0.0,1.0) circle [radius=0.05];
\draw [fill=white] (0.5,1.0) circle [radius=0.05];
\draw [fill=white] (1.0,1.0) circle [radius=0.05];
\draw [fill=white] (1.5,1.0) circle [radius=0.05];
\draw [fill=white] (2.0,1.0) circle [radius=0.05];
\draw [fill=white] (2.5,1.0) circle [radius=0.05];
\draw [fill=white] (3.0,1.0) circle [radius=0.05];
\draw [fill=white] (3.5,1.0) circle [radius=0.05];
\draw [fill=white] (4.0,1.0) circle [radius=0.05];
\draw [fill=white] (4.5,1.0) circle [radius=0.05];
\draw [fill=white] (5.0,1.0) circle [radius=0.05];
\draw [fill=white] (5.5,1.0) circle [radius=0.05];
\draw [fill=white] (6.0,1.0) circle [radius=0.05];

\draw [fill=white] (0.0,0.5) circle [radius=0.05];
\draw [fill=white] (0.5,0.5) circle [radius=0.05];
\draw [fill=white] (1.0,0.5) circle [radius=0.05];
\draw [fill=white] (1.5,0.5) circle [radius=0.05];
\draw [fill=white] (2.0,0.5) circle [radius=0.05];
\draw [fill=white] (2.5,0.5) circle [radius=0.05];
\draw [fill=white] (3.0,0.5) circle [radius=0.05];
\draw [fill=white] (3.5,0.5) circle [radius=0.05];
\draw [fill=white] (4.0,0.5) circle [radius=0.05];
\draw [fill=white] (4.5,0.5) circle [radius=0.05];
\draw [fill=white] (5.0,0.5) circle [radius=0.05];
\draw [fill=white] (5.5,0.5) circle [radius=0.05];
\draw [fill=white] (6.0,0.5) circle [radius=0.05];

\draw [fill=white] (0.0,0.0) circle [radius=0.05];
\draw [fill=white] (0.5,0.0) circle [radius=0.05];
\draw [fill=white] (1.0,0.0) circle [radius=0.05];
\draw [fill=white] (1.5,0.0) circle [radius=0.05];
\draw [fill=white] (2.0,0.0) circle [radius=0.05];
\draw [fill=white] (2.5,0.0) circle [radius=0.05];
\draw [fill=white] (3.0,0.0) circle [radius=0.05];
\draw [fill=white] (3.5,0.0) circle [radius=0.05];
\draw [fill=white] (4.0,0.0) circle [radius=0.05];
\draw [fill=white] (4.5,0.0) circle [radius=0.05];
\draw [fill=white] (5.0,0.0) circle [radius=0.05];
\draw [fill=white] (5.5,0.0) circle [radius=0.05];
\draw [fill=white] (6.0,0.0) circle [radius=0.05];

%
\node at (2.0+0.2,1.5+0.25) {\footnotesize ($i,j$)};
%

\end{tikzpicture}
\caption{\small Example of BwS (see \cite{BFK09} for more details). The circles represent qubits, the edges represent $cZ$ gates and the green boxes represent the 10-qubit bricks. In the rest of the paper, we will label physical qubits with indices $i$ and $j$. We divide the BwS into ``{tapes}'' (four-column layers between dashed lines) and label them with index $y=1,..,w$. Any $n\times m$ BwS is composed by $w=(m-1)/4$ tapes.}
\label{fig:BSmain}
\end{figure}
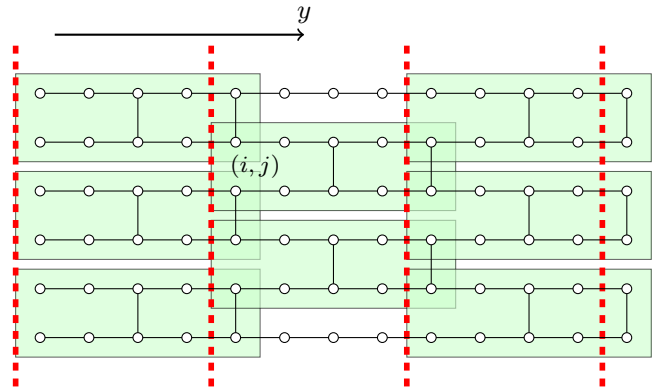
\noindent hand, the overhead of our protocols (i.e. the number of operations required) is linear, as well as the overhead of the most efficient existing protocols \cite{B15,KD17,KW17}. As far as we know, our protocols are the only ones achieving linear overhead while employing such a minimal set of resources.

The paper is structured as follows. In Section \ref{sec:background} we introduce measurement-based quantum computing (MBQC) and formally define cryptographic protocols.  In Section \ref{sec:comparison} we provide a detailed comparison between our protocols and the existing ones. In Section \ref{sec:verifSP} we illustrate our verification protocol and prove its validity when state preparation is ideal. In Section \ref{sec:verifM} we show how to adapt our verification protocol to the case of trusted measurements.

We use the following notation. We denote $XY$-plane rotations (respectively $ZY$-plane rotations) by angle $\phi$ as $R_Z(\phi)=\textup{diag}(1,e^{i\phi})$, and refer to them as ``$R_Z$-gate'' (respectively as $R_X(\phi)=HR_Z(\phi)H$, and refer to them as ``$R_X$-gate''). We denote the controlled-$Z$ gate as $cZ$ and the controlled-$X$ gate as CNOT.
\section{Background: MBQC and cryptographic protocols}~\label{sec:background}
In this section, we provide an introduction to MBQC and cryptographic games.

\subsection{Measurement-based quantum computing}~\label{sec:backgroundMBQC}
MBQC is a model for universal quantum computation equivalent to the circuit model \cite{RB01}. In MBQC, the computation is implemented throughout adaptive measurements of qubits belonging to a large entangled resource

\newpage
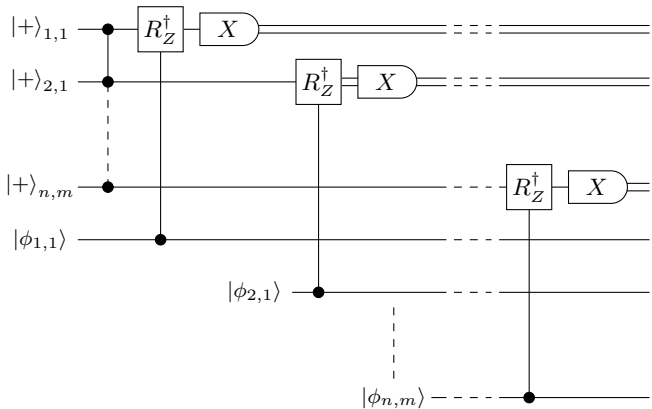
\begin{figure}[H]
	\begin{tikzpicture}
	\draw (0,4.2) -- (2.2,4.2);
	\draw (2.2,4.25) -- (4.9,4.25);
	\draw (2.2,4.15) -- (4.9,4.15);
	\draw [dashed] (5,4.25) -- (5.5,4.25);
	\draw [dashed] (5,4.15) -- (5.5,4.15);
	\draw (5.6,4.25) -- (7.6,4.25);
	\draw (5.6,4.15) -- (7.6,4.15);
	
	
	\draw (0,3.5) -- (3.5,3.5);
	\draw (3.5,3.55) -- (4.9,3.55);
	\draw (3.5,3.45) -- (4.9,3.45);
	\draw [dashed] (5,3.55) -- (5.5,3.55);
	\draw [dashed] (5,3.45) -- (5.5,3.45);
	\draw (5.6,3.55) -- (7.6,3.55);
	\draw (5.6,3.45) -- (7.6,3.45);
	
	
	\draw (0,2.1) -- (4.9,2.1);
	\draw [dashed] (5,2.1) -- (5.5,2.1);
	\draw (5.6,2.1) -- (7,2.1);
	\draw (7,2.15) -- (7.6,2.15);
	\draw (7,2.05) -- (7.6,2.05);
	
	\draw (0,1.4) -- (4.9,1.4);
	\draw [dashed] (5,1.4) -- (5.5,1.4);
	\draw (5.6,1.4) -- (7.6,1.4);
	
	\draw (2.85,0.7) -- (4.9,0.7);
	\draw [dashed] (5,0.7) -- (5.5,0.7);
	\draw (5.6,0.7) -- (7.6,0.7);
	
	\draw (4.7,-0.7) -- (4.9,-0.7);
	\draw [dashed] (5,-0.7) -- (5.5,-0.7);
	\draw (5.6,-0.7) -- (7.6,-0.7);
	
	\draw [dashed] (4.2,0.5) -- (4.2,-0.5);
	
	\draw (0.4,3.5) -- (0.4,4.2);
	\draw [dashed] (0.4,2.1) -- (0.4,3.5);
	\draw [fill] (0.4,4.2) circle [radius=0.7mm];
	\draw [fill] (0.4,3.5) circle [radius=0.7mm];
	\draw [fill] (0.4,2.1) circle [radius=0.7mm];
	
	\node at (-0.5,4.2) {\footnotesize $\ket{+}_{1,1}$};
	\node at (-0.5,3.5) {\footnotesize $\ket{+}_{2,1}$};
	\node at (-0.5,2.1) {\footnotesize $\ket{+}_{n,m}$};
	\node at (-0.5,1.4) {\footnotesize $\ket{\phi_{1,1}}$};
	\node at (2.35,0.7) {\footnotesize $\ket{\phi_{2,1}}$};
	\node at (4.2,-0.7) {\footnotesize $\ket{\phi_{n,m}}$};
	
	\draw (1.1,4.2) -- (1.1,1.4);
	\draw [fill] (1.1,1.4) circle [radius=0.7mm];
	\draw (3.2,3.5) -- (3.2,0.7);
	\draw [fill] (3.2,0.7) circle [radius=0.7mm];
	\draw (6,2.1) -- (6,-0.7);
	\draw [fill] (6,-0.7) circle [radius=0.7mm];
	
	\draw [fill=white] (0.8,3.9) rectangle (1.4,4.5);
	\node at (1.1,4.2) {\small $R^\dagger_{Z}$};
	\draw [fill=white] (2.9,3.2) rectangle (3.5,3.8);
	\node at (3.2,3.5) {\small $R^\dagger_{Z}$};
	\draw [fill=white] (5.7,1.8) rectangle (6.3,2.4);
	\node at (6.0,2.1) {\small $R^\dagger_Z$};
	
	
	\node [fill=white,draw,rounded rectangle,rounded rectangle left arc=none,minimum width=1cm] (name) at (2.0,4.2) {$X$};
	\node [fill=white,draw,rounded rectangle,rounded rectangle left arc=none,minimum width=1cm] (name) at (4.1,3.5) {$X$};
	\node [fill=white,draw,rounded rectangle,rounded rectangle left arc=none,minimum width=1cm] (name) at (6.9,2.1) {$X$};
	
	\end{tikzpicture}
	\caption{\small Representation of a computation on an $n\times m$ BwS (Figure \ref{fig:BSmain}) in the circuit model. First, the BwS is generated by applying a global entangling operation to $nm$ qubits in the state $\ket{+}$. Next, each qubit is rotated and subsequently measured in the Pauli-$X$ basis (this is equivalent to a measurement in one of the rotated bases $\{|\pm\rangle_\phi\langle\pm|\}$). Since the measurements are performed adaptively, the rotations are represented as controlled operations \cite{FK12}.}
	\label{fig:MBQC1circ1}
\end{figure}
\noindent  state. The qubits are initialized in the state $\ket{+}$ and entangled to their nearest-neighbors \textrm{via} a $cZ$ operation. When a qubit is measured, the rest of the resource state is modified in a way that can equivalently be described by a series of gates in the circuit model. 

An important resource state is the 10-qubit ``brick'', which can be used to generate the so-called ``{brickwork state}'' (BwS; see Figure \ref{fig:BSmain}). The BwS is universal for quantum computation, provided that the measurements - which are performed column-by-column from left to right in one of the bases $\{|\pm\rangle_{\phi}\langle\pm|\}=\{R_Z(\phi)|\pm\rangle\langle\pm|R^\dagger_Z(\phi)\}$, $\phi\in\{0,\pi/4,..,7\pi/4\}$ - are performed adaptively \citep{BFK09}. By ``{adaptive} measurement'', we mean that after the measurement of qubit $(i,j)$, the angles $\phi_{i^{\prime},j^{\prime}}$ of yet-to-be-measured qubits are recomputed as $(-1)^{s_X}\phi_{i^{\prime},j^{\prime}}+s_Z\pi$, where $s_X$ and $s_Z$ are computed on the basis of measurement outcomes of previous qubits (see \cite{DK05} for further details). This can be seen as a ``correction'', in the sense that the dependency of the computation on the measurement outcomes vanishes. Thus, adaptive measurements allow to implement computations in a deterministic way, regardless of the non-deterministic nature of quantum measurements.

In the rest of the paper, the computations on a $n\times m$ BwS will often be described through their corresponding logical circuit (Figure \ref{fig:MBQC1circ1}). In the circuit model representation, a measurement in the basis $\{|\pm\rangle_{\phi}\langle\pm|\}$ is expressed as the rotation $R^{\dagger}_Z(\phi)$ followed by a measurement in the Pauli-$X$ basis $\{|\pm\rangle\langle\pm|\}$. Since the angles are recomputed after every measurement, the rotations are expressed as controlled-$R_Z$ gates. 

To simplify the circuit in Figure \ref{fig:MBQC1circ1}, we notice that the measurement of each physical qubit $(i,j)$ (with $j=1,..,m-1$) breaks the entanglement between qubit $(i,j)$ and the rest of the BwS. At the same time, the state of qubit $(i,j)$ is teleported to qubit $(i,j+1)$ modulo the unitary $HR^\dagger_Z(\phi_{i,j})$ (the measurements are performed adaptively, hence the dependence of the unitary on the measurement outcomes can be omitted). Since every measured qubit is discarded, the outcome of the computation can equivalently be obtained by means of the logical circuit in Figure \ref{fig:BS}. 

Notice that the circuits represented in Figures \ref{fig:MBQC1circ1} and \ref{fig:BS} describe the same computation at two different levels. To highlight this, we will often distinguish between ``physical'' qubits (the qubits belonging to the BwS and processed along the circuit in Figure \ref{fig:MBQC1circ1}) and ``logical'' qubits (the qubits processed along the circuit in Figure \ref{fig:BS}).

\subsection{Cryptographic protocols}
We now introduce some definitions for cryptographic protocols. We define quantum states as belonging to the Hilbert space $\h_{ABC}=\h_A\otimes\h_B\otimes\h_C$, where $A$ and $B$ label Alice and Bob's private registers and $C$ a common register used to move qubits from $A$ to $B$ and {vice-versa}. We denote the trace distance between the states $\rho$ and $\rho'$ as $D(\rho,\rho')=\frac{1}{2}\mathrm{Tr}|\rho-\rho'|$. The symbol $\circ$ represents the ``normal ordered product'' {of maps} - as an example, considering the collection of $q$ maps $\{\E^{(p)}\}_{p=1}^q$ acting on a state $\rho$, we have $\circ_{p=1}^q\E^{(p)}(\rho)=\E^{(q)}\E^{(q-1)}..\textrm{ }\E^{(1)}(\rho)$.

We start by defining our notion of a protocol.
\begin{mydef}~\label{def:delegated}
\textsc{\textbf{[Protocol]}} 
We define a $q$-step protocol on input $\rho_{\mathrm{in}}\in\h_{ABC}$  as a series of maps $\{\E_{ABC}^{(p)}\}_{p=1}^q=\{\E^{(p)}_{AC}\otimes\E^{(p)}_{BC}\}_{p=1}^q$ acting on both Alice's and Bob's registers and on the common register, and such that the output is of the form $\rho_{\mathrm{out}}=\circ_{p=1}^q\E_{ABC}^{(p)}(\rho_{\mathrm{in}})$.\end{mydef}

A protocol is thus a sequence of instructions that define the actions that Alice and Bob need to take along the computation. A crucial requirement for any protocol is that whenever Bob follows Alice's instruction, the outcome obtained in Alice's register is correct. A protocol with this property is said to be ``correct''. We formally define correctness as follows:
\begin{mydef}~\label{def:corr}
\textsc{\textbf{[Correctness]}}
Suppose that Alice wants to apply a CPTP-map $\F_A$ to an input {$\rho_{\mathrm{in}}\in\h_{ABC}$}. A $q$-step protocol $\{\E_{ABC}^{(p)}\}_{p=1}^q=\{\E^{(p)}_{AC}\otimes\E^{(p)}_{BC}\}_{p=1}^q$ on input $\rho_{\mathrm{in}}$ is correct if
\begin{equation}
D\bigg(\mathrm{Tr}_{BC}\big[\F_A\big(\rho_{\mathrm{in}}\big)\big],\mathrm{Tr}_{BC}\big[{\rho}_{\mathrm{out}}\big]\bigg)=0\textrm{ ,}
\end{equation}
where ${\rho}_{\mathrm{out}}=\circ_{p=1}^{q}\big(\E^{(p)}_{AC}\otimes{\E}^{(p)}_{BC}\big)\big(\rho_{\textrm{in}}\big)$.
\end{mydef}
In cryptographic protocols, a natural requirement is that no information is leaked to Bob, so that privacy is guaranteed to Alice. Nevertheless, it must be noticed that Bob is not forced to follow Alice's instructions. In fact, he might deviate from them and try to fool Alice

\onecolumngrid
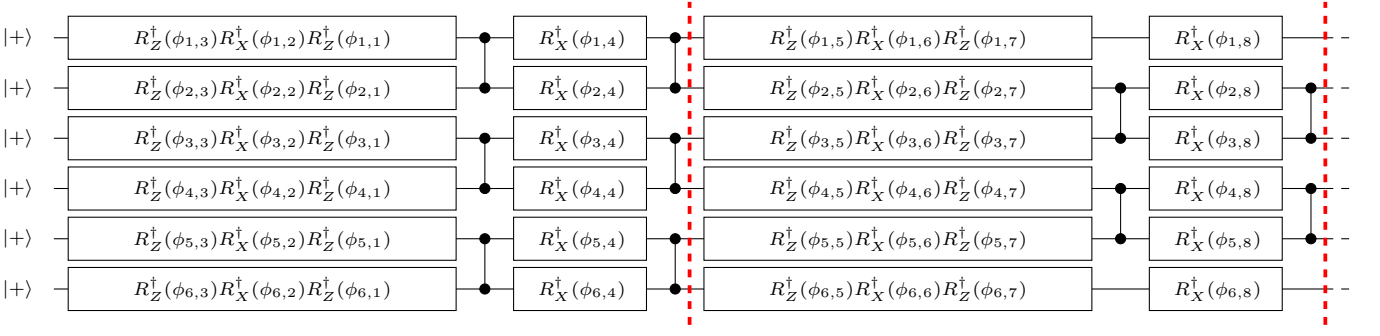
\begin{figure}[H]
\begin{tikzpicture}[scale=0.955, every node/.style={scale=1}]


\node at (-0.5,0.0) {\scriptsize $\ket{+}$};
\node at (-0.5,0.7) {\scriptsize $\ket{+}$};
\node at (-0.5,1.4) {\scriptsize $\ket{+}$};
\node at (-0.5,2.1) {\scriptsize $\ket{+}$};
\node at (-0.5,2.8) {\scriptsize $\ket{+}$};
\node at (-0.5,3.5) {\scriptsize $\ket{+}$};

\draw (0.0,0.0) -- (8.65+8.65+0.2+0.2,0.0);
\draw (0.0,0.7) -- (8.65+8.65+0.2+0.2,0.7);
\draw (0.0,1.4) -- (8.65+8.65+0.2+0.2,1.4);
\draw (0.0,2.1) -- (8.65+8.65+0.2+0.2,2.1);
\draw (0.0,2.8) -- (8.65+8.65+0.2+0.2,2.8);
\draw (0.0,3.5) -- (8.65+8.65+0.2+0.2,3.5);

\draw [dashed] (8.65+8.65+0.2+0.2,0.0) -- (8.65+8.65+0.2+0.6,0.0);
\draw [dashed] (8.65+8.65+0.2+0.2,0.7) -- (8.65+8.65+0.2+0.6,0.7);
\draw [dashed] (8.65+8.65+0.2+0.2,1.4) -- (8.65+8.65+0.2+0.6,1.4);
\draw [dashed] (8.65+8.65+0.2+0.2,2.1) -- (8.65+8.65+0.2+0.6,2.1);
\draw [dashed] (8.65+8.65+0.2+0.2,2.8) -- (8.65+8.65+0.2+0.6,2.8);
\draw [dashed] (8.65+8.65+0.2+0.2,3.5) -- (8.65+8.65+0.2+0.6,3.5);

\draw [fill=white] (0.2,-0.3) rectangle (5.6,0.3);
\draw [fill=white] (0.2,0.4) rectangle (5.6,1.0);
\draw [fill=white] (0.2,1.1) rectangle (5.6,1.7);
\draw [fill=white] (0.2,1.8) rectangle (5.6,2.4);
\draw [fill=white] (0.2,2.5) rectangle (5.6,3.1);
\draw [fill=white] (0.2,3.2) rectangle (5.6,3.8);

\node at (2.9,0.0) {\scriptsize $R^\dagger_Z{(\phi_{6,3})}R^\dagger_X(\phi_{6,2})R^\dagger_Z(\phi_{6,1})$};
\node at (2.9,0.7) {\scriptsize $R^\dagger_Z{(\phi_{5,3})}R^\dagger_X(\phi_{5,2})R^\dagger_Z(\phi_{5,1})$};
\node at (2.9,1.4) {\scriptsize $R^\dagger_Z{(\phi_{4,3})}R^\dagger_X(\phi_{4,2})R^\dagger_Z(\phi_{4,1})$};
\node at (2.9,2.1) {\scriptsize $R^\dagger_Z{(\phi_{3,3})}R^\dagger_X(\phi_{3,2})R^\dagger_Z(\phi_{3,1})$};
\node at (2.9,2.8) {\scriptsize $R^\dagger_Z{(\phi_{2,3})}R^\dagger_X(\phi_{2,2})R^\dagger_Z(\phi_{2,1})$};
\node at (2.9,3.5) {\scriptsize $R^\dagger_Z{(\phi_{1,3})}R^\dagger_X(\phi_{1,2})R^\dagger_Z(\phi_{1,1})$};

\draw [fill=black] (6,0.0) circle [radius=0.07cm];
\draw [fill=black] (6,0.7) circle [radius=0.07cm];
\draw [fill=black] (6,1.4) circle [radius=0.07cm];
\draw [fill=black] (6,2.1) circle [radius=0.07cm];
\draw [fill=black] (6,2.8) circle [radius=0.07cm];
\draw [fill=black] (6,3.5) circle [radius=0.07cm];

\draw (6,0.0) -- (6,0.7);
\draw (6,1.4) -- (6,2.1);
\draw (6,2.8) -- (6,3.5);

\draw [fill=white] (6.4,-0.3) rectangle (8.25,0.3);
\draw [fill=white] (6.4,0.4) rectangle (8.25,1.0);
\draw [fill=white] (6.4,1.1) rectangle (8.25,1.7);
\draw [fill=white] (6.4,1.8) rectangle (8.25,2.4);
\draw [fill=white] (6.4,2.5) rectangle (8.25,3.1);
\draw [fill=white] (6.4,3.2) rectangle (8.25,3.8);

\node at (7.35,0.0) {\scriptsize $R^\dagger_X{(\phi_{6,4})}$};
\node at (7.35,0.7) {\scriptsize $R^\dagger_X{(\phi_{5,4})}$};
\node at (7.35,1.4) {\scriptsize $R^\dagger_X{(\phi_{4,4})}$};
\node at (7.35,2.1) {\scriptsize $R^\dagger_X{(\phi_{3,4})}$};
\node at (7.35,2.8) {\scriptsize $R^\dagger_X{(\phi_{2,4})}$};
\node at (7.35,3.5) {\scriptsize $R^\dagger_X{(\phi_{1,4})}$};

\draw [fill=black] (8.65,0.0) circle [radius=0.07cm];
\draw [fill=black] (8.65,0.7) circle [radius=0.07cm];
\draw [fill=black] (8.65,1.4) circle [radius=0.07cm];
\draw [fill=black] (8.65,2.1) circle [radius=0.07cm];
\draw [fill=black] (8.65,2.8) circle [radius=0.07cm];
\draw [fill=black] (8.65,3.5) circle [radius=0.07cm];

\draw (8.65,0.0) -- (8.65,0.7);
\draw (8.65,1.4) -- (8.65,2.1);
\draw (8.65,2.8) -- (8.65,3.5);

\draw [fill=white] (0.4+8.65,-0.3) rectangle (5.6+8.65+0.2,0.3);

\draw [red, dashed, line width=0.5mm] (8.65+8.65+0.4,-0.5) -- (8.65+8.65+0.4,4.0);
\draw [red, dashed, line width=0.5mm] (8.65+0.2,-0.5) -- (8.65+0.2,4.0);
\draw [fill=white] (0.4+8.65,0.4) rectangle (5.6+8.65+0.2,1.0);
\draw [fill=white] (0.4+8.65,1.1) rectangle (5.6+8.65+0.2,1.7);
\draw [fill=white] (0.4+8.65,1.8) rectangle (5.6+8.65+0.2,2.4);
\draw [fill=white] (0.4+8.65,2.5) rectangle (5.6+8.65+0.2,3.1);
\draw [fill=white] (0.4+8.65,3.2) rectangle (5.6+8.65+0.2,3.8);

\node at (2.9+0.2+8.65,0.0) {\scriptsize $R^\dagger_Z{(\phi_{6,5})}R^\dagger_X(\phi_{6,6})R^\dagger_Z(\phi_{6,7})$};
\node at (2.9+8.65+0.2,0.7) {\scriptsize $R^\dagger_Z{(\phi_{5,5})}R^\dagger_X(\phi_{5,6})R^\dagger_Z(\phi_{5,7})$};
\node at (2.9+8.65+0.2,1.4) {\scriptsize $R^\dagger_Z{(\phi_{4,5})}R^\dagger_X(\phi_{4,6})R^\dagger_Z(\phi_{4,7})$};
\node at (2.9+8.65+0.2,2.1) {\scriptsize $R^\dagger_Z{(\phi_{3,5})}R^\dagger_X(\phi_{3,6})R^\dagger_Z(\phi_{3,7})$};
\node at (2.9+8.65+0.2,2.8) {\scriptsize $R^\dagger_Z{(\phi_{2,5})}R^\dagger_X(\phi_{2,6})R^\dagger_Z(\phi_{2,7})$};
\node at (2.9+8.65+0.2,3.5) {\scriptsize $R^\dagger_Z{(\phi_{1,5})}R^\dagger_X(\phi_{1,6})R^\dagger_Z(\phi_{1,7})$};

\draw [fill=black] (6+8.65+0.2,0.7) circle [radius=0.07cm];
\draw [fill=black] (6+8.65+0.2,1.4) circle [radius=0.07cm];
\draw [fill=black] (6+8.65+0.2,2.1) circle [radius=0.07cm];
\draw [fill=black] (6+8.65+0.2,2.8) circle [radius=0.07cm];

\draw (6+8.65+0.2,0.7) -- (6+8.65+0.2,1.4);
\draw (6+8.65+0.2,2.1) -- (6+8.65+0.2,2.8);

\draw [fill=white] (6+8.65+0.2+0.4,-0.3) rectangle (8.25+8.65+0.2,0.3);
\draw [fill=white] (6+8.65+0.2+0.4,0.4) rectangle (8.25+8.65+0.2,1.0);
\draw [fill=white] (6+8.65+0.2+0.4,1.1) rectangle (8.25+8.65+0.2,1.7);
\draw [fill=white] (6+8.65+0.2+0.4,1.8) rectangle (8.25+8.65+0.2,2.4);
\draw [fill=white] (6+8.65+0.2+0.4,2.5) rectangle (8.25+8.65+0.2,3.1);
\draw [fill=white] (6+8.65+0.2+0.4,3.2) rectangle (8.25+8.65+0.2,3.8);

\node at (7.35+8.65+0.2,0.0) {\scriptsize $R^\dagger_X{(\phi_{6,8})}$};
\node at (7.35+8.65+0.2,0.7) {\scriptsize $R^\dagger_X{(\phi_{5,8})}$};
\node at (7.35+8.65+0.2,1.4) {\scriptsize $R^\dagger_X{(\phi_{4,8})}$};
\node at (7.35+8.65+0.2,2.1) {\scriptsize $R^\dagger_X{(\phi_{3,8})}$};
\node at (7.35+8.65+0.2,2.8) {\scriptsize $R^\dagger_X{(\phi_{2,8})}$};
\node at (7.35+8.65+0.2,3.5) {\scriptsize $R^\dagger_X{(\phi_{1,8})}$};

\draw [fill=black] (8.65+8.65+0.2,0.7) circle [radius=0.07cm];
\draw [fill=black] (8.65+8.65+0.2,1.4) circle [radius=0.07cm];
\draw [fill=black] (8.65+8.65+0.2,2.1) circle [radius=0.07cm];
\draw [fill=black] (8.65+8.65+0.2,2.8) circle [radius=0.07cm];

\draw (8.65+8.65+0.2,0.7) -- (8.65+8.65+0.2,1.4);
\draw (8.65+8.65+0.2,2.1) -- (8.65+8.65+0.2,2.8);

\end{tikzpicture}
\caption{\small Logical circuit associated to a computation on a six-row BwS. Red dashed lines separate operations implemented within different tapes of the BwS.}
\label{fig:BS}
\end{figure}
\twocolumngrid

\noindent (in this case, we say that Bob is ``dishonest''). Thus, protocols must ensure that Bob can not increase his knowledge by cheating, where by cheating we mean that along the protocol run Bob applies some dishonest collection of maps $\{\widetilde{\E}^{(p)}_{BC}\}_{p=1}^q$ instead of the ``honest'' collection $\{{\E}^{(p)}_{BC}\}_{p=1}^q$. This property is called ``blindness'' and is defined as follows:
\begin{mydef}~\label{def:blind}
\textsc{\textbf{[Blindness]}}
Suppose that Alice and Bob jointly run a $q$-step protocol $\{\E_{ABC}^{(p)}\}_{p=1}^q=\{\E^{(p)}_{AC}\otimes\E^{(p)}_{BC}\}_{p=1}^q$ on input $\rho_{\mathrm{in}}\in\h_{ABC}$. The protocol is blind if, for any set of maps $\{\widetilde{\E}^{(p)}_{BC}\}_{p=1}^q$ acting on Bob's register $B$ and on the common register $C$, the state $\mathrm{Tr}_{\mathrm{AC}}[\circ_{p=1}^q\{\E^{(p)}_{AC}\otimes\widetilde\E^{(p)}_{BC}\}(\rho_{\mathrm{in}})]$ leaks at most a constant function of the input.
\end{mydef}

\noindent A typical example of constant function of the input leaked by a protocol is the size of the computation (the number of qubits and gates used). This does not depend on the information that Alice is interested to hide (the state of the input and of the output and the gates used) and can be leaked. Thus, at the end of a dishonest run of a blind protocol, Bob obtains as much information about the computation as after an honest run.

Another important property for a protocol is verifiability, namely the possibility of verifying whether the output of the computation is correct or wrong. In our Protocols, Alice verifies the computation by checking the outcome of various deterministic quantum computations (the ``traps''). If at the end of the computation the traps are found in a specific state (here denoted by $|\textup{acc}\rangle$), then Alice accepts, otherwise she rejects. With this in mind, denoting the input state $\rho_{\textrm{in}}$ as a tensor product between the input state $\rho_{\textup{in}}^{\textup{comp}}$ of the actual computation and the input state $|\textrm{trap}\rangle$ of the traps, we define verifiability as follows  \cite{DFPR14,GKK17}:

\begin{mydef}~\label{def:ver}
{\textsc{\textbf{[Verifiability]}}
Suppose that Alice and Bob jointly run a $q$-step protocol $\{\E_{ABC}^{(p)}\}_{p=1}^q=\{\E^{(p)}_{AC}\otimes\E^{(p)}_{BC}\}_{p=1}^q$ on input $\rho_{\textup{in}}=\rho_{\textup{in}}^{\textup{comp}}\otimes|\textup{trap}\rangle\langle\textup{trap}|\in\h_{ABC}$. The protocol is ``$\delta$-complete'' if
\begin{align*}
D\bigg(\mathrm{Tr}_{BC}\big[{\rho}_{\mathrm{out}}\big]\textup{, }\mathrm{Tr}_{BC}\big[&\textrm{ }{\rho}_{\mathrm{out}}^{\mathrm{comp}}\otimes|\textup{acc}\rangle\langle\textup{acc}|\textrm{ }\big]\bigg)\leq1-\delta\textrm{ ,}
\end{align*}
\noindent where $0\leq\delta\leq1$, ${\rho}_{\mathrm{out}}=\circ_{p=1}^q\E_{ABC}^{(p)}({\rho}_{\mathrm{in}})$, ${\rho}_{\mathrm{out}}^{\mathrm{comp}}=\textup{Tr}_{\textup{trap}}({\rho}_{\mathrm{out}})$ is the honest outcome of the actual computation and $|\textup{acc}\rangle\langle\textup{acc}|=\textup{Tr}_{\textup{comp}}({\rho}_{\mathrm{out}})$ is a fixed state. If $\delta=1$, then we say that the protocol is ``complete''.

The protocol is ``$\varepsilon$-sound'' if, for any set of maps $\{\widetilde{\E}^{(p)}_{BC}\}_{p=1}^q$ acting on Bob's register $B$ and on the common register $C$, the output $\widetilde{\rho}_{\textup{out}}=\circ_p\big(\E^{(p)}_{AC}\otimes\widetilde{\E}^{(p)}_{BC}\big)\big(\rho_{\mathrm{in}}\big)$ is such that}
\begin{align*}
D\bigg(\mathrm{Tr}_{BC}\big[\widetilde{\rho}_{\mathrm{out}}\big]\textup{, }\mathrm{Tr}_{BC}\big[&r\textrm{ }{\rho}_{\mathrm{out}}^{\mathrm{comp}}\otimes|\textup{acc}\rangle\langle\textup{acc}|\textup{ +}\cr
&(1-r)\textrm{ }\widetilde{\rho}_{\mathrm{out}}^{\textrm{ }\mathrm{comp}}\otimes|\textup{rej}\rangle\langle\textup{rej}|\big]\bigg)\leq\varepsilon\textrm{ ,}
\end{align*}
\noindent where $0\leq\varepsilon\leq1$ is called ``soundness'', $0\leq r \leq1$, $\widetilde{\rho}_{\mathrm{out}}^{\textrm{ }\mathrm{comp}}$ is an arbitrary state and $|\textup{rej}\rangle$ is orthogonal to $|\textup{acc}\rangle$.  If the protocol is both $\varepsilon$-sound and $\delta$-complete, then we say that it is ``$(\varepsilon,\delta)$-verifiable''. If a protocol is $(\varepsilon,\delta)$-verifiable with $\delta=1$, we say that the protocol is ``$\varepsilon$-verifiable''.  
\end{mydef}
Thus, a protocol is verifiable if with high probability, independently of Bob's behaviour, either the computation is correct and Alice accepts or the computation is rejected. As we will see, both Protocol \hyperlink{pr:pr1}{1} and \hyperlink{pr:pr2}{2} are complete with $\delta=1$.

\section{Related Works}~\label{sec:comparison}

Here, we provide a comparison between our schemes and the existing protocols. We refer the reader to \cite{GKK17} for a recent detailed review of quantum verification.

Protocol \hyperlink{pr:pr1}{1} (Section \ref{sec:verifSP}) belongs to the class of ``prepare-and-send'' protocols, such as \cite{FK12,B15,KD17,KW17}. In these schemes, Alice prepares single qubits from a finite set of states and sends them to Bob, who blindly performs the rest of the computation.
Compared to our protocol, the other schemes in this class are more expensive in terms of resources while achieving (in the best case) the same overhead. As an example, in Fitzsimons and Kashefi's protocol \cite{FK12} Alice needs to prepare single qubits in the state $|+\rangle_{\theta}$, $\theta\in\{0,\pi/4,..,7\pi/4\}$, as well as in the ``dummy'' states $\ket{0}$ and $\ket{1}$. The overhead of Fitzsimons and Kashefi's protocol is quadratic in the size of the computation, although it was subsequently made linear \cite{KD17,KW17}. Compared to our Protocol \hyperlink{pr:pr1}1, Fitzsimons and Kashefi's protocol requires trusted state preparation of more types  of states (ten instead of eight) while achieving the same overhead. Similarly, in Broadbent's protocol \cite{B15}, Alice needs to be able to generate qubits in the states $\ket{0}$ and $\ket{+}$ and tp apply the gates $X,Z,S$ and $T$. Overall, Broadbent's protocol requires the same amount of resources as the Fitzsimons and Kashefi's one and its overhead is linear in the input. Other protocols in the ``prepare-and-send'' class are that by Aharonov et al. \cite{ABE08}. In terms of resources, Aharonov's schemes are more demanding. Alice holds a multi-qubit register and in one of these protocols, she needs to apply gates from the Clifford group and subsequently make measurements, while in the other one she need to apply to her quantum inputs a sophisticated encoding inspired to a polynomial Calderbank-Shor-Steane quantum error correcting codes \cite{AB99}.

Protocol \hyperlink{pr:pr2}{2} (Section \ref{sec:verifM}) belongs to the class of ``receive-and-measure'' protocols, such as the ``measurement-only'' scheme \cite{HM15} and the ``post-hoc'' verification techniques \cite{MF16,HKSE17}. In the measurement-only protocol, Alice needs to measure the observables $X\textrm{, }Y\textrm{, }Z\textrm{, }(X\pm Y)/\sqrt{2}$. Compared to our Protocol \hyperlink{pr:pr2}{2}, this protocol requires noise-free measurement of more observables (five instead of four). Also, its soundness $\varepsilon$ is $O(v+1)$, while the soundness of the protocol presented in \cite{HM15} goes as $1/\sqrt{v+1}$. However, it has to be mentioned that in Ref. \cite{HM15}, the physical qubits are discarded after being measured, while Protocol \hyperlink{pr:pr2}{2} relies on the assumption that qubits can be reused after the measurement has been done. We leave as an open question the possibility of adapting Protocol \hyperlink{pr:pr2}{2} to the more general scenario where the qubits can not be reused after being measured.

Differently from the protocols mentioned so far, the post-hoc protocols are schemes with a single round of communication between the verifier and the prover where verification is performed after the computation has been carried out. In these protocols, Alice is solely required to make measurements in the Pauli-$Z$ and Pauli-$X$ bases. However, post-hoc protocols are not blind\cite{GKK17}. Also, their overhead is quadratic in the input.

Finally, a third class of protocols is the class of ``entanglement-based'' schemes. In these works, the computation is carried out on two entangled and spatially-separated servers (the provers). The verifier is ``classical'', in the sense that does not require any resource. Both provers are untrusted, but their spatial separation prevents the provers from communicating with each other and agreeing on a specific cheating strategy after the protocol is started. The computation is verified by means of CHSH games \cite{RUV12,GKW15,HPF15}, self-testing techniques \cite{M16,HH16} or post-hoc protocols \cite{NV16}. {The overhead of entanglement-based protocols is higher than that of the above mentioned schemes \cite{GKK17}}, although the communication is only classical. Notice the different perspective of the previous classes of protocols and the entanglement-based one: in the first case, the verifier is convinced that she can trust some quantum device, namely a state generator or a measurement device; in the second case, instead, the verifier does not trust any device, but presumes that the provers do not communicate.

\section{Verification for trusted state preparation}~\label{sec:verifSP}

In this section, we present a protocol to verify the correctness of the outcome of a universal computation (the ``target'') run on a $n\times m$ BwS, modulo a trust assumption on the preparation of qubits in the set of states $\ket{+}_\theta=(\ket{0}+e^{i\theta}\ket{1})/\sqrt{2}$, $\theta\in\{0,\pi/4,..,7\pi/4\}$. We first describe the main ideas behind the Protocol (Subsection \ref{subsec:desc1}) and analyse its overhead (Subsection \ref{subsec:comp1}). Finally, we show that the protocol is correct, blind and verifiable (Subsection \ref{subsec:proo1}).

\subsection{Description of the Protocol}~\label{subsec:desc1}
In our protocol, the correctness of the target computation is verified by checking the outcome of several other computations that can be efficiently simulated classically, the ``traps''. The trap computations are chosen so that their outcome is deterministic and can thus be used as {witnesses}: if all of their outcomes correspond to the expected ones, one assumes that they have not been affected by errors, and concludes that the target computation itself has been carried out correctly. Otherwise, there is no guarantee that the target computation is correct, hence it has to be rejected. 

The specific computation implemented within each trap is chosen at random from two classes, respectively denoted as ``R-traps'' (or ``rotation traps'') and ``C-traps'' (or ``CNOT traps''). The physical qubits of the BwS implementing R-traps are assigned measurement angles as described in Sub-protocol \hyperlink{pr:spr1}{1.1}. In particular, for any tape $y\in(1,..,w)$ (where by ``tape'' we mean vertical layers of the BwS composed of four columns of qubits, see Figure \ref{fig:BSmain}) and for any row $i\in(1,..,n)$, a coin is flipped. As illustrated in Figure \ref{fig:rtraps}, if the coin outputs 0 (respectively 1), the physical qubits belonging to row $i$ and tape $y$ are measured so that the logical qubit corresponding to the $i$th row undergoes a Hadamard (respectively a rotation). Logical rotations are performed either on the $XY$-plane or on the $ZY$-plane of the Bloch sphere. Sub-protocol \hyperlink{pr:spr1}{1.1} combines Hadamards and rotations so that any logi-

\newpage
\begin{small}
\noindent\makebox[\linewidth]{\rule{8.8cm}{0.4pt}}
\textbf{Sub-protocol \hypertarget{pr:spr1}{1.1}} R-trap.\\
\noindent\makebox[\linewidth]{\rule{8.8cm}{0.4pt}}
\textbf{Input}: The size of the computation: $n\times m$.\vspace{0.4cm} 

\noindent\textbf{0. Preliminary operations.}\\
Define the $n \times m$ $\{\phi_{i,j}\}$, where $\phi_{i,j}=0\textrm{ }\forall\textrm{ }i,j$.\\

\noindent\textbf{1. Assigning measurement angles.}\\
For every row $i=1,..,n$, set counter $count=0$. Next, 
\begin{itemize} 
			\item[1.1] For each tape $y=1,..,w-1$: \\flip a coin and obtain outcome $c_y=\{0,1\}$. 
	\begin{itemize}

	\item If $c_y=0$, set $count=count\oplus1$ and
	\begin{eqnarray*}
		\begin{tabular}{lllllllllllllll}
			$\phi_{i,4y-3}$=$\pi/2$&,&
			$\phi_{i,4y-2}$=$\pi/2$&,&
			$\phi_{i,4y-1}$=$\pi/2$
		\end{tabular}
	\end{eqnarray*} 
	
	\item If $c_y=1$ and $count=0$, set
	\begin{eqnarray*}
		\begin{tabular}{llllllllll}
			$\phi_{i,4y-3}$=$k_{i,4y-3}\pi/4$&,&
			$\phi_{i,4y-1}$=$k_{i,4y-1}\pi/4$
		\end{tabular}
	\end{eqnarray*}
	where any $k_{i,j}$ is chosen at random in $\{0,1,..,7\}$. Next, set $$\phi_{i,m}=mod(\phi_{i,m}+\phi_{i,4y-3}+\phi_{i,4y-1},2\pi)$$

\item If $c_y=1$ and $count=1$, set
\begin{eqnarray*}
\begin{tabular}{llllllllll}
$\phi_{i,4y-2}$=$k_{i,4y-2}\pi/4$
\end{tabular}
\end{eqnarray*}
where any $k_{i,j}$ is chosen at random in $\{0,1,..,7\}$. Next, set $$\phi_{i,m}=mod(\phi_{i,m}+\phi_{i,4y-2},2\pi)$$

\end{itemize}

\item[1.2] For vertical tape $y=w$:
\begin{itemize}
	\item if $count=0$, set
	\begin{eqnarray*}
		\begin{tabular}{llllllllll}
			$\phi_{i,4y-3}$=$k_{i,4y-3}\pi/4$&,&
			$\phi_{i,4y-1}$=$k_{i,4y-1}\pi/4$
		\end{tabular}
	\end{eqnarray*}
	where any $k_{i,j}$ is chosen at random in $\{0,1,..,7\}$. Next, set $$\phi_{i,m}=mod(\phi_{i,m}+\phi_{i,4y-3},+\phi_{i,4y-1},2\pi)$$
	
	\item if $count=1$, set
	\begin{eqnarray*}
		\begin{tabular}{lllllllllllllll}
			$\phi_{i,4y-3}$=$\pi/2$&,&
			$\phi_{i,4y-2}$=$\pi/2$&,&
			$\phi_{i,4y-1}$=$\pi/2$
		\end{tabular}
	\end{eqnarray*}
\end{itemize}

\end{itemize} 

\noindent\textbf{Output:} The set of measurement angles $\{\phi_{i,j}\}$.\\
\noindent\makebox[\linewidth]{\rule{8.8cm}{0.4pt}}
\\
\end{small}

\vspace{0.3cm}

\begin{figure}[H]
\centering
\begin{tikzpicture}
\node at (0.4,4.1) {\small Coin outputs 0};
\node at (0.4,2.55) {\small Coin outputs 1};

\node at (7.5,4.1) {$H$};
\node at (7.5,3) {$R_Z$};
\node at (7.5,2.45) {\small or};
\node at (7.5,1.9) {$R_X$};

\draw [dashed, red, line width=0.5mm] (6,4.4) -- (6,1.5);
\draw [dashed] (7,4.4) -- (7,1.5);

\draw (2.0,4.1) -- (6.5,4.1);
\draw [fill=white] (2.5,4.1) circle [radius=2mm];
\draw [fill=white] (3.5,4.1) circle [radius=2mm];
\draw [fill=white] (4.5,4.1) circle [radius=2mm];
\draw [fill=white] (5.5,4.1) circle [radius=2mm];
\draw [fill=white] (6.5,4.1) circle [radius=2mm];
\node at (2.7,4.45) {\normalsize $\frac{\pi}{2}$};
\node at (3.7,4.45) {\normalsize $\frac{\pi}{2}$};
\node at (4.7,4.45) {\normalsize $\frac{\pi}{2}$};
\node at (5.7,4.45) {\small $0$};

\draw (2.0,3) -- (6.5,3);
\draw [fill=white] (2.5,3) circle [radius=2mm];
\draw [fill=white] (3.5,3) circle [radius=2mm];
\draw [fill=white] (4.5,3) circle [radius=2mm];
\draw [fill=white] (5.5,3) circle [radius=2mm];
\draw [fill=white] (6.5,3) circle [radius=2mm];
\node at (2.7,3.35) { $\phi_{i,4y-3}$};
\node at (3.7,3.35) { $0$};
\node at (4.7,3.35) { $\phi_{i,4y-1}$};
\node at (5.7,3.35) { $0$};

\draw (2.0,1.9) -- (6.5,1.9);
\draw [fill=white] (2.5,1.9) circle [radius=2mm];
\draw [fill=white] (3.5,1.9) circle [radius=2mm];
\draw [fill=white] (4.5,1.9) circle [radius=2mm];
\draw [fill=white] (5.5,1.9) circle [radius=2mm];
\draw [fill=white] (6.5,1.9) circle [radius=2mm];
\node at (2.7,2.25) { $0$};
\node at (3.7,2.25) {$\phi_{i,4y-2}$};
\node at (4.7,2.25) { $0$};
\node at (5.7,2.25) { $0$};
\draw [decorate,decoration={brace,amplitude=10pt},xshift=-0pt,yshift=0pt] (2.0,1.7) -- (2.0,3.4) node [black,midway,xshift=-0.6cm] {};

\end{tikzpicture}
\caption{\small Sub-protocol 1.1. For any tape $y$ and for any row $i$, a coin is flipped. If it outputs 0 - respectively 1 -, qubits $(i,4y-3)$, $(i,4y-2)$ $(i,4y-1)$ and $(i,4y)$ are measured so that logical qubit $i$ undergoes a Hadamard (circuit on the top) - respectively a rotation (circuits on the bottom; angles $\phi_{i,j}\in\{0,\pi/4,..,7\pi/4\}$ are chosen at random). Hadamards and rotations are combined so that overall, qubit $i$ undergoes a $R_Z$-gate by random angle from the set $\{0,\pi/4,..,7\pi/4\}$.}
\label{fig:rtraps}
\end{figure}
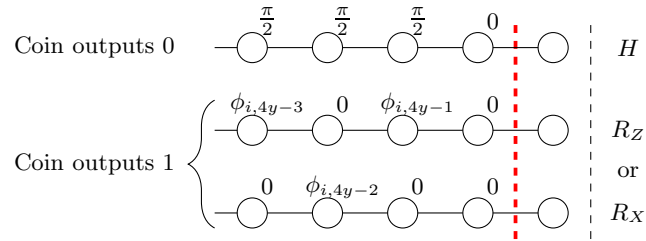

\newpage
\begin{small}

\noindent\makebox[\linewidth]{\rule{8.8cm}{0.4pt}}
\textbf{Sub-protocol \hypertarget{pr:spr2}{1.2}} C-trap.\\
\noindent\makebox[\linewidth]{\rule{8.8cm}{0.4pt}}
\textbf{Input}: Size of the computation: $n\times m$.\vspace{0.4cm}

\noindent\textbf{0. Preliminary operations.}\\
Define the $n \times m$ $\{\phi_{i,j}\}$, where $\phi_{i,j}=0\textrm{ }\forall\textrm{ }i,j$.\\

\noindent\textbf{1. Assigning measurement angles.}\\
For each row $i=1,..,n$:
\begin{itemize}
\item[1.1] Flip a coin and obtain $c_i\in\{0,1\}$. If $c_i=0$, set $\phi_{i,1}=\pi/2$, otherwise do nothing. 
\item[1.2] For each tape $y=1,..,w$:\\
If $mod(i+y,2)=0$, flip a coin and obtain $c_{i,y}\in\{0,1\}$. Then,
\begin{itemize}
\item[$-$]If $c_{i,y}=0$, set
\begin{eqnarray*}
\begin{tabular}{ccccccccc}
$\phi_{2i-1,4y-1}$=$\pi/2$&,&
$\phi_{2i,4y-2}$=$\pi/2$&,&
$\phi_{2i,4y}$=$-\pi/2$\cr
\end{tabular}
\end{eqnarray*}
and $\phi_{i,m}=mod(\phi_{i,m}+\phi_{i+1,m},2\pi)$.
\item[$-$]If $c_{i,y}=1$, set
\begin{eqnarray*}
\begin{tabular}{llrlllllll}
$\phi_{2i+1,4y-1}$=$\pi/2$&,&
$\phi_{2i-1,4y-2}$=$\pi/2$&,&
$\phi_{2i-1,4y}$=$-\pi/2$
\end{tabular}
\end{eqnarray*}
and $\phi_{i+1,m}=mod(\phi_{i+1,m}+\phi_{i,m},2\pi)$.
\end{itemize}

\end{itemize}

\noindent\textbf{Output:} The set of measurement angles $\{\phi_{i,j}\}$.\\
\noindent\makebox[\linewidth]{\rule{8.8cm}{0.4pt}}

\end{small}

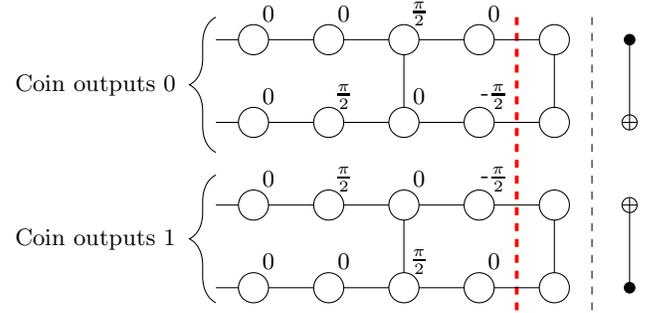
\begin{figure}[H]
	\centering
	\begin{tikzpicture}
	\node at (0.4,4.6) {\small Coin outputs 0};
	\node at (0.4,2.55) {\small Coin outputs 1};
	
	\draw (7.5,5.2) -- (7.5,4.0);
	\draw (7.4,4.1) -- (7.6,4.1);
	\draw [] (7.5,4.1) circle [radius=0.1cm];
	\draw [fill=black] (7.5,5.2) circle [radius=0.07cm];
	
	\draw (7.5,1.9) -- (7.5,3.1);
	\draw (7.4,3) -- (7.6,3);
	\draw [] (7.5,3) circle [radius=0.1cm];
	\draw [fill=black] (7.5,1.9) circle [radius=0.07cm];
	
	\draw [dashed, red, line width=0.5mm] (6,5.5) -- (6,1.6);
	\draw [dashed] (7,5.5) -- (7,1.6);
	
	\draw (2.0,5.2) -- (6.5,5.2);
	\draw (4.5,5.2) -- (4.5,4.1);
	\draw (6.5,5.2) -- (6.5,4.1);
    \draw [fill=white] (2.5,5.2) circle [radius=2mm];
	\draw [fill=white] (3.5,5.2) circle [radius=2mm];
	\draw [fill=white] (4.5,5.2) circle [radius=2mm];
	\draw [fill=white] (5.5,5.2) circle [radius=2mm];
	\draw [fill=white] (6.5,5.2) circle [radius=2mm];
	\node at (2.7,5.55) {\small $0$};
	\node at (3.7,5.55) {\small $0$};
	\node at (4.7,5.55) {\small $\frac{\pi}{2}$};
	\node at (5.7,5.55) {\small $0$};

	\draw (2.0,4.1) -- (6.5,4.1);
	\draw [fill=white] (2.5,4.1) circle [radius=2mm];
	\draw [fill=white] (3.5,4.1) circle [radius=2mm];
	\draw [fill=white] (4.5,4.1) circle [radius=2mm];
	\draw [fill=white] (5.5,4.1) circle [radius=2mm];
	\draw [fill=white] (6.5,4.1) circle [radius=2mm];
	\node at (2.7,4.45) {\small $0$};
	\node at (3.7,4.45) {\small $\frac{\pi}{2}$};
	\node at (4.7,4.45) {\small $0$};
	\node at (5.7,4.45) {\small -$\frac{\pi}{2}$};
	
	\draw (2.0,3) -- (6.5,3);
	\draw (4.5,3) -- (4.5,1.9);
	\draw (6.5,3) -- (6.5,1.9);
	\draw [fill=white] (2.5,3) circle [radius=2mm];
	\draw [fill=white] (3.5,3) circle [radius=2mm];
	\draw [fill=white] (4.5,3) circle [radius=2mm];
	\draw [fill=white] (5.5,3) circle [radius=2mm];
	\draw [fill=white] (6.5,3) circle [radius=2mm];
	\node at (2.7,3.35) {\small $0$};
	\node at (3.7,3.35) {\small $\frac{\pi}{2}$};
	\node at (4.7,3.35) {\small $0$};
	\node at (5.7,3.35) {\small -$\frac{\pi}{2}$};
	
	\draw (2.0,1.9) -- (6.5,1.9);
	\draw [fill=white] (2.5,1.9) circle [radius=2mm];
	\draw [fill=white] (3.5,1.9) circle [radius=2mm];
	\draw [fill=white] (4.5,1.9) circle [radius=2mm];
	\draw [fill=white] (5.5,1.9) circle [radius=2mm];
	\draw [fill=white] (6.5,1.9) circle [radius=2mm];
	\node at (2.7,2.25) {\small $0$};
	\node at (3.7,2.25) {\small $0$};
	\node at (4.7,2.25) {\small $\frac{\pi}{2}$};
	\node at (5.7,2.25) {\small $0$};
	\draw [decorate,decoration={brace,amplitude=10pt},xshift=-0pt,yshift=0pt] (2.0,1.7) -- (2.0,3.4) node [black,midway,xshift=-0.6cm] {};
	\draw [decorate,decoration={brace,amplitude=10pt},xshift=-0pt,yshift=0pt] (2.0,3.7) -- (2.0,5.5) node [black,midway,xshift=-0.6cm] {};
	
	\end{tikzpicture}
	\caption{\small Sub-protocol 1.2. For any 10-qubit brick involving rows $i$ and $i+1$, a coin is flipped. If it outputs 0 (respectively 1), physical qubits composing the brick are measured so that logical qubits $i$ and $i+1$ undergo a CNOT, qubit $i$ being the control and qubit $i+1$ being the target (respectively qubit $i$ being the target and qubit $i+1$ being the control).}
	\label{fig:ctraps}
\end{figure}
\vspace{0.3cm}

\noindent cal qubit $i$ is subject to an overall $XY$-plane rotation $R_Z(\Phi_i)$, where $\Phi_i$ is a random angle belonging to the set $\{0,\pi/4,..,7\pi/4\}$. (In the rest of the paper, the angles are labelled with the upper-case letter $\Phi$ if they refer to operations at the logical level, and with the lower-case letters $\phi$ and $\theta$ if they refer to operations at the physical level). The angles $\Phi_i$ can be easily computed classically (we show this later when we provide the proof of Theorem \ref{theorem:corr1}) and are assigned to qubits in the last column of the BwS. 

The C-traps are assigned measurement angles by Sub-protocol \hyperlink{pr:spr2}{1.2} in such a way that any 10-qubit brick in the logical circuit is used to implement a CNOT (Figure \ref{fig:ctraps}). For any CNOT, the target and control qubits are chosen at random between the two logical qubits that correspond to this brick. Also, a randomly chosen subset of the first-column physical qubits are assigned angle $\pi$. This corresponds to a Pauli-$Z$ gate acting on the corresponding logical qubits at the beginning of the logical circuit.

The traps are ``sensitive'' to Bob's deviations: if Bob does not follow Alice's instructions, it is likely that some traps will output an outcome that does not match with the expected one. To show this, we now give a more formal description of our protocol in the language of cryptographic protocols.

Protocol \hyperlink{pr:pr1}{1} defines the roles of Alice and Bob in the interactive game, under the assumption that Alice can prepare qubits in the discrete set of states $\ket{+}_{\theta}$, $\theta\in\{0,\pi/4,..,7\pi/4\}$. {We assume that Alice wants to perform some quantum computation on a $n\times m$ BwS defined by the set of measurement angles $\{\phi_{i,j}\}$. As a first step, Alice decides the number $v$ of traps that she wants to use. She also decides at random which graph $v_t\in(1,..,v+1)$ will be used to implement the target computation. For any other graph $k\neq v_t$, she randomly chooses whether it will be used to implement a R-trap or a C-trap and subsequently runs the corresponding Sub-protocol to obtain a valid set of measurement angles $\{\phi_{i,j}^{(k)}\}$ for the BwS.} Also, Alice defines two sets of variables $\{r_{i,j}^{(k)}\}$ and $\{r_{i,j}^{\prime(k)}\}$ and a set of angles $\{\theta_{i,j}^{(k)}\}$ for any computation $k=1,..,v+1$, where any $r_{i,j}^{(k)}$ and ${r_{i,j}^{\prime(k)}}$ is chosen at random in $\{0,1\}$ and any angle $\theta^{(k)}_{i,j}$ is chosen at random in $\{0,\pi/4,..,7\pi/4\}$. Next, for any computation $k\in(1,..,v+1)$, Alice and Bob interact as follows:

\begin{itemize}
\item[]\textbf{State preparation}:  Alice sends Bob $nm$ qubits in the state $R_Z(\theta^{(k)}_{i,j}+\pi\sum_{(i',j')\sim(i,j)}^{(k)}r_{i^{\prime},j^\prime}^{\prime(k)})\ket{+}$, where the summation runs over all qubits neighbouring with qubit $(i,j)$. Bob stores the qubits in his register and creates the BwS by entangling them with $cZ$ gates.
\item[]\textbf{Blind computation}: Alice asks Bob to measure each qubit by angle $\delta^{(k)}_{i,j}=(-1)^{r^{\prime(k)}_{i,j}}\phi^{(k)}_{i,j}+\theta^{(k)}_{i,j}+r^{(k)}_{i,j}\pi$. Bob measures qubit $(i,j)$ and reveals the outcome $s_{i,j}^{(k)}$ to Alice. The measurements are performed adaptively: any angle $\delta^{(k)}_{i,j}$ is modified on the basis of measurement outcomes and of the parameters $r^{(k)}_{i,j}$ of previous qubits, so that the overall computation does not depend on the random variables.
\item[]\textbf{Verification}: If the $k$th computation is a trap, Alice checks the outcomes of the measurements of the last-column qubits. If they are all 0, she proceeds with the computation $k+1$, otherwise she rejects the whole run.
\end{itemize} 

\begin{small}
	\newpage
	
	\noindent\makebox[\linewidth]{\rule{8.8cm}{0.4pt}}
	\textbf{Protocol \hypertarget{pr:pr1}{1}.}\\
	\noindent\makebox[\linewidth]{\rule{8.8cm}{0.4pt}}
	\textbf{Hypothesis:} 
	\begin{itemize}
	\item[] Alice can prepare single qubits in the state $\ket{+}_{\theta^{(k)}_{i,j}}$, where $\theta^{(k)}_{i,j}\in\{0,\pi/4,..,7\pi/4\}$.
	\end{itemize}
	\textbf{Input}: 
	\begin{itemize}
		\item[(i) ] the number of computations $v$.
		\item[(ii)] the set of measurement angles $\{\phi_{i,j}\}$ for the target computation.
		\item[(iii)] the sets of random variables $\{r_{i,j}^{(k)}=0,1\}$, $\{r_{i,j}^{\prime(k)}=0,1\}$ and the set of random angles $\{\theta_{i,j}^{(k)}=0,\pi/4,..,7\pi/4\}$ for any computation $k=1,..,v+1$.
	\end{itemize}
	\noindent\textbf{0. Preliminary operation.}\\
	Alice randomly chooses $v_t\in(1,2,\ldots,v+1)$ and sets $\{\phi^{(v_t)}_{i,j}\}=\{\phi_{i,j}\}$.\\
	
	\noindent For $k=1,..,v+1$:
	\begin{itemize}
		\item[]\textbf{1. Assigning measurement angles.}\\
		{If $k\neq v_t$}, Alice randomly runs Sub-protocol \hyperlink{pr:spr2}{1.1} or Sub-protocol \hyperlink{pr:spr2}{1.2} on input $n\times m$ and obtains the set $\{\phi^{(k)}_{i,j}\}$.
		
		\item[]\textbf{2. State preparation.}		
		{For $i=1,..,n$ and for $j=1,..,m$, Alice sends to Bob a qubit in the state $R_Z(\theta^{(k)}_{i,j}+\pi\sum_{(i',j')\sim(i,j)}^{(k)}r_{i',j'}^{\prime(k)})\ket{+}$, where the summation runs over all qubits $(i',j')$ that are nearest neighbours of qubit $(i,j)$ in the $k$th computation.}
		
		\item[]\textbf{3. Blind Computation.}
		\begin{itemize}
			\item[3.1] Bob entangles the qubits in his register and creates the BwS.
			\item[3.2] For $j=1,..,m$ and for $i=1,..,n$,
			\begin{itemize}
				\item[-] Alice computes the angle $\delta^{(k)}_{i,j}=(-1)^{r^{\prime(k)}_{i,j}}\phi^{(k)}_{i,j}+\theta^{(k)}_{i,j}+r^{(k)}_{i,j}\pi$ and reveals it to Bob.
				
				\item[-] Bob measures qubit $(i,j)$ in the basis $|\pm\rangle_{\delta^{(k)}_{i,j}}\langle\pm|$ and reveals the outcome $s^{(k)}_{i,j}$ to Alice.
				\item[-]Alice recomputes the measurement outcome $s^{(k)}_{i,j}$ as $s^{(k)}_{i,j}\oplus r^{(k)}_{i,j}$. Next, she recomputes measurement angles of yet-to-be-measured qubits as $\{(-1)^{s_X}\phi^{(k)}_{i,j}+s_Z\pi\}$.
			\end{itemize}
		\end{itemize}
		
		\item[]\textbf{4. Verification.}\\
		If $k\neq v_t$ and $\{s^{(k)}_{i,m}\}_{i=1}^n\neq(0,0,..,0)$, Alice rejects the whole computation.
		
	\end{itemize}

	\noindent\textbf{Output:} Outcomes $\{s^{(v_t)}_{i,m}\}$ of measurements of last-column qubits of the target computation.\\
	\noindent\makebox[\linewidth]{\rule{8.8cm}{0.4pt}}\\
\end{small}
\newpage
\subsection{Overheads}~\label{subsec:comp1}
The amount of classical bits and qubits sent by Alice to Bob is $N^{\textup{A}}_{\textup{bit}}=3(v+1)nm$ and $N_{\textup{qubit}}^{\textup{A}}=(v+1)nm$ respectively. Notice that the qubits can be sent to Bob through an ``off-line'' interaction {(i.e. in a single round and before the input is given)}, while the exchange of bits is done via an ``on-line'' interaction (i.e. through separate uses of the channel and after the computation is started; recall that Alice needs to adapt the angles after every measurement). On the other hand, the number of qubits sent by Bob to Alice is $N^{\textup{B}}_{\textup{qubit}}=0$, while the amount of bits is equal to $N^{\textup{B}}_{\textup{bit}}=(v+1)nm$. Bob sends to Alice the classical bits of information through an on-line interaction. Therefore we conclude that the overhead is linear on the size of the target computation.

\subsection{Correctness, Blindness and Verifiability}~\label{subsec:proo1}
\noindent We now show the results regarding Protocol \hyperlink{pr:pr1}{1}. 

\begin{theorem}~\label{theorem:corr1}
$\textup{[\textbf{Correctness]}}$ Protocol \hyperlink{pr:pr1}{1} is correct.
\end{theorem}

\begin{proof}
The correctness of Protocol \hyperlink{pr:pr1}{1} can be proven using similar arguments as those in \citep{BFK09}. Since $cZ$ commutes with $R_Z$-rotations, the rotation characterizing any given physical qubit $(i,j)$ cancels at the same time the rotation by angle $\theta_{i,j}$ in $\delta_{i,j}$ and the negative sign in front of angles $\phi_{i,j}^{(k)}$. Second, the effect of the random variables $r_{i,j}$ is cancelled out by Alice when she recomputes the measurement outcome $s_{i,j}^{(k)}$ as $s_{i,j}^{(k)}\oplus r_{i,j}^{(k)}$ (step 3.2 of Protocol \hyperlink{pr:pr1}{1}). Thus, any computation $k=1,..,v+1$ correctly reproduces a computation on a BwS initially in the state $\ket{+}^{\otimes nm}$ and determined by the set of measurement angles $\{\phi_{i,j}^{(k)}\}$. We refer to Appendix \ref{app:corr} for more details. 
\end{proof}

\begin{theorem}
\textup{\textbf{[Blindness]}}~\label{th:bl1} Protocol \hyperlink{pr:pr1}{1} is blind.
\end{theorem}
\begin{proof}
The theorem can be proven with similar arguments as those in \citep{BFK09} (see also \cite{D12} for a more formal proof). We refer to Appendix \ref{app:oldproofs} for a proof of the theorem.
\end{proof}

\begin{theorem}~\label{theorem:ver1}
\textup{\textbf{[Verifiability] }}For any $v\geq7$, Protocol \hyperlink{pr:pr1}{1} is $\varepsilon$-verifiable with soundness
\begin{equation}~\label{eq:oneround}
\varepsilon=\frac{7}{v+1}\bigg(\frac{7}{8}\bigg)^{6}\cong\frac{3.14}{v+1}\textrm{ ,}
\end{equation}
where $v$ represents the number of trap computations.
\end{theorem}

The proof of Theorem \ref{theorem:ver1} relies on the following Lemmas:

\begin{lemma}~\label{lem:tw}
	 Let $\rho$ be a $2^N\times2^N$ density matrix and let $P,P'$ be two n-fold tensor products of the set of operators $\{\mathbb{1},Z,X,Y\}$. Denoting with $\{Q_{r}\}$ the set of all n-fold
\end{lemma}

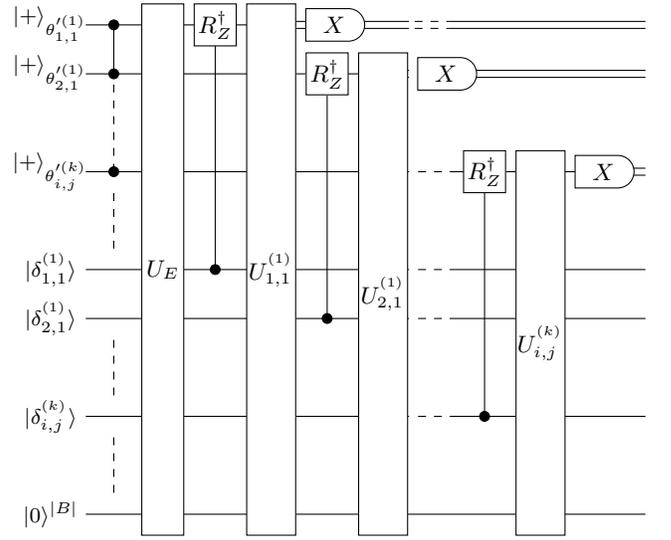
\begin{figure}[H]
	\begin{tikzpicture}[scale=0.93, every node/.style={scale=1}]
	
	\draw (0,4.2) -- (2.7,4.2);
	\draw (2.7,4.25) -- (4.62,4.25);
	\draw (2.7,4.15) -- (4.62,4.15);
	\draw [dashed] (4.7,4.25) -- (5.2,4.25);
	\draw [dashed] (4.7,4.15) -- (5.2,4.15);
	\draw (5.2,4.25) -- (8,4.25);
	\draw (5.2,4.15) -- (8,4.15);
	
	\draw (0,3.5) -- (4.3,3.5);
	\draw (4.3,3.55) -- (4.62,3.55);
	\draw (4.3,3.45) -- (4.62,3.45);
	\draw [dashed] (4.75,3.55) -- (5.2,3.55);
	\draw [dashed] (4.75,3.45) -- (5.2,3.45);
	\draw (5.2,3.55) -- (8,3.55);
	\draw (5.2,3.45) -- (8,3.45);
	
	\draw (0,2.1) -- (4.62,2.1);
	\draw [dashed] (4.75,2.1) -- (5.2,2.1);
	\draw (5.2,2.1) -- (7.8,2.1);
	\draw (7.8,2.15) -- (8,2.15);
	\draw (7.8,2.05) -- (8,2.05);
	
	\draw (0.,0.7) -- (4.62,0.7);
	\draw [dashed] (4.75,1.4-0.7) -- (5.2,1.4-0.7);
	\draw (5.2,1.4-0.7) -- (8,1.4-0.7);
	
	\draw (0,0.7-0.7) -- (4.62,0.7-0.7);
	\draw [dashed] (4.75,0.7-0.7) -- (5.2,0.7-0.7);
	\draw (5.2,0.7-0.7) -- (8,0.7-0.7);
	
	\draw (0,-0.7-0.7) -- (4.62,-0.7-0.7);
	\draw [dashed] (4.75,-0.7-0.7) -- (5.2,-0.7-0.7);
	\draw (5.2,-0.7-0.7) -- (8,-0.7-0.7);
	
	\draw (0.4,3.5) -- (0.4,4.2);
	\draw [dashed] (0.4,2.1) -- (0.4,3.5);
	\draw [fill] (0.4,4.2) circle [radius=0.7mm];
	\draw [fill] (0.4,3.5) circle [radius=0.7mm];
	\draw [fill] (0.4,2.1) circle [radius=0.7mm];
	
	\node at (-0.5,4.2) {\footnotesize $\ket{+}_{\theta^{^{\prime(1)}}_{1,1}}$};
	\node at (-0.5,3.5) {\footnotesize $\ket{+}_{\theta^{^{\prime(1)}}_{2,1}}$};
	\node at (-0.5,2.1) {\footnotesize $\ket{+}_{\theta^{^{\prime(k)}}_{i,j}}$};
	\node at (-0.5,0.7) {\footnotesize $\ket{\delta^{^{(1)}}_{1,1}}$};
	\node at (-0.5,0.0) {\footnotesize $\ket{\delta^{^{(1)}}_{2,1}}$};
	\node at (-0.5,-1.4) {\footnotesize $\ket{\delta^{^{(k)}}_{i,j}}$};
	
	\draw (0.0,-2.8) -- (8,-2.8);
	\node at (-0.5,-2.8) {\footnotesize $\ket{0}^{|B|}$};
	
	\draw (1.1+0.75,4.2) -- (1.1+0.75,0.7);
	\draw [fill] (1.1+0.75,0.7) circle [radius=0.7mm];
	\draw (2.7+0.75,3.5) -- (2.7+0.75,0.0);
	\draw [fill] (2.7+0.75,0.0) circle [radius=0.7mm];
	\draw (5.7+0.,2.1) -- (5.7+0.,-1.4);
	\draw [fill] (5.7+0.,-1.4) circle [radius=0.7mm];
	
	\draw [fill=white] (0.8+0.75,3.9) rectangle (1.4+0.75,4.5);
	\node at (1.1+0.75,4.2) {\small $R^{\dagger}_{Z}$};
	\draw [fill=white] (2.4+0.75,3.2) rectangle (3.0+0.75,3.8);
	\node at (2.7+0.75,3.5) {\small $R^{\dagger}_{Z}$};
	\draw [fill=white] (5.4+0.,1.8) rectangle (6.0+0.,2.4);
	\node at (5.7+0.,2.1) {\small $R^{\dagger}_Z$};
	
	\draw [dashed] (0.4,1.0) -- (0.4,2.5-0.7);
	\draw [dashed] (0.4,-1.1) -- (0.4,-0.3);
	\draw [dashed] (0.4,-1.7) -- (0.4,-2.5);
	
	\draw [fill=white] (0.8,4.5) rectangle (1.4,-3.1);
	\node at (1.1,-3.1+3.8) {\small $U_E$};
	\draw [fill=white] (1.55+0.75,4.5) rectangle (2.25+0.75,-3.1);
	\node at (1.9+0.75,-3.1+3.8) {\small $U^{^{(1)}}_{1,1}$};
	\draw [fill=white] (3.15+0.75,3.8) rectangle (3.85+0.75,-3.1);
	\node at (3.5+0.75,-3.1+3.8-0.35) {\small $U^{^{(1)}}_{2,1}$};
	\draw [fill=white] (6.15,2.4) rectangle (6.85,-3.1);
	\node at (6.5,-3.1+3.8-0.35-0.7) {\small $U^{^{(k)}}_{i,j}$};
	
	\node [fill=white,draw,rounded rectangle,rounded rectangle left arc=none,minimum width=1cm] (name) at (1.55+0.75+1.25,4.2) {$X$};
	\node [fill=white,draw,rounded rectangle,rounded rectangle left arc=none,minimum width=1cm] (name) at (3.15+0.75+1.25,3.5) {$X$};
	\node [fill=white,draw,rounded rectangle,rounded rectangle left arc=none,minimum width=1cm] (name) at (7.4,2.1) {$X$};
	
	\end{tikzpicture}
	\caption{\small Circuit diagram of a computation on a BwS. $U_E$ represents Bob's deviations during the entangling operation, while each $U^{^{(k)}}_{i,j}$ represents Bob's deviations before the measurement of physical qubit $(i,j)$ along computation $k$. Bob's private register is initialized in the state $\ket{0}^{|B|}$. For simplicity, in the above picture we have rewritten each angle $\theta_{i,j}^{(k)}+\pi\sum_{(i',j')\sim(i,j)}{r}_{i',j'}^{\prime(k)}$ as $\theta_{i,j}^{\prime(k)}$.}
	\label{fig:MBQC1circ6}
\end{figure}

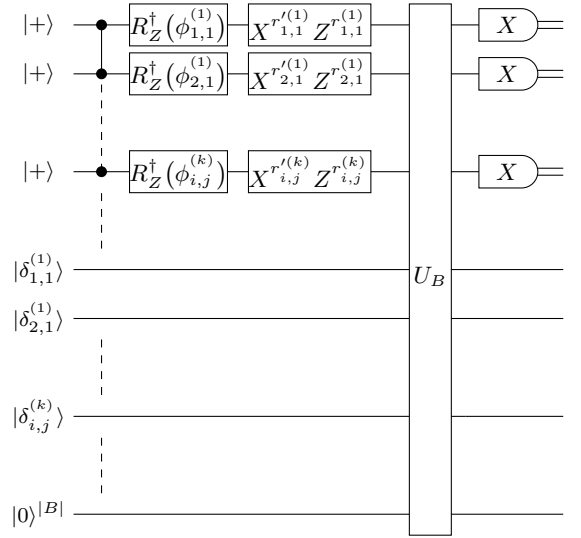
\begin{figure}[H]
	\centering
	\begin{tikzpicture}[scale=0.93, every node/.style={scale=1}]
	\draw (0.4,3.5) -- (0.4,4.2);
	\draw [dashed] (0.4,2.1) -- (0.4,3.5);
	\draw [fill] (0.4,4.2) circle [radius=0.7mm];
	\draw [fill] (0.4,3.5) circle [radius=0.7mm];
	\draw [fill] (0.4,2.1) circle [radius=0.7mm];
	
	\draw (0,4.2) -- (6,4.2);
	\draw (6.4,4.25) -- (7,4.25);
	\draw (6,4.15) -- (7,4.15);
	\draw (0,3.5) -- (6,3.5);
	\draw (6.4,3.55) -- (7,3.55);
	\draw (6.4,3.45) -- (7,3.45);
	\draw (0,2.1) -- (6,2.1);
	\draw (6.4,2.15) -- (7,2.15);
	\draw (6.4,2.05) -- (7,2.05);
	\draw (0,0.7) -- (7,0.7);
	\draw (0,0.0) -- (7,0.0);
	\draw (0,-1.4) -- (5.6,-1.4);
	\draw (5.6,-1.4) -- (7,-1.4);
	\draw [dashed] (0.4,-0.3) -- (0.4,-1.1);
	\draw (0,-2.8) -- (7,-2.8);
	\draw [dashed] (0.4,-2.5) -- (0.4,-1.7);
	\draw [dashed] (0.4,1.0) -- (0.4,1.8);
	
	\node at (-0.5,4.2) {\footnotesize $\ket{+}$};
	\node at (-0.5,3.5) {\footnotesize $\ket{+}$};
	\node at (-0.5,2.1) {\footnotesize $\ket{+}$};
	\node at (-0.5,0.7) {\footnotesize $\ket{\delta^{^{(1)}}_{1,1}}$};
	\node at (-0.5,0.0) {\footnotesize $\ket{\delta^{^{(1)}}_{2,1}}$};
	\node at (-0.5,-1.4) {\footnotesize $\ket{\delta^{^{(k)}}_{i,j}}$};
	
	\draw [fill=white] (0.8,3.9) rectangle (2.2,4.5);
	\node at (1.5,4.2) {\small $R^{\dagger}_{Z}\big(\phi^{^{(1)}}_{1,1}\big)$};
	\draw [fill=white] (0.8,3.2) rectangle (2.2,3.8);
	\node at (1.5,3.5) {\small $R^{\dagger}_{Z}\big(\phi^{^{(1)}}_{2,1}\big)$};
	\draw [fill=white] (0.8,1.8) rectangle (2.2,2.4);
	\node at (1.5,2.1) {\small $R^{\dagger}_{Z}\big(\phi^{^{(k)}}_{i,j}\big)$};	
	
	\draw [fill=white] (3.3-0.8,4.5) rectangle (3.35+0.9,3.9);
	\node at (3.355,4.2) {\small $X^{r^{{\prime(1)}}_{_{1,1}}}Z^{{r}^{{(1)}}_{_{1,1}}}$};
	\draw [fill=white] (3.3-0.8,3.2) rectangle (3.35+0.9,3.8);
	\node at (3.355,3.5) {\small $X^{r^{{\prime(1)}}_{_{2,1}}}Z^{{r}^{{(1)}}_{_{2,1}}}$};
	\draw [fill=white] (3.3-0.8,1.8) rectangle (3.35+0.9,2.4);
	\node at (3.355,2.1) {\small $X^{r^{{\prime(k)}}_{_{i,j}}}Z^{{r}^{{(k)}}_{_{i,j}}}$};
	
	\node at (-0.5,-2.8) {\footnotesize $\ket{0}^{|B|}$};
	\draw [fill=white] (4.8,4.5) rectangle (5.4,-3.1);
	\node at (5.1,0.6) {\small $U_B$};
	
	\node [fill=white,draw,rounded rectangle,rounded rectangle left arc=none,minimum width=1cm] (name) at (6.2,4.2) {$X$};
	\node [fill=white,draw,rounded rectangle,rounded rectangle left arc=none,minimum width=1cm] (name) at (6.2,3.5) {$X$};
	\node [fill=white,draw,rounded rectangle,rounded rectangle left arc=none,minimum width=1cm] (name) at (6.2,2.1) {$X$};

	\end{tikzpicture}
	\caption{\small Simplification of circuit in Figure \ref{fig:MBQC1circ6} where (i) deviations are moved toward the end of the circuit and {merged} into $U_B$, (ii) the angles $\theta^{(k)}_{i,j}$ contained in rotations are cancelled out with rotations of the state of the physical qubits, (iii) a summation is made over the angles $\theta_{i,j}^{(k)}$ and (iv) the controlled-$R_Z$ gates are rewritten as un-controlled unitaries.}
	\label{fig:MBQC1circ7}
\end{figure}

\noindent\textit{tensor products of the set of operators $\{\mathbb{1},Z,X,Y\}$, the following equation holds:}
\begin{equation}
\sum_{r=1}^{4^N}Q_{r}PQ_{r}\rho Q_{r}P'Q_{r}=0\textrm{ }\forall\textrm{ }P\neq P' 
\end{equation}
\begin{lemma}~\label{lem:r}
Consider a computation implementing an R-trap. Suppose that some of the physical qubits composing the BwS are phase-flipped. Then, for any combination of phase-flips, apart from two well-defined sets of combinations denoted as ``Type-I'' and ``Type-II'', the average probability of obtaining the string $(0,0,..,0)$ as outcome of last-column measurements is upper bounded by 3/4, i.e. phase-flips are detected with probability at least 1/4. 
\end{lemma}
\begin{lemma}~\label{lem:c}
Consider a computation implementing a C-trap. Suppose that some of the physical qubits composing the BwS are phase-flipped. Then, for any combination of phase-flips belonging to the ``Type-I''  or ``Type-II'' {sets
 or to their ``cross product''}, the average probability of obtaining the string $(0,0,..,0)$ as outcome of last-column measurements is upper bounded by 1/2, i.e. phase-flips are detected with probability at least 1/2. 
\end{lemma}
We first use the Lemmas to prove Theorem \ref{theorem:ver1} and then we prove the Lemmas. The sets of Type-I and Type-II errors are formally introduced along the proofs of the Lemmas and illustrated in Figure \ref{fig:BSflip5}. 

\begin{proof}\textit{(Theorem \ref{theorem:ver1})}
First, we show that the protocol is complete with $\delta=1$. To do this, we show that the Sub-protocols assign the correct measurement angles:

\begin{itemize}
\item \underline{Sub-protocol \hyperlink{pr:spr1}{1.1}} For any tape $y\in(1,..,w)$ and row $i\in(1,..,n)$, the assignments made in steps (1.1) and (1.2) of Sub-protocol \hyperlink{pr:spr1}{1.1} correspond to Hadamards, $R_Z$-gates and $R_X$-gates (Figure \ref{fig:BS}). Thus, the logical unitary implemented within the computation acts locally on each logical qubit $i$. Considering the role of the counter $count$ in steps (1.1) and (1.2), it can be seen that (i) the gates implemented on each logical qubit $i$ by measuring physical qubits in the first tape are either Hadamards or $XY$-plane rotations, (ii) an odd number of Hadamards is implemented between $XY$-plane and $ZY$-plane rotations, and (iii) the gates implemented on each logical qubit $i$ by measuring physical qubits in the last tape are either Hadamards or $XY$-plane rotations. As a consequence, the relation $HR_X(\Phi)H=R_Z(\Phi)$ valid for any $\Phi$ implies that the overall unitary acting on each logical qubit $i$ is a rotation on the $XY$-plane of the Bloch sphere. Due to the same relation, the angle of the overall rotation can be (efficiently) computed as a summation of all of the angles characterizing the rotations implemented in the various tapes. Thus, the measurements of qubits in the last column of the BwS are expected to yield 0.

\item \underline{Sub-protocol \hyperlink{pr:spr2}{1.2}} Considering Figure \ref{fig:BS}, it can be seen that (i) the assignments made in step (1.1) of Sub-protocol \hyperlink{pr:spr2}{1.2} produce Pauli-$Z$ gates acting on a random subset of logical qubits and (ii) the assignments made in step (1.2) produce CNOTs in the logical circuit, target and control qubits being chosen at random. To see that the measurements
\end{itemize}

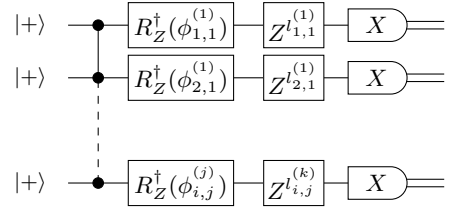
\begin{figure}[H]
	\centering
	\begin{tikzpicture}
	\draw (0.4,3.5) -- (0.4,4.2);
	\draw [dashed] (0.4,2.1) -- (0.4,3.5);
	\draw [fill] (0.4,4.2) circle [radius=0.7mm];
	\draw [fill] (0.4,3.5) circle [radius=0.7mm];
	\draw [fill] (0.4,2.1) circle [radius=0.7mm];
	
	\draw (0,4.2) -- (4.1,4.2);
	\draw (4.1,4.25) -- (5,4.25);
	\draw (4.1,4.15) -- (5,4.15);
	\draw (0,3.5) -- (4.1,3.5);
	\draw (4.1,3.55) -- (5,3.55);
	\draw (4.1,3.45) -- (5,3.45);
	\draw (0,2.1) -- (4.1,2.1);
	\draw (4.1,2.15) -- (5,2.15);
	\draw (4.1,2.05) -- (5,2.05);
	
	\node at (-0.5,4.2) {\footnotesize $\ket{+}$};
	\node at (-0.5,3.5) {\footnotesize $\ket{+}$};
	\node at (-0.5,2.1) {\footnotesize $\ket{+}$};
	
	\draw [fill=white] (0.8,3.9) rectangle (2.2,4.5);
	\node at (1.5,4.2) {\small $R^{\dagger}_{Z}(\phi^{^{(1)}}_{1,1})$};
	
	\draw [fill=white] (0.8,3.2) rectangle (2.2,3.8);
	\node at (1.5,3.5) {\small $R^{\dagger}_{Z}(\phi^{^{(1)}}_{2,1})$};
	
	\draw [fill=white] (0.8,1.8) rectangle (2.2,2.4);
	\node at (1.5,2.1) {\small $R^{\dagger}_Z(\phi^{^{(j)}}_{i,j})$};

	\draw [fill=white] (2.6,4.5) rectangle (3.4,3.9);
	\node at (3.0,4.2) {\small $Z^{l^{^{(1)}}_{1,1}}$};
	\draw [fill=white] (2.6,3.8) rectangle (3.4,3.2);
	\node at (3.0,3.5) {\small $Z^{l^{^{(1)}}_{2,1}}$};
	\draw [fill=white] (2.6,2.4) rectangle (3.4,1.8);
	\node at (3.0,2.1) {\small $Z^{l^{^{(k)}}_{i,j}}$};
	
	\node [fill=white,draw,rounded rectangle,rounded rectangle left arc=none,minimum width=1cm] (name) at (4.1,4.2) {$X$};
	\node [fill=white,draw,rounded rectangle,rounded rectangle left arc=none,minimum width=1cm] (name) at (4.1,3.5) {$X$};
	\node [fill=white,draw,rounded rectangle,rounded rectangle left arc=none,minimum width=1cm] (name) at (4.1,2.1) {$X$};

	\end{tikzpicture}
	\caption{\small Simplification of circuit in Figure \ref{fig:MBQC1circ7} after Bob's private system and the classical registers are traced out. The parameters $l^{^{(k)}}_{i,j}=0,1$ take into account deviations affecting physical qubits $(i,j)$.}
	\label{fig:MBQC1circ8}
\end{figure}

\begin{itemize}
	 \item[] of last-column qubits yield 0, it can be used that CNOT$_{c,t}=(\mathbb{1}_{\textrm{c}}\otimes|+\rangle_{\textrm{t}}\langle+|+Z_{\textrm{c}}\otimes|-\rangle_{\textrm{t}}\langle-|)/2$, where the index ``c'' labels the control qubit and ``t'' the target. This relation also allows to compute angles $\phi_{i,m}$ in an efficient way.
\end{itemize}
\noindent Thus, if Bob is honest, the trap computations always yield the expected outcome and Alice does always accept. This proves that $\delta=1$.

To prove that the protocol is $\varepsilon$-sound, consider the operations on Bob's side (represented in the circuit diagram in Figure \ref{fig:MBQC1circ6}). As in \cite{FK12,KD17}, we describe Bob's dishonest operations through a collection of unitaries (the ``deviations'') acting on yet-to-be-measured qubits, on the angles and on Bob's private system. The deviation $U_E$ represents the dishonest entangling operation, while $U^{^{(k)}}_{i,j}$ represents Bob deviations before the measurement of qubit $(i,j)$ in the $k$th graph. Thus, using for simplicity the notation $O\cdot\sigma=O\sigma O^{\dagger}$ for an operator $O$ and a state $\sigma$, the state of the system immediately before the measurements is of the form
\begin{align*}
{\sigma}_{\textup{out}}&=\sum_{\overline{\theta},\overline{r},\overline{r}'}\frac{U_{n,m}^{(v+1)}{cR_{n,m}^{\dagger(v+1)}}..U_{1,1}^{(1)}cR_{1,1}^{\dagger(1)}U_EE\cdot(\sigma^{\overline{\theta},\overline{r},\overline{r}'}_{\textup{in}})}{2^{2nm(v+1)}8^{nm(v+1)}}\textrm{ ,}
\end{align*}
where each operator ${cR}_{i,j}^{\dagger(k)}$ represents the rotation acting on qubit $(i,j)$ of computation $k$ and controlled by angle $\delta_{i,j}^{(k)}$, $E$ represents the global entangling operation and the summation is made over all of the combinations of the angles $\{\theta_{i,j}^{(k)}\}$ and of the sets $\{r_{i,j}^{(k)}\}$ and $\{r_{i,j}^{\prime(k)}\}$. For any particular choice of $\overline{\theta}$, $\overline{r}$ and $\overline{r}'$, the state $\sigma^{\overline{\theta},\overline{r},\overline{r}'}_{\textup{in}}$ is equal to
\begin{equation*}
\sigma_{\textup{in}}^{\overline{\theta},\overline{r},\overline{r}'}=\bigotimes_{i,j,k}|+\rangle_{\theta_{i,j}^{^{\prime(k)}}}\langle+|\otimes|\delta_{i,j}^{(k)}\rangle\langle\delta_{i,j}^{(k)}|\otimes|0\rangle^{|B|}\langle0|\textrm{ ,}
\end{equation*}
where $\theta_{i,j}^{^{\prime(k)}}=\theta_{i,j}^{^{(k)}}+\pi\sum_{(i',j')\sim(i,j)}^{(k)}r_{i',j'}^{\prime(k)}$.
To proceed, we rewrite ${\sigma}_{\textup{out}}$ as follows. First, we move the deviations towards the end of the circuit and ``{merge}'' them into a sole unitary $U_B$. Next, as we do along the proof of correctness (see Appendix \ref{app:corr}), we rewrite the controlled $R_Z$-gates as un-controlled rotations and cancel every angle $\theta^{(k)}_{i,j}$ in the $R_Z$-gates with the $XY$-plane pre-rotations characterizing the state of the physical qubits. At this point, summing over all the possible $\theta^{(k)}_{i,j}$, the angles $\delta^{(k)}_{i,j}$ become the completely mixed state. We thus obtain the state
\begin{align*}
{\sigma}_{\textup{out}}&=\sum_{\overline{r},\overline{r}'}\frac{U_B\big[\otimes_{i,j,k}Z^{r_{i,j}^{(k)}}X^{r_{i,j}^{\prime(k)}}R^{\dagger}(\phi_{i,j}^{(k)})X^{r_{i,j}^{\prime(k)}}\big]E\cdot(\sigma_{\textup{in}}^{\overline{0},\overline{r},\overline{r}'})}{2^{2nm(v+1)}}
\end{align*}
where
\begin{align*}
\sigma_{\textup{in}}^{\overline{0},\overline{r},\overline{r}'}=\bigotimes_{i,j,k}&Z^{\sum_{(i',j')\sim(i,j)}{r}_{i',j'}^{\prime(k)}}|+\rangle_{i,j}^{(k)}\langle+|Z^{\sum_{(i',j')\sim(i,j)}{r}_{i',j'}^{\prime(k)}}\cr
&\otimes\mathbb{1}^{3nm(v+1)}\otimes|0\rangle^{|B|}\langle0|
\end{align*}
and $\mathbb{1}^{3nm(v+1)}$ is the identity on the classical registers containing the information about the angles. Using the relation $(X\otimes\mathbb{1})cZ=cZ(X\otimes Z)$, we can commute the Pauli-$X$\textit{s} on the right-hand side of the rotations with $E$, and subsequently cancel them out with the Pauli-$Zs$\textit{} in $\sigma_{\textup{in}}^{\overline{0},\overline{r},\overline{r}'}$. We thus obtain the state in Figure \ref{fig:MBQC1circ7}.

When both Bob's private system and the classical register are traced out, $U_B$ becomes a CPTP-map $\E=\{E_u\}$. Rewriting the operational elements $\{E_u\}$ as a linear combination of tensor products of Pauli operators $P_l$, the state of the system just before the measurement can be rewritten as
\begin{align*}
{\rho}^{}_{\textrm{out}}&=\sum_{\substack{u,l,l^\prime \\ \overline{r}, \overline{r}^\prime}}\frac{a_{u,l}a^*_{u,l'}}{2^{2nm(v+1)}}\textrm{ }\bigotimes_{i,j,k}P_{l|(i,j)}^{(k)}Q_{(i,j)}^{(k)}\rho_{\textrm{in}} \textrm{ }Q_{(i,j)}^{ (k)}P_{l'|(i,j)}^{(k)}\textrm{ ,}
\end{align*}
where $P_{l|(i,j)}^{(k)}$ is the component of $P_l$ acting on qubit $(i,j)$ of the $k$th computation, $Q_{(i,j)}^{(k)}=Z_{i,j}^{r_{i,j}^{(k)}}X_{i,j}^{r^{\prime(k)}_{i,j}}$, $a_{u,l}$ are complex numbers and
\begin{equation*}
\rho_{\textrm{in}}=E\bigg[\bigotimes_{i,j,k}{}R^{\dagger}_Z(\phi_{i,j}^{(k)})|+\rangle_{i,j}^{^{(k)}}\langle+|R_Z(\phi_{i,j}^{(k)})\bigg]E
\end{equation*}
After the measurements in the Pauli-$X$ basis, the above state becomes
\begin{align*}
{\rho}^{\prime}_{\textrm{out}}= & \sum_{\substack{u,l,l^\prime \\ \overline{r}, \overline{r}^\prime,\overline{s}}}\frac{a_{u,l}a^*_{u,l'}}{2^{2nm(v+1)}}\textrm{ }\bigotimes_{i,j,k}\textrm{ }Z^{s_{i,j}^{(k)}}|+\rangle_{i,j}^{^{(k)}}\langle+|Z^{s_{i,j}^{(k)}}\textrm { }\times\cr
\times&\textrm{ }_{i,j}^{^{(k)}}\langle+|Z^{s_{i,j}^{(k)}}P_{l|(i,j)}^{(k)}Q_{(i,j)}^{(k)}\rho_{\textrm{in}} \textrm{ }Q_{(i,j)}^{ (k)}P_{l'|(i,j)}^{(k)}Z^{s_{i,j}^{(k)}}|+\rangle_{i,j}^{^{(k)}}
\end{align*}
where $\overline{s}=\{s_{i,j}^{(k)}\}$ represents a combination of outcomes. After Alice recomputes $s_{i,j}^{(k)}$ as $s_{i,j}^{(k)}\oplus r_{i,j}^{(k)}$, ${\rho}^{\prime}_{\textrm{out}}$ becomes

\begin{align*}
{\rho}^{\prime}_{\textrm{out}}= & \sum_{\substack{u,l,l^\prime \\ \overline{r}, \overline{r}^\prime,\overline{s}}}\frac{a_{u,l}a^*_{u,l'}}{2^{2nm(v+1)}}\textrm{ }\bigotimes_{i,j,k}Z^{s_{i,j}^{(k)}\oplus r_{i,j}^{(k)}}|+\rangle_{i,j}^{^{(k)}}\langle+|Z^{s_{i,j}^{(k)}\oplus r_{i,j}^{(k)}} \cr
\times&\textrm{ }_{i,j}^{^{(k)}}\langle+|Z^{s_{i,j}^{(k)}}P_{l|(i,j)}^{(k)}Q_{(i,j)}^{(k)}\rho_{\textrm{in}} \textrm{ }Q_{(i,j)}^{(k)}P_{l'|(i,j)}^{(k)}Z^{s_{i,j}^{(k)}}|+\rangle_{i,j}^{^{(k)}}\\ \cr
= & \sum_{\substack{u,l,l^\prime \\ \overline{r}, \overline{r}^\prime,\overline{s}^\prime}}\frac{a_{u,l}a^*_{u,l'}}{2^{2nm(v+1)}}\textrm{ }\bigotimes_{i,j,k}\textrm{ }Z^{s_{i,j}^{\prime(k)}}|+\rangle_{i,j}^{^{(k)}}\langle+|Z^{s_{i,j}^{\prime(k)}} \cr
&\times\textrm{ }_{i,j}^{^{(k)}}\langle+|Z^{s_{i,j}^{\prime(k)}}\big[\p_{l,l'}^{(i,j,k)}(\rho_{\textrm{in}})\big]Z^{s_{i,j}^{\prime(k)}}|+\rangle_{i,j}^{^{(k)}}
\end{align*}
where
$$\p_{l,l'}^{(i,j,k)}(\rho_{\textrm{in}})=Q_{(i,j)}^{(k)}P_{l|(i,j)}^{(k)}Q_{(i,j)}^{(k)}\rho_{\textrm{in}} \textrm{ }Q_{(i,j)}^{(k)}P_{l'|(i,j)}^{(k)}Q_{(i,j)}^{(k)}$$
To obtain the second equality, one can (i) make the change of variable $s_{(i,j)}^{(k)}\rightarrow s_{(i,j)}^{\prime(k)}=s_{(i,j)}^{(k)}\oplus r_{(i,j)}^{(k)}$ and (ii) use the fact that Pauli-$X$ operators stabilize $\ket{+}$ states to add extra $X_{i,j}^{r^{\prime(k)}_{i,j}}$. Summing over all possible $\{r_{i,j}^{(k)}\}$ and $\{r_{i,j}^{\prime(k)}\}$ and using the Pauli-twirl (Lemma \ref{lem:tw}), the above state becomes
\begin{align*}
{\rho}^{\prime}_{\textrm{out}}= & \sum_{\substack{l,\overline{s}}}b_l\textrm{ }\bigotimes_{i,j,k}\textrm{ }Z^{s_{i,j}^{(k)}}|+\rangle_{i,j}^{^{(k)}}\langle+|Z^{s_{i,j}^{(k)}}\cr
&\times\textrm{ }_{i,j}^{^{(k)}}\langle+|Z^{s_{i,j}^{(k)}}P_{l|(i,j)}^{(k)}\rho_{\textrm{in}} \textrm{ }P_{l|(i,j)}^{(k)}Z^{s_{i,j}^{(k)}}|+\rangle_{i,j}^{^{(k)}}
\end{align*}
where $b_l=\sum_u|a_{u,l}|^2$ and $\sum_l b_l=1$. Thus, the effects of $\E$ are reduced to that of a convex combination of tensor products of Pauli operators affecting each qubit locally. Moreover, since Pauli-$X$ operators stabilize $\ket{+}$ states, the effect of the Pauli-$X$ components of the Pauli operators $P_l$ on the system is trivial. Thus, the circuit in Figure \ref{fig:MBQC1circ7} can be rewritten as the circuit in Figure \ref{fig:MBQC1circ8}. This result is similar to that obtained in \cite{ABE08,B15,KD17}crucial, and is crucial. Indeed, it allows us to use Lemmas \ref{lem:r} and \ref{lem:c}: on average, Bob's deviations reduce to local phase-flips of physical qubits, therefore the problem of detecting Bob's deviations becomes the problem of detecting all of the possible combinations of phase-flips affecting the various graphs. As illustrated in Figures \ref{fig:BSdev2} and \ref{fig:BSdev3}, this produces by-products of Pauli-$X$ and Pauli-$Z$ affecting the computation at the logical level.

To obtain equation \ref{eq:oneround}, we consider the case where Bob deviates on $\widetilde{v}$ computations. We first compute an upper bound to the probability $p(E_1\wedge E_2| \widetilde{v})$ of the events $E_1$ \textit{Bob corrupts the target computation} and $E_2$ \textit{Bob is not detected} happening simultaneously when Bob deviates on $\widetilde{v}$ computations. Next, we maximize over $\widetilde{v}$. The probability $p(E_1\wedge E_2| \widetilde{v})$ is upper-bounded as 
\begin{align*}~\label{fig:}
p(E_1\wedge E_2| \widetilde{v})= &\textrm{ } \textup{}{p{(E_1|\widetilde{v})}}\textup{}{p{(E_2|E_1,\widetilde{v})}}\cr
&\leq\frac{\widetilde{v}}{v+1}\bigg(\frac{7}{8}\bigg)^{\widetilde{v}-1}\textrm{ ,}
\end{align*}
\vspace{0.1cm}

\onecolumngrid

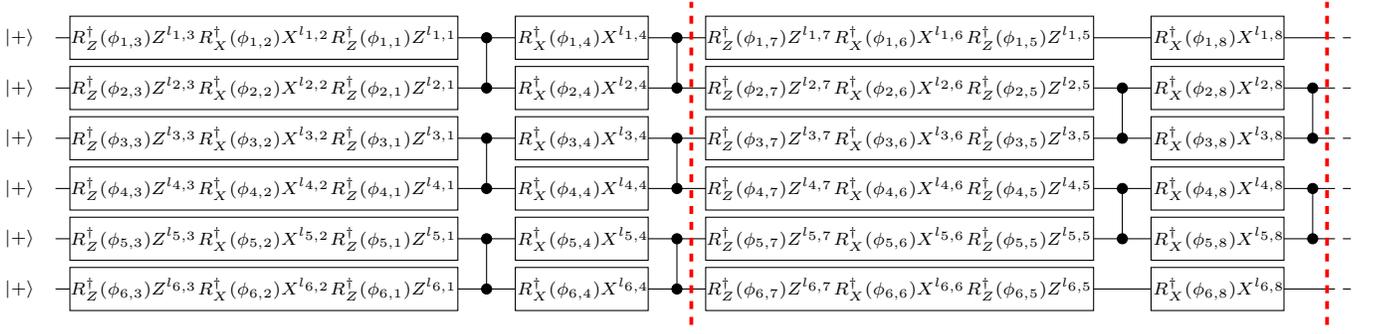
\begin{figure}[H]
\begin{tikzpicture}[scale=0.955, every node/.style={scale=0.95}]


\node at (-0.5,0.0) {\scriptsize $\ket{+}$};
\node at (-0.5,0.7) {\scriptsize $\ket{+}$};
\node at (-0.5,1.4) {\scriptsize $\ket{+}$};
\node at (-0.5,2.1) {\scriptsize $\ket{+}$};
\node at (-0.5,2.8) {\scriptsize $\ket{+}$};
\node at (-0.5,3.5) {\scriptsize $\ket{+}$};

\draw (0.0,0.0) -- (8.65+8.65+0.2+0.2,0.0);
\draw (0.0,0.7) -- (8.65+8.65+0.2+0.2,0.7);
\draw (0.0,1.4) -- (8.65+8.65+0.2+0.2,1.4);
\draw (0.0,2.1) -- (8.65+8.65+0.2+0.2,2.1);
\draw (0.0,2.8) -- (8.65+8.65+0.2+0.2,2.8);
\draw (0.0,3.5) -- (8.65+8.65+0.2+0.2,3.5);

\draw [dashed] (8.65+8.65+0.2+0.2,0.0) -- (8.65+8.65+0.2+0.6,0.0);
\draw [dashed] (8.65+8.65+0.2+0.2,0.7) -- (8.65+8.65+0.2+0.6,0.7);
\draw [dashed] (8.65+8.65+0.2+0.2,1.4) -- (8.65+8.65+0.2+0.6,1.4);
\draw [dashed] (8.65+8.65+0.2+0.2,2.1) -- (8.65+8.65+0.2+0.6,2.1);
\draw [dashed] (8.65+8.65+0.2+0.2,2.8) -- (8.65+8.65+0.2+0.6,2.8);
\draw [dashed] (8.65+8.65+0.2+0.2,3.5) -- (8.65+8.65+0.2+0.6,3.5);

\draw [fill=white] (0.2,-0.3) rectangle (5.6,0.3);
\draw [fill=white] (0.2,0.4) rectangle (5.6,1.0);
\draw [fill=white] (0.2,1.1) rectangle (5.6,1.7);
\draw [fill=white] (0.2,1.8) rectangle (5.6,2.4);
\draw [fill=white] (0.2,2.5) rectangle (5.6,3.1);
\draw [fill=white] (0.2,3.2) rectangle (5.6,3.8);

\node at (2.9,0.0) {\scriptsize $R^{\dagger}_Z{(\phi_{6,3})}Z^{l_{6,3}}R^{\dagger}_X(\phi_{6,2})X^{l_{6,2}}R^{\dagger}_Z(\phi_{6,1})Z^{l_{6,1}}$};\node at (2.9,0.7) {\scriptsize $R^{\dagger}_Z{(\phi_{5,3})}Z^{l_{5,3}}R^{\dagger}_X(\phi_{5,2})X^{l_{5,2}}R^{\dagger}_Z(\phi_{5,1})Z^{l_{5,1}}$};\node at (2.9,1.4) {\scriptsize $R^{\dagger}_Z{(\phi_{4,3})}Z^{l_{4,3}}R^{\dagger}_X(\phi_{4,2})X^{l_{4,2}}R^{\dagger}_Z(\phi_{4,1})Z^{l_{4,1}}$};\node at (2.9,2.1) {\scriptsize $R^{\dagger}_Z{(\phi_{3,3})}Z^{l_{3,3}}R^{\dagger}_X(\phi_{3,2})X^{l_{3,2}}R^{\dagger}_Z(\phi_{3,1})Z^{l_{3,1}}$};\node at (2.9,2.8) {\scriptsize $R^{\dagger}_Z{(\phi_{2,3})}Z^{l_{2,3}}R^{\dagger}_X(\phi_{2,2})X^{l_{2,2}}R^{\dagger}_Z(\phi_{2,1})Z^{l_{2,1}}$};\node at (2.9,3.5) {\scriptsize $R^{\dagger}_Z{(\phi_{1,3})}Z^{l_{1,3}}R^{\dagger}_X(\phi_{1,2})X^{l_{1,2}}R^{\dagger}_Z(\phi_{1,1})Z^{l_{1,1}}$};

\draw [fill=black] (6,0.0) circle [radius=0.07cm];
\draw [fill=black] (6,0.7) circle [radius=0.07cm];
\draw [fill=black] (6,1.4) circle [radius=0.07cm];
\draw [fill=black] (6,2.1) circle [radius=0.07cm];
\draw [fill=black] (6,2.8) circle [radius=0.07cm];
\draw [fill=black] (6,3.5) circle [radius=0.07cm];

\draw (6,0.0) -- (6,0.7);
\draw (6,1.4) -- (6,2.1);
\draw (6,2.8) -- (6,3.5);

\draw [fill=white] (6.4,-0.3) rectangle (8.25,0.3);
\draw [fill=white] (6.4,0.4) rectangle (8.25,1.0);
\draw [fill=white] (6.4,1.1) rectangle (8.25,1.7);
\draw [fill=white] (6.4,1.8) rectangle (8.25,2.4);
\draw [fill=white] (6.4,2.5) rectangle (8.25,3.1);
\draw [fill=white] (6.4,3.2) rectangle (8.25,3.8);

\node at (7.35,0.0) {\scriptsize $R^{\dagger}_X{(\phi_{6,4})}X^{l_{6,4}}$};
\node at (7.35,0.7) {\scriptsize $R^{\dagger}_X{(\phi_{5,4})}X^{l_{5,4}}$};
\node at (7.35,1.4) {\scriptsize $R^{\dagger}_X{(\phi_{4,4})}X^{l_{4,4}}$};
\node at (7.35,2.1) {\scriptsize $R^{\dagger}_X{(\phi_{3,4})}X^{l_{3,4}}$};
\node at (7.35,2.8) {\scriptsize $R^{\dagger}_X{(\phi_{2,4})}X^{l_{2,4}}$};
\node at (7.35,3.5) {\scriptsize $R^{\dagger}_X{(\phi_{1,4})}X^{l_{1,4}}$};

\draw [fill=black] (8.65,0.0) circle [radius=0.07cm];
\draw [fill=black] (8.65,0.7) circle [radius=0.07cm];
\draw [fill=black] (8.65,1.4) circle [radius=0.07cm];
\draw [fill=black] (8.65,2.1) circle [radius=0.07cm];
\draw [fill=black] (8.65,2.8) circle [radius=0.07cm];
\draw [fill=black] (8.65,3.5) circle [radius=0.07cm];

\draw (8.65,0.0) -- (8.65,0.7);
\draw (8.65,1.4) -- (8.65,2.1);
\draw (8.65,2.8) -- (8.65,3.5);

\draw [fill=white] (0.4+8.65,-0.3) rectangle (5.6+8.65+0.2,0.3);

\draw [red, dashed, line width=0.5mm] (8.65+8.65+0.4,-0.5) -- (8.65+8.65+0.4,4.0);
\draw [red, dashed, line width=0.5mm] (8.65+0.2,-0.5) -- (8.65+0.2,4.0);
\draw [fill=white] (0.4+8.65,0.4) rectangle (5.6+8.65+0.2,1.0);
\draw [fill=white] (0.4+8.65,1.1) rectangle (5.6+8.65+0.2,1.7);
\draw [fill=white] (0.4+8.65,1.8) rectangle (5.6+8.65+0.2,2.4);
\draw [fill=white] (0.4+8.65,2.5) rectangle (5.6+8.65+0.2,3.1);
\draw [fill=white] (0.4+8.65,3.2) rectangle (5.6+8.65+0.2,3.8);

\node at (2.9+0.2+8.65,0.0) {\scriptsize $R^{\dagger}_Z{(\phi_{6,7})}Z^{l_{6,7}}R^{\dagger}_X(\phi_{6,6})X^{l_{6,6}}R^{\dagger}_Z(\phi_{6,5})Z^{l_{6,5}}$};\node at (2.9+8.65+0.2,0.7) {\scriptsize $R^{\dagger}_Z{(\phi_{5,7})}Z^{l_{5,7}}R^{\dagger}_X(\phi_{5,6})X^{l_{5,6}}R^{\dagger}_Z(\phi_{5,5})Z^{l_{5,5}}$};\node at (2.9+8.65+0.2,1.4) {\scriptsize $R^{\dagger}_Z{(\phi_{4,7})}Z^{l_{4,7}}R^{\dagger}_X(\phi_{4,6})X^{l_{4,6}}R^{\dagger}_Z(\phi_{4,5})Z^{l_{4,5}}$};\node at (2.9+8.65+0.2,2.1) {\scriptsize $R^{\dagger}_Z{(\phi_{3,7})}Z^{l_{3,7}}R^{\dagger}_X(\phi_{3,6})X^{l_{3,6}}R^{\dagger}_Z(\phi_{3,5})Z^{l_{3,5}}$};\node at (2.9+8.65+0.2,2.8) {\scriptsize $R^{\dagger}_Z{(\phi_{2,7})}Z^{l_{2,7}}R^{\dagger}_X(\phi_{2,6})X^{l_{2,6}}R^{\dagger}_Z(\phi_{2,5})Z^{l_{2,5}}$};\node at (2.9+8.65+0.2,3.5) {\scriptsize $R^{\dagger}_Z{(\phi_{1,7})}Z^{l_{1,7}}R^{\dagger}_X(\phi_{1,6})X^{l_{1,6}}R^{\dagger}_Z(\phi_{1,5})Z^{l_{1,5}}$};

\draw [fill=black] (6+8.65+0.2,0.7) circle [radius=0.07cm];
\draw [fill=black] (6+8.65+0.2,1.4) circle [radius=0.07cm];
\draw [fill=black] (6+8.65+0.2,2.1) circle [radius=0.07cm];
\draw [fill=black] (6+8.65+0.2,2.8) circle [radius=0.07cm];

\draw (6+8.65+0.2,0.7) -- (6+8.65+0.2,1.4);
\draw (6+8.65+0.2,2.1) -- (6+8.65+0.2,2.8);

\draw [fill=white] (6+8.65+0.2+0.4,-0.3) rectangle (8.25+8.65+0.2,0.3);
\draw [fill=white] (6+8.65+0.2+0.4,0.4) rectangle (8.25+8.65+0.2,1.0);
\draw [fill=white] (6+8.65+0.2+0.4,1.1) rectangle (8.25+8.65+0.2,1.7);
\draw [fill=white] (6+8.65+0.2+0.4,1.8) rectangle (8.25+8.65+0.2,2.4);
\draw [fill=white] (6+8.65+0.2+0.4,2.5) rectangle (8.25+8.65+0.2,3.1);
\draw [fill=white] (6+8.65+0.2+0.4,3.2) rectangle (8.25+8.65+0.2,3.8);

\node at (7.35+8.65+0.2,0.0) {\scriptsize $R^{\dagger}_X{(\phi_{6,8})}X^{l_{6,8}}$};
\node at (7.35+8.65+0.2,0.7) {\scriptsize $R^{\dagger}_X{(\phi_{5,8})}X^{l_{5,8}}$};
\node at (7.35+8.65+0.2,1.4) {\scriptsize $R^{\dagger}_X{(\phi_{4,8})}X^{l_{4,8}}$};
\node at (7.35+8.65+0.2,2.1) {\scriptsize $R^{\dagger}_X{(\phi_{3,8})}X^{l_{3,8}}$};
\node at (7.35+8.65+0.2,2.8) {\scriptsize $R^{\dagger}_X{(\phi_{2,8})}X^{l_{2,8}}$};
\node at (7.35+8.65+0.2,3.5) {\scriptsize $R^{\dagger}_X{(\phi_{1,8})}X^{l_{1,8}}$};

\draw [fill=black] (8.65+8.65+0.2,0.7) circle [radius=0.07cm];
\draw [fill=black] (8.65+8.65+0.2,1.4) circle [radius=0.07cm];
\draw [fill=black] (8.65+8.65+0.2,2.1) circle [radius=0.07cm];
\draw [fill=black] (8.65+8.65+0.2,2.8) circle [radius=0.07cm];

\draw (8.65+8.65+0.2,0.7) -- (8.65+8.65+0.2,1.4);
\draw (8.65+8.65+0.2,2.1) -- (8.65+8.65+0.2,2.8);

\end{tikzpicture}
\caption{\small Logical circuit associated to a computation on a six-row BwS affected by phase-flips at the physical levels. Red dashed lines separate operations implemented within different tapes of the BwS. Each row represents a logical qubit. Extra Pauli operators $X^{l_{i,j}}$ and $Z^{l_{i,j}}$ are logical by-products generated by phase-flipping the physical qubit $(i,j)$. We also refer to by-products as ``errors''.}
\label{fig:BSdev2}
\end{figure}

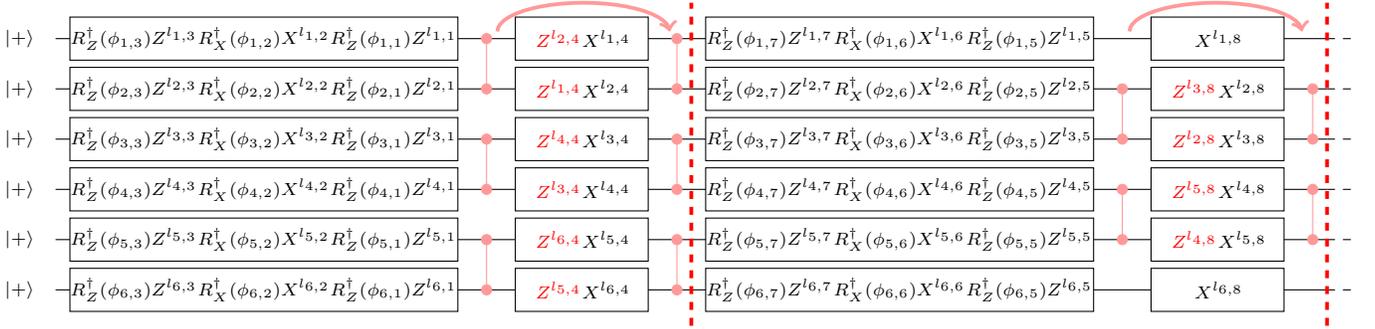
\begin{figure}[H]
\begin{tikzpicture}[scale=0.955, every node/.style={scale=0.95}]


\node at (-0.5,0.0) {\scriptsize $\ket{+}$};
\node at (-0.5,0.7) {\scriptsize $\ket{+}$};
\node at (-0.5,1.4) {\scriptsize $\ket{+}$};
\node at (-0.5,2.1) {\scriptsize $\ket{+}$};
\node at (-0.5,2.8) {\scriptsize $\ket{+}$};
\node at (-0.5,3.5) {\scriptsize $\ket{+}$};

\draw (0.0,0.0) -- (8.65+8.65+0.2+0.2,0.0);
\draw (0.0,0.7) -- (8.65+8.65+0.2+0.2,0.7);
\draw (0.0,1.4) -- (8.65+8.65+0.2+0.2,1.4);
\draw (0.0,2.1) -- (8.65+8.65+0.2+0.2,2.1);
\draw (0.0,2.8) -- (8.65+8.65+0.2+0.2,2.8);
\draw (0.0,3.5) -- (8.65+8.65+0.2+0.2,3.5);

\draw [dashed] (8.65+8.65+0.2+0.2,0.0) -- (8.65+8.65+0.2+0.6,0.0);
\draw [dashed] (8.65+8.65+0.2+0.2,0.7) -- (8.65+8.65+0.2+0.6,0.7);
\draw [dashed] (8.65+8.65+0.2+0.2,1.4) -- (8.65+8.65+0.2+0.6,1.4);
\draw [dashed] (8.65+8.65+0.2+0.2,2.1) -- (8.65+8.65+0.2+0.6,2.1);
\draw [dashed] (8.65+8.65+0.2+0.2,2.8) -- (8.65+8.65+0.2+0.6,2.8);
\draw [dashed] (8.65+8.65+0.2+0.2,3.5) -- (8.65+8.65+0.2+0.6,3.5);

\draw [fill=white] (0.2,-0.3) rectangle (5.6,0.3);
\draw [fill=white] (0.2,0.4) rectangle (5.6,1.0);
\draw [fill=white] (0.2,1.1) rectangle (5.6,1.7);
\draw [fill=white] (0.2,1.8) rectangle (5.6,2.4);
\draw [fill=white] (0.2,2.5) rectangle (5.6,3.1);
\draw [fill=white] (0.2,3.2) rectangle (5.6,3.8);

\node at (2.9,0.0) {\scriptsize $R^{\dagger}_Z{(\phi_{6,3})}Z^{l_{6,3}}R^{\dagger}_X(\phi_{6,2})X^{l_{6,2}}R^{\dagger}_Z(\phi_{6,1})Z^{l_{6,1}}$};\node at (2.9,0.7) {\scriptsize $R^{\dagger}_Z{(\phi_{5,3})}Z^{l_{5,3}}R^{\dagger}_X(\phi_{5,2})X^{l_{5,2}}R^{\dagger}_Z(\phi_{5,1})Z^{l_{5,1}}$};\node at (2.9,1.4) {\scriptsize $R^{\dagger}_Z{(\phi_{4,3})}Z^{l_{4,3}}R^{\dagger}_X(\phi_{4,2})X^{l_{4,2}}R^{\dagger}_Z(\phi_{4,1})Z^{l_{4,1}}$};\node at (2.9,2.1) {\scriptsize $R^{\dagger}_Z{(\phi_{3,3})}Z^{l_{3,3}}R^{\dagger}_X(\phi_{3,2})X^{l_{3,2}}R^{\dagger}_Z(\phi_{3,1})Z^{l_{3,1}}$};\node at (2.9,2.8) {\scriptsize $R^{\dagger}_Z{(\phi_{2,3})}Z^{l_{2,3}}R^{\dagger}_X(\phi_{2,2})X^{l_{2,2}}R^{\dagger}_Z(\phi_{2,1})Z^{l_{2,1}}$};\node at (2.9,3.5) {\scriptsize $R^{\dagger}_Z{(\phi_{1,3})}Z^{l_{1,3}}R^{\dagger}_X(\phi_{1,2})X^{l_{1,2}}R^{\dagger}_Z(\phi_{1,1})Z^{l_{1,1}}$};

\draw [fill=white] (6.4,-0.3) rectangle (8.25,0.3);
\draw [fill=white] (6.4,0.4) rectangle (8.25,1.0);
\draw [fill=white] (6.4,1.1) rectangle (8.25,1.7);
\draw [fill=white] (6.4,1.8) rectangle (8.25,2.4);
\draw [fill=white] (6.4,2.5) rectangle (8.25,3.1);
\draw [fill=white] (6.4,3.2) rectangle (8.25,3.8);

\node at (7.35,0.0) {\scriptsize $\textcolor{red}{Z^{l_{5,4}}}X^{l_{6,4}}$};
\node at (7.35,0.7) {\scriptsize $\textcolor{red}{Z^{l_{6,4}}}X^{l_{5,4}}$};
\node at (7.35,1.4) {\scriptsize $\textcolor{red}{Z^{l_{3,4}}}X^{l_{4,4}}$};
\node at (7.35,2.1) {\scriptsize $\textcolor{red}{Z^{l_{4,4}}}X^{l_{3,4}}$};
\node at (7.35,2.8) {\scriptsize $\textcolor{red}{Z^{l_{1,4}}}X^{l_{2,4}}$};
\node at (7.35,3.5) {\scriptsize $\textcolor{red}{Z^{l_{2,4}}}X^{l_{1,4}}$};

\draw [fill=white] (0.4+8.65,-0.3) rectangle (5.6+8.65+0.2,0.3);

\draw [red, dashed, line width=0.5mm] (8.65+8.65+0.4,-0.5) -- (8.65+8.65+0.4,4.0);
\draw [red, dashed, line width=0.5mm] (8.65+0.2,-0.5) -- (8.65+0.2,4.0);
\draw [fill=white] (0.4+8.65,0.4) rectangle (5.6+8.65+0.2,1.0);
\draw [fill=white] (0.4+8.65,1.1) rectangle (5.6+8.65+0.2,1.7);
\draw [fill=white] (0.4+8.65,1.8) rectangle (5.6+8.65+0.2,2.4);
\draw [fill=white] (0.4+8.65,2.5) rectangle (5.6+8.65+0.2,3.1);
\draw [fill=white] (0.4+8.65,3.2) rectangle (5.6+8.65+0.2,3.8);

\node at (2.9+0.2+8.65,0.0) {\scriptsize $R^{\dagger}_Z{(\phi_{6,7})}Z^{l_{6,7}}R^{\dagger}_X(\phi_{6,6})X^{l_{6,6}}R^{\dagger}_Z(\phi_{6,5})Z^{l_{6,5}}$};\node at (2.9+8.65+0.2,0.7) {\scriptsize $R^{\dagger}_Z{(\phi_{5,7})}Z^{l_{5,7}}R^{\dagger}_X(\phi_{5,6})X^{l_{5,6}}R^{\dagger}_Z(\phi_{5,5})Z^{l_{5,5}}$};\node at (2.9+8.65+0.2,1.4) {\scriptsize $R^{\dagger}_Z{(\phi_{4,7})}Z^{l_{4,7}}R^{\dagger}_X(\phi_{4,6})X^{l_{4,6}}R^{\dagger}_Z(\phi_{4,5})Z^{l_{4,5}}$};\node at (2.9+8.65+0.2,2.1) {\scriptsize $R^{\dagger}_Z{(\phi_{3,7})}Z^{l_{3,7}}R^{\dagger}_X(\phi_{3,6})X^{l_{3,6}}R^{\dagger}_Z(\phi_{3,5})Z^{l_{3,5}}$};\node at (2.9+8.65+0.2,2.8) {\scriptsize $R^{\dagger}_Z{(\phi_{2,7})}Z^{l_{2,7}}R^{\dagger}_X(\phi_{2,6})X^{l_{2,6}}R^{\dagger}_Z(\phi_{2,5})Z^{l_{2,5}}$};\node at (2.9+8.65+0.2,3.5) {\scriptsize $R^{\dagger}_Z{(\phi_{1,7})}Z^{l_{1,7}}R^{\dagger}_X(\phi_{1,6})X^{l_{1,6}}R^{\dagger}_Z(\phi_{1,5})Z^{l_{1,5}}$};

\draw [fill=white] (6+8.65+0.2+0.4,-0.3) rectangle (8.25+8.65+0.2,0.3);
\draw [fill=white] (6+8.65+0.2+0.4,0.4) rectangle (8.25+8.65+0.2,1.0);
\draw [fill=white] (6+8.65+0.2+0.4,1.1) rectangle (8.25+8.65+0.2,1.7);
\draw [fill=white] (6+8.65+0.2+0.4,1.8) rectangle (8.25+8.65+0.2,2.4);
\draw [fill=white] (6+8.65+0.2+0.4,2.5) rectangle (8.25+8.65+0.2,3.1);
\draw [fill=white] (6+8.65+0.2+0.4,3.2) rectangle (8.25+8.65+0.2,3.8);

\node at (7.35+8.65+0.2,0.0) {\scriptsize $X^{l_{6,8}}$};
\node at (7.35+8.65+0.2,0.7) {\scriptsize $\textcolor{red}{Z^{l_{4,8}}}X^{l_{5,8}}$};
\node at (7.35+8.65+0.2,1.4) {\scriptsize $\textcolor{red}{Z^{l_{5,8}}}X^{l_{4,8}}$};
\node at (7.35+8.65+0.2,2.1) {\scriptsize $\textcolor{red}{Z^{l_{2,8}}}X^{l_{3,8}}$};
\node at (7.35+8.65+0.2,2.8) {\scriptsize $\textcolor{red}{Z^{l_{3,8}}}X^{l_{2,8}}$};
\node at (7.35+8.65+0.2,3.5) {\scriptsize $X^{l_{1,8}}$};

\draw [red!40,fill=red!40] (6,0.0) circle [radius=0.07cm];
\draw [red!40,fill=red!40] (6,0.7) circle [radius=0.07cm];
\draw [red!40,fill=red!40] (6,1.4) circle [radius=0.07cm];
\draw [red!40,fill=red!40] (6,2.1) circle [radius=0.07cm];
\draw [red!40,fill=red!40] (6,2.8) circle [radius=0.07cm];
\draw [red!40,fill=red!40] (6,3.5) circle [radius=0.07cm];

\draw [red!40] (6,0.0) -- (6,0.7);
\draw [red!40] (6,1.4) -- (6,2.1);
\draw [red!40] (6,2.8) -- (6,3.5);

\draw [red!40,fill=red!40] (8.65,0.0) circle [radius=0.07cm];
\draw [red!40,fill=red!40] (8.65,0.7) circle [radius=0.07cm];
\draw [red!40,fill=red!40] (8.65,1.4) circle [radius=0.07cm];
\draw [red!40,fill=red!40] (8.65,2.1) circle [radius=0.07cm];
\draw [red!40,fill=red!40] (8.65,2.8) circle [radius=0.07cm];
\draw [red!40,fill=red!40] (8.65,3.5) circle [radius=0.07cm];

\draw [red!40] (8.65,0.0) -- (8.65,0.7);
\draw [red!40] (8.65,1.4) -- (8.65,2.1);
\draw [red!40] (8.65,2.8) -- (8.65,3.5);

\draw [red!40,fill=red!40] (6+8.65+0.2,0.7) circle [radius=0.07cm];
\draw [red!40,fill=red!40] (6+8.65+0.2,1.4) circle [radius=0.07cm];
\draw [red!40,fill=red!40] (6+8.65+0.2,2.1) circle [radius=0.07cm];
\draw [red!40,fill=red!40] (6+8.65+0.2,2.8) circle [radius=0.07cm];

\draw [red!40] (6+8.65+0.2,0.7) -- (6+8.65+0.2,1.4);
\draw [red!40] (6+8.65+0.2,2.1) -- (6+8.65+0.2,2.8);

\draw [red!40,fill=red!40] (8.65+8.65+0.2,0.7) circle [radius=0.07cm];
\draw [red!40,fill=red!40] (8.65+8.65+0.2,1.4) circle [radius=0.07cm];
\draw [red!40,fill=red!40] (8.65+8.65+0.2,2.1) circle [radius=0.07cm];
\draw [red!40,fill=red!40] (8.65+8.65+0.2,2.8) circle [radius=0.07cm];

\draw [red!40] (8.65+8.65+0.2,0.7) -- (8.65+8.65+0.2,1.4);
\draw [red!40] (8.65+8.65+0.2,2.1) -- (8.65+8.65+0.2,2.8);

\draw [domain=0:180,red!40,<-,line width=0.5mm]  plot ({7.35+1.2*cos(\x)}, {3.6+0.4*sin(\x)});
\draw [domain=0:180,red!40,<-,line width=0.5mm]  plot ({7.35+8.8+1.2*cos(\x)}, {3.6+0.4*sin(\x)});

\end{tikzpicture}
\caption{\small Simplification of circuit in Figure \ref{fig:BSdev2} valid for R-traps. This circuit can be obtained setting to 0 the angles between $cZ$-gates, hence moving $cZ$-gates toward each other and cancelling them out. Before cancelling out, $cZ$-gates have to be commuted with extra Pauli-$X$ operators. This generates the red Pauli-$Z$ errors. Logical qubits are now kept disentangled all along the computation, hence errors affecting a given logical qubit can not spread and affect other logical qubits.}
\label{fig:BSdev3}
\end{figure}

\twocolumngrid

The factor $\widetilde{v}/(v+1)$ upper bounds the probability $p{(E_1,\widetilde{v})}$ of event $E_1$ happening, while the term $({7}/{8})^{\widetilde{v}-1}$ maximizes $p{(E_2|E_1,\widetilde{v})}$. Indeed, by Lemmas \ref{lem:r} and \ref{lem:c}, $p{(E_2|E_1,\widetilde{v})}$ is maximized by those combinations of phase-flips that corrupt R-traps (and are detected with probability $\geq3/4$) and do not corrupt C-traps. Considering that in the Protocol, the choice between R-traps and C-traps is made independently at random, we obtain
\begin{align*}
\textup{}p{(E_2|E_1,\widetilde{v})}\leq&\bigg(\frac{1}{2}\times\frac{3}{4}+\frac{1}{2}\times1\bigg)^{\widetilde{v}-1}=&\bigg(\frac{7}{8}\bigg)^{\widetilde{v}-1}\textrm{ ,}
\end{align*}
where the exponent $\widetilde{v}-1$ is due to the fact that if Bob affects the target, then he will affect $\widetilde{v}-1$ trap computations. Maximizing $p{(E_1\wedge E_2|\widetilde{v})}$ over $\widetilde{v}$ (which is an integer), we find Bob's highest probability of corrupting the target computation without being detected:
\begin{equation}~\label{fig:}
\varepsilon=\max_{\widetilde{v}}[p{(E_1\wedge E_2|\widetilde{v})}]=\frac{3.14}{v+1}\textrm{ , for }\widetilde{v}=7
\end{equation}

The quantity $\varepsilon$ represents the soundness of Protocol \hyperlink{pr:pr1}{1}. Indeed, the condition $p(E_1\wedge E_2|\widetilde{v})\leq\varepsilon$ implies that the output is $\varepsilon$-close to the state $p\textrm{ }{\rho}_{\textrm{out}}\otimes|\textrm{acc}\rangle\langle\textrm{acc}|+(1-p)\textrm{ }\widetilde{\rho}_{\textrm{out}}\otimes|\textrm{rej}\rangle\langle\textrm{rej}|$ (with $0\leq p\leq1$), as required by Definition \ref{def:ver}.
\end{proof} 

We now prove the Lemmas.

  \begin{proof} \textit{(Lemma \ref{lem:tw})}
We refer to \cite{DCEL09} for a proof of the Lemma.
\end{proof}

\begin{proof} \textit{(Lemma \ref{lem:r})}	The main tool that will be used in this proof is the following fact, which can be verified via an explicit calculation. Suppose that a qubit in the state $\ket{+}$ is rotated around the $Z$-axis of the Bloch sphere by a random and unknown angle ${\Phi}^{\star}\in\{0,\pi/2,2\pi/2,3\pi/2\}$. If the state is subsequently measured in the basis $\{|\pm\rangle_{}\langle\pm|\}$, the average probability of getting outcome 0 is equal to
	\begin{equation}~\label{eq:rr}
	\frac{1}{4}\sum_{{\Phi}^{\star}}\big|\textrm{ }_{}\langle+|R_Z({\Phi}^{\star})|+\rangle_{}\textrm{ }\big|^2=\frac{1}{2}
	\end{equation}
	To prove the Lemma we proceed as follows:
	\begin{itemize}
		\item[]\textbf{Step I}: We argue that the combinations of phase-flips that have the least probability of being detected are those that affect at most one logical qubit. We find the set of combinations affecting an arbitrary logical qubit $i$, and subsequently restrict our analysis to this set of combinations.
		
		\item[]\textbf{Step II}: We consider all of the combinations of phase-flips of physical qubits belonging to {a single vertical tape} $y$ and affecting logical qubit $i$. We show that these errors are detected with probability at least $1/4$.
		
		\item[]\textbf{Step III}: We consider phase-flips of physical qubits belonging to {two {neighbouring} tapes} $y$ and $y+1$ and affecting logical qubit $i$. Making exception for Type-I errors, we show that these errors are detected with probability at least $1/4$.
		
		\item[]\textbf{Step IV}: We consider phase-flips of physical qubits belonging to {two tapes $y$ and $y'$} and affecting logical qubit $i$, as well as to more than two tapes. We show that the same upper bound as above is obtained,
	\end{itemize}
	
	\begin{itemize}
		\item[]making exceptions for Type-II errors.
	\end{itemize}
	
	\noindent We now elaborate on these steps.\\
	
	\noindent \noindent\textbf{Step I. }Consider Figure \ref{fig:BSdev2}, which represents a logical computation in the presence of phase-flips affecting the physical qubits. As it can be seen in Sub-protocol \hyperlink{pr:spr1}{1.1} (steps 1.1, 1.2),  $\phi_{i,4y}=0$ for any row $i$ and tape $y$. Hence, with reference to the Figure, the only logical operators in between the $cZ$-gates are the Pauli-$X$ errors. Moving the $cZ$-gates toward each other and cancelling them out, Pauli-$X$ errors propagate to the other logical qubits and by-products of Pauli-$Z$ are produced (Figure \ref{fig:BSdev3}). 
	
	After cancelling the $cZ$-gates, the logical qubits are kept disentangled all along the rest of the computation. Hence, the by-products affecting a given logical qubit can not spread and affect the rest of the logical qubits, nor cancel out with each other. In other words, after cancelling the $cZ$-gates, logical errors remain ``{local}''. For this reason, the error with the least probability of being detected must be such that after the $cZ$-gates are cancelled out, only one logical qubit is affected by by-products (for if more than one is affected, the probability of detecting the errors can not be smaller than it is if only one qubit is affected). This enables us to restrict our analysis to the combinations of phase-flips affecting a single logical qubit $i$. 
	
	The combinations of phase-flips affecting logical qubit $i$ are contained in the set of combinations of phase-flips of red physical qubits in Figure \ref{fig:BSflip1}. Eventually, some of the combinations will also corrupt logical qubits $i-1$ and $i+1$, but we do not consider (unless explicitly specified) the effects they have on these qubits.\\
	
	\noindent\textbf{Step II. }We restrict the analysis to phase-flips of qubits belonging to a single tape $y$ and affecting logical qubit $i$. We show that they are detected by R-traps with probability larger than 1/4.
	
	Let $U_{i}^{(y)}=R_Z(\phi_3)R_X(\phi_2)R_Z(\phi_1)$ be the unitary implemented on logical qubit $i$ within tape $y$ in the absence of phase-flips. In the presence of errors, $U_{i}^{(y)}$ becomes (see Figure \ref{fig:BSflip2} for labels)
	\begin{eqnarray*}
		\begin{tabular}{lll}
			$\widetilde{U}_{i}^{(y)}$&=&\small $X^{l_4}Z^{l_5\oplus l_3}R_Z(\phi_3)X^{l_2}R_X(\phi_2)Z^{l_1}R_Z(\phi_1)$\cr
		\end{tabular}
	\end{eqnarray*}
	\noindent Moving the errors toward the beginning of the circuit, we obtain
	\begin{eqnarray*}
		\begin{tabular}{lll}
			$\widetilde{U}_{i}^{(y)}$&=& $R_Z\big((-1)^{l_4}\phi_3\big)R_X\big((-1)^{l_3\oplus l_5}\phi_2\big)R_Z\big((-1)^{l_4\oplus l_2}\phi_1\big)$\\ \cr
			&&$\cdot\textrm{ } X^{l_4\oplus l_2}Z^{l_5\oplus l_3\oplus l_1}$\cr
		\end{tabular}
	\end{eqnarray*}
	Thus, when the by-products are moved toward the beginning of the tape, the sign of some angles characterizing $U_{i}^{(y)}$ might be flipped. To understand which combinations of phase-flips have a non-trivial effect on $U_{i}^{(y)}$, we proceed by direct calculation: we consider all of the possible combinations and see if $\widetilde{U}_{i}^{(y)}=U_{i}^{(y)}$, keeping in mind that ${U}_{i}^{(y)}$ will either be an $XY$-plane rotation, a $ZY$-plane rotation or a Hadamard.
	
	The explicit calculations are illustrated in Table \hyperlink{tab:tab1}{1}. As it can be seen, the only error that never corrupts $ {U}_{i}^{(y)}$ corresponds to $l_3=l_5=1$ and the other $l^{}$ equal to 0. However, $l_5=1$ implies a ``$l_4$ Pauli-$X$ error'' affecting logical qubit $i+1$ (or $i-1$), which has a non-trivial effect on $ {U}_{i+1}^{(y)}$ (or $ \widetilde{U}_{i-1}^{(y)}$). All of the remaining errors affect ${U}_{i}^{(y)}$ non-trivially whenever it implements at least two gates from the set $\{H,R_Z,R_X\}$. This is crucial, hence we present it with some examples. As can be seen, error $l_1=l_2=1$ and the other $l$ equal to 0 always corrupts ${U}_{i}^{(y)}$, both in the case where it implements a rotation or a Hadamard. On the contrary, error $l_1=l_3=1$ and the other $l$ equal to 0 has a non-trivial effect on ${U}_{i}^{(y)}$ only if implements a Hadamard or a $R_X$-gate, but does not affect $R_Z$-gates. No error corrupts ${U}_{i}^{(y)}$ in less than two cases (apart from the already mentioned $l_3=l_5=1$ and the other $l^{}$ equal to 0). 
	
	We can now lower bound the probability of detecting single-tape phase-flips by 1/4. To compute the bound, we first notice that for any combination of phase-flips affecting solely tape $y$, $\widetilde{U}_{i}^{(y)}\neq U_{i}^{(y)}$ with probability at least 1/2. This can be seen considering that (i) as pointed out above, each combination of phase-flips affects non-trivially at least two logical unitaries out of three, and (ii) $ {U}_{i}^{(y)}$ is a Hadamard with probability 1/2 and a rotation with probability 1/2. Thus, errors that affect $R_Z$-gates and $R_X$-gates, but do not affect Hadamards (e.g. $l_1=l_4=1$ and the other $l=0$) have the least probability of corrupting the R-trap, and this probability is equal to 1/2. Next, we use equation \ref{eq:rr} to show that whenever $\widetilde{U}_{i}^{(y)}\neq U_{i}^{(y)}$, the error is detected with probability $\geq$1/2. To see this, consider the possible ways phase-flips can 
\begin{center}
\setlength{\extrarowheight}{0.2cm}
\begin{tabular}{|c|c|c|c|}
\hline
\hypertarget{tab:tab1}{Flipped qubits}&$R_Z(\phi_1+\phi_3)$&$R_X(\phi_2)$&$H$\cr
\hline
1&$R_Z(\phi_1+\phi_3)Z$&$R_X(\phi_2)Z$&$HZ$\cr
2&$R_Z(-\phi_1+\phi_3)X$&$R_X(\phi_2)X$&$HXZ$\cr
3&$R_Z(\phi_1+\phi_3)Z$&$R_X(-\phi_2)Z$&$HX$\cr
4&$R_Z(-\phi_1-\phi_3)X$&$R_X(\phi_2)X$&$HZ$\cr
5&$R_Z(\phi_1+\phi_3)Z$&$R_X(-\phi_2)Z$&$HX$\cr
1,2&$R_Z(-\phi_1+\phi_3)XZ$&$R_X(\phi_2)XZ$&$HX$\cr
1,3&&$R_X(-\phi_2)$&$HXZ$\cr
1,4&$R_Z(-\phi_1-\phi_3)XZ$&$R_X(\phi_2)XZ$&\cr
1,5&&$R_X(-\phi_2)$&$HXZ$\cr
2,3&$R_Z(-\phi_1+\phi_3)XZ$&$R_X(-\phi_2)XZ$&$HZ$\cr
2,4&$R_Z(\phi_1-\phi_3)$&&$HX$\cr
2,5&$R_Z(-\phi_1+\phi_3)XZ$&$R_X(-\phi_2)XZ$&$HZ$\cr
3,4&$R_Z(-\phi_1-\phi_3)XZ$&$R_X(-\phi_2)XZ$&$HXZ$\cr
3,5&&&\cr
4,5&$R_Z(-\phi_1-\phi_3)XZ$&$R_X(-\phi_2)XZ$&$HXZ$\cr
1,2,3&$R_Z(-\phi_1+\phi_3)X$&$R_X(-\phi_2)X$&\cr
1,2,4&$R_Z(\phi_1-\phi_3)Z$&$R_X(\phi_2)Z$&$HX$\cr 
1,2,5&$R_Z(-\phi_1+\phi_3)X$&$R_X(-\phi_2)X$&\cr
1,3,4&$R_Z(-\phi_1-\phi_3)X$&$R_X(-\phi_2)X$&$HX$\cr
1,3,5&$R_Z(\phi_1+\phi_3)Z$&$R_X(\phi_2)Z$&$HZ$\cr
1,4,5&$R_Z(-\phi_1-\phi_3)X$&$R_X(-\phi_2)X$&$HX$\cr
2,3,4&$R_Z(\phi_1-\phi_3)Z$&$R_X(-\phi_2)Z$&\cr
2,3,5&$R_Z(-\phi_1+\phi_3)X$&$R_X(\phi_2)X$&$HXZ$\cr
2,4,5&$R_Z(\phi_1-\phi_3)Z$&$R_X(-\phi_2)Z$&\cr
3,4,5&$R_Z(-\phi_1-\phi_3)X$&$R_X(\phi_2)X$&$HZ$\cr
1,2,3,4&$R_Z(\phi_1-\phi_3)$&$R_X(-\phi_2)$&$HZ$\cr
1,2,3,5&$R_Z(-\phi_1+\phi_3)XZ$&$R_X(\phi_2)XZ$&$HX$\cr
1,2,4,5&$R_Z(\phi_1-\phi_3)$&$R_X(-\phi_2)$&$HZ$\cr
1,3,4,5&$R_Z(-\phi_1-\phi_3)XZ$&$R_X(\phi_2)XZ$&\cr
2,3,4,5&$R_Z(\phi_1-\phi_3)$&&$HX$\cr
1,2,3,4,5&$R_Z(\phi_1-\phi_3)Z$&$R_X(\phi_2)Z$&$HXZ$\cr &&&\cr
\hline
\end{tabular}
\end{center}
\vspace{0.4cm}
\noindent Table {1}:\begin{small}
The effects of phase-flips on a single tape $y$ after the by-products are moved toward the beginning of the tape (i.e. ``to the right'' of the unitary implemented within the tape). White spaces correspond to the cases where the by-products have a trivial effect on the unitary implemented within tape $y$ (i.e. $\widetilde{U}_i^{(y)}={U}_i^{(y)}$).  We refer to Figure \ref{fig:BSflip2} for labels regarding phase-flipped qubits.
\end{small}

\newpage
\begin{figure}[H]
\centering
\begin{tikzpicture}

\draw [line width=0.5mm,dashed] (0.0-0.25,1.) -- (0.0-0.25,3.0);
\draw [line width=0.5mm,dashed] (0.0+1.75,1) -- (0.0+1.75,3.0);
\draw [line width=0.5mm,dashed] (2.0+1.75,1) -- (2.0+1.75,3.0);
\draw [line width=0.5mm,dashed] (4.0+1.75,1) -- (4.0+1.75,3.0);

\draw (0.0,1.5) -- (6.0,1.5);
\draw [dashed] (-0.5,1.5) -- (0.0,1.5);
\draw [dashed] (6.0,1.5) -- (6.5,1.5);
\draw (0.0,2.0) -- (6.0,2.0);
\draw [dashed] (-0.5,2.0) -- (0.0,2.0);
\draw [dashed] (6.0,2.0) -- (6.5,2.0);
\draw (0.0,2.5) -- (6.0,2.5);
\draw [dashed] (-0.5,2.5) -- (0.0,2.5);
\draw [dashed] (6.0,2.5) -- (6.5,2.5);

\draw (1.0,2.0) -- (1.0,2.5);
\draw (2.0,2.0) -- (2.0,2.5);

\draw (3.0,1.5) -- (3.0,2.0);
\draw (4.0,1.5) -- (4.0,2.0);

\draw (5.0,2.0) -- (5.0,2.5);
\draw (6.0,2.0) -- (6.0,2.5);

\draw [fill=white] (0.0,2.5) circle [radius=0.1];
\draw [fill=white] (0.5,2.5) circle [radius=0.1];
\draw [fill=white] (1.0,2.5) circle [radius=0.1];
\draw [fill=red!40] (1.5,2.5) circle [radius=0.1];
\draw [fill=white] (2.0,2.5) circle [radius=0.1];
\draw [fill=white] (2.5,2.5) circle [radius=0.1];
\draw [fill=white] (3.0,2.5) circle [radius=0.1];
\draw [fill=white] (3.5,2.5) circle [radius=0.1];
\draw [fill=white] (4.0,2.5) circle [radius=0.1];
\draw [fill=white] (4.5,2.5) circle [radius=0.1];
\draw [fill=white] (5.0,2.5) circle [radius=0.1];
\draw [fill=red!40] (5.5,2.5) circle [radius=0.1];
\draw [fill=white] (6.0,2.5) circle [radius=0.1];

\draw [fill=red!40] (0.0,2.0) circle [radius=0.1];
\draw [fill=red!40] (0.5,2.0) circle [radius=0.1];
\draw [fill=red!40] (1.0,2.0) circle [radius=0.1];
\draw [fill=red!40] (1.5,2.0) circle [radius=0.1];
\draw [fill=red!40] (2.0,2.0) circle [radius=0.1];
\draw [fill=red!40] (2.5,2.0) circle [radius=0.1];
\draw [fill=red!40] (3.0,2.0) circle [radius=0.1];
\draw [fill=red!40] (3.5,2.0) circle [radius=0.1];
\draw [fill=red!40] (4.0,2.0) circle [radius=0.1];
\draw [fill=red!40] (4.5,2.0) circle [radius=0.1];
\draw [fill=red!40] (5.0,2.0) circle [radius=0.1];
\draw [fill=red!40] (5.5,2.0) circle [radius=0.1];
\draw [fill=red!40] (6.0,2.0) circle [radius=0.1];

\draw [fill=white] (0.0,1.5) circle [radius=0.1];
\draw [fill=white] (0.5,1.5) circle [radius=0.1];
\draw [fill=white] (1.0,1.5) circle [radius=0.1];
\draw [fill=white] (1.5,1.5) circle [radius=0.1];
\draw [fill=white] (2.0,1.5) circle [radius=0.1];
\draw [fill=white] (2.5,1.5) circle [radius=0.1];
\draw [fill=white] (3.0,1.5) circle [radius=0.1];
\draw [fill=red!40] (3.5,1.5) circle [radius=0.1];
\draw [fill=white] (4.0,1.5) circle [radius=0.1];
\draw [fill=white] (4.5,1.5) circle [radius=0.1];
\draw [fill=white] (5.0,1.5) circle [radius=0.1];
\draw [fill=white] (5.5,1.5) circle [radius=0.1];
\draw [fill=white] (6.0,1.5) circle [radius=0.1];

\node at (-1.0,1.5) {$i+1$};
\node at (-1.0,2.0) {$i$};
\node at (-1.0,2.5) {$i-1$};

\node at (0.65,3.2) {$y-1$};
\node at (2.65,3.2) {$y$};
\node at (4.65,3.2) {$y+1$};
%
%

\end{tikzpicture}
\caption{\small Phase-flips affecting the overall unitary acting on logical qubit $i$. Phase-flips of red physical qubits generate by-products of Pauli-$Z$ or Pauli-$X$ that affect logical qubit $i$ along the computation.}
\label{fig:BSflip1}
\end{figure}
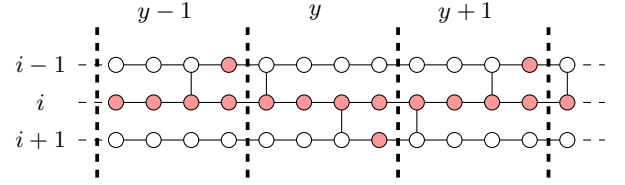

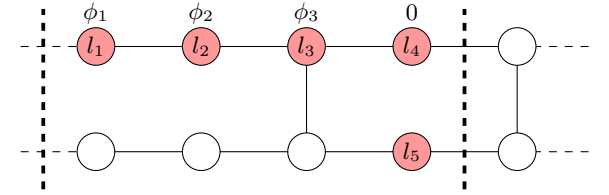
\begin{figure}[H]
\centering
\begin{tikzpicture}

\draw [line width=0.5mm,dashed] (4.9,-0.5) -- (4.9,1.9);
\draw [line width=0.5mm,dashed] (-0.7,-0.5) -- (-0.7,1.9);

\draw [dashed] (-1,1.4) -- (0.0,1.4);
\draw [dashed] (5.6,1.4) -- (6.6,1.4);
\draw (0.0,1.4) -- (5.6,1.4);
\draw [dashed] (-1,0.0) -- (0.0,0.0);
\draw [dashed] (5.6,0.0) -- (6.6,0.0);
\draw (0.0,0.0) -- (5.6,0.0);
\draw (2.8,0.0) -- (2.8,1.4);
\draw (5.6,0.0) -- (5.6,1.4);

\draw [fill=red!40] (0.0,1.4) circle [radius=0.25cm];
\draw [fill=red!40] (1.4,1.4) circle [radius=0.25cm];
\draw [fill=red!40] (2.8,1.4) circle [radius=0.25cm];
\draw [fill=red!40] (4.2,1.4) circle [radius=0.25cm];
\draw [fill=white] (5.6,1.4) circle [radius=0.25cm];
\draw [fill=white] (0.0,0.0) circle [radius=0.25cm];
\draw [fill=white] (1.4,0.0) circle [radius=0.25cm];
\draw [fill=white] (2.8,0.0) circle [radius=0.25cm];
\draw [fill=red!40] (4.2,0.0) circle [radius=0.25cm];
\draw [fill=white] (5.6,0.0) circle [radius=0.25cm];

\node at (0.0,1.4) {$l_1$};
\node at (1.4,1.4) {$l_2$};
\node at (2.8,1.4) {$l_3$};
\node at (4.2,1.4) {$l_4$};
\node at (4.2,0.0) {$l_5$};

\node at (0.0,1.85) {$\phi_1$};
\node at (1.4,1.85) {$\phi_2$};
\node at (2.8,1.85) {$\phi_3$};
\node at (4.2,1.85) {$0$};

\end{tikzpicture}
\caption{\small Phase-flips affecting the unitary implemented on logical qubit $i$ within tape $y$ are described by the set of parameters $l_q$, $q=1,..,5$. $l_q$ can either be 0 (no flip) or 1 (flip).}
\label{fig:BSflip2}
\end{figure}

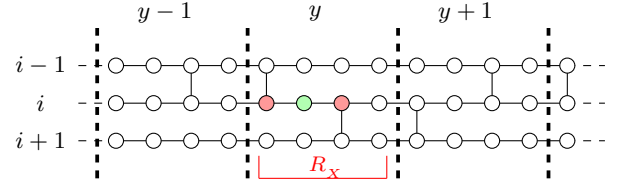
\begin{figure}[H]
\centering
\begin{tikzpicture}

\draw [line width=0.5mm,dashed] (0.0-0.25,1.) -- (0.0-0.25,3.0);
\draw [line width=0.5mm,dashed] (0.0+1.75,1) -- (0.0+1.75,3.0);
\draw [line width=0.5mm,dashed] (2.0+1.75,1) -- (2.0+1.75,3.0);
\draw [line width=0.5mm,dashed] (4.0+1.75,1) -- (4.0+1.75,3.0);

\draw (0.0,1.5) -- (6.0,1.5);
\draw [dashed] (-0.5,1.5) -- (0.0,1.5);
\draw [dashed] (6.0,1.5) -- (6.5,1.5);
\draw (0.0,2.0) -- (6.0,2.0);
\draw [dashed] (-0.5,2.0) -- (0.0,2.0);
\draw [dashed] (6.0,2.0) -- (6.5,2.0);
\draw (0.0,2.5) -- (6.0,2.5);
\draw [dashed] (-0.5,2.5) -- (0.0,2.5);
\draw [dashed] (6.0,2.5) -- (6.5,2.5);

\draw (1.0,2.0) -- (1.0,2.5);
\draw (2.0,2.0) -- (2.0,2.5);

\draw (3.0,1.5) -- (3.0,2.0);
\draw (4.0,1.5) -- (4.0,2.0);

\draw (5.0,2.0) -- (5.0,2.5);
\draw (6.0,2.0) -- (6.0,2.5);

\draw [fill=white] (0.0,2.5) circle [radius=0.1];
\draw [fill=white] (0.5,2.5) circle [radius=0.1];
\draw [fill=white] (1.0,2.5) circle [radius=0.1];
\draw [fill=white] (1.5,2.5) circle [radius=0.1];
\draw [fill=white] (2.0,2.5) circle [radius=0.1];
\draw [fill=white] (2.5,2.5) circle [radius=0.1];
\draw [fill=white] (3.0,2.5) circle [radius=0.1];
\draw [fill=white] (3.5,2.5) circle [radius=0.1];
\draw [fill=white] (4.0,2.5) circle [radius=0.1];
\draw [fill=white] (4.5,2.5) circle [radius=0.1];
\draw [fill=white] (5.0,2.5) circle [radius=0.1];
\draw [fill=white] (5.5,2.5) circle [radius=0.1];
\draw [fill=white] (6.0,2.5) circle [radius=0.1];

\draw [fill=white] (0.0,2.0) circle [radius=0.1];
\draw [fill=white] (0.5,2.0) circle [radius=0.1];
\draw [fill=white] (1.0,2.0) circle [radius=0.1];
\draw [fill=white] (1.5,2.0) circle [radius=0.1];
\draw [fill=red!40] (2.0,2.0) circle [radius=0.1];
\draw [fill=green!30] (2.5,2.0) circle [radius=0.1];
\draw [fill=red!40] (3.0,2.0) circle [radius=0.1];
\draw [fill=white] (3.5,2.0) circle [radius=0.1];
\draw [fill=white] (4.0,2.0) circle [radius=0.1];
\draw [fill=white] (4.5,2.0) circle [radius=0.1];
\draw [fill=white] (5.0,2.0) circle [radius=0.1];
\draw [fill=white] (5.5,2.0) circle [radius=0.1];
\draw [fill=white] (6.0,2.0) circle [radius=0.1];

\draw [fill=white] (0.0,1.5) circle [radius=0.1];
\draw [fill=white] (0.5,1.5) circle [radius=0.1];
\draw [fill=white] (1.0,1.5) circle [radius=0.1];
\draw [fill=white] (1.5,1.5) circle [radius=0.1];
\draw [fill=white] (2.0,1.5) circle [radius=0.1];
\draw [fill=white] (2.5,1.5) circle [radius=0.1];
\draw [fill=white] (3.0,1.5) circle [radius=0.1];
\draw [fill=white] (3.5,1.5) circle [radius=0.1];
\draw [fill=white] (4.0,1.5) circle [radius=0.1];
\draw [fill=white] (4.5,1.5) circle [radius=0.1];
\draw [fill=white] (5.0,1.5) circle [radius=0.1];
\draw [fill=white] (5.5,1.5) circle [radius=0.1];
\draw [fill=white] (6.0,1.5) circle [radius=0.1];

\node at (-1.0,1.5) {$i+1$};
\node at (-1.0,2.0) {$i$};
\node at (-1.0,2.5) {$i-1$};

\node at (0.65,3.2) {$y-1$};
\node at (2.65,3.2) {$y$};
\node at (4.65,3.2) {$y+1$};

\draw [red] (1.9,1.0) -- (3.6,1.0);
\draw [red] (1.9,1.0) -- (1.9,1.3);
\draw [red] (3.6,1.0) -- (3.6,1.3);
\node at (2.8,1.15) {\footnotesize \textcolor{red}{$R_{_X}$}};

%
%

\end{tikzpicture}
\caption{\small Tape $y$ is used to implement a $R_X$-gate on logical qubit $i$. If red qubits are phase-flipped, rotation angle $\phi_{i,4y-2}$ associated to green physical qubit is mapped into $-\phi_{i,4y-2}$. Thus, instead of being output in some expected state $R_Z({\Phi}_i)\ket{+}$, logical qubit $i$ is output in the state $R_Z({\Phi}^{\star}_i)R_Z({\Phi}_i)\ket{+}$, where $\Phi^{\star}_i=-2\phi_{i,4y-2}$ is a random angle in $\{0,\pi/2,\pi,3\pi/2\}$. By equation \ref{eq:rr}, the probability of detecting the phase-flips is 1/2.}
\label{fig:BSflip3}
\end{figure}

\begin{figure}[H]
\centering
\begin{tikzpicture}

\draw [line width=0.5mm,dashed] (0.0-0.25,1.) -- (0.0-0.25,3.0);
\draw [line width=0.5mm,dashed] (0.0+1.75,1) -- (0.0+1.75,3.0);
\draw [line width=0.5mm,dashed] (2.0+1.75,1) -- (2.0+1.75,3.0);
\draw [line width=0.5mm,dashed] (4.0+1.75,1) -- (4.0+1.75,3.0);

\draw (0.0,1.5) -- (6.0,1.5);
\draw [dashed] (-0.5,1.5) -- (0.0,1.5);
\draw [dashed] (6.0,1.5) -- (6.5,1.5);
\draw (0.0,2.0) -- (6.0,2.0);
\draw [dashed] (-0.5,2.0) -- (0.0,2.0);
\draw [dashed] (6.0,2.0) -- (6.5,2.0);
\draw (0.0,2.5) -- (6.0,2.5);
\draw [dashed] (-0.5,2.5) -- (0.0,2.5);
\draw [dashed] (6.0,2.5) -- (6.5,2.5);

\draw (1.0,2.0) -- (1.0,2.5);
\draw (2.0,2.0) -- (2.0,2.5);

\draw (3.0,1.5) -- (3.0,2.0);
\draw (4.0,1.5) -- (4.0,2.0);

\draw (5.0,2.0) -- (5.0,2.5);
\draw (6.0,2.0) -- (6.0,2.5);

\draw [fill=white] (0.0,2.5) circle [radius=0.1];
\draw [fill=white] (0.5,2.5) circle [radius=0.1];
\draw [fill=white] (1.0,2.5) circle [radius=0.1];
\draw [fill=white] (1.5,2.5) circle [radius=0.1];
\draw [fill=white] (2.0,2.5) circle [radius=0.1];
\draw [fill=white] (2.5,2.5) circle [radius=0.1];
\draw [fill=white] (3.0,2.5) circle [radius=0.1];
\draw [fill=white] (3.5,2.5) circle [radius=0.1];
\draw [fill=white] (4.0,2.5) circle [radius=0.1];
\draw [fill=white] (4.5,2.5) circle [radius=0.1];
\draw [fill=white] (5.0,2.5) circle [radius=0.1];
\draw [fill=white] (5.5,2.5) circle [radius=0.1];
\draw [fill=white] (6.0,2.5) circle [radius=0.1];

\draw [fill=white] (0.0,2.0) circle [radius=0.1];
\draw [fill=green!30] (0.5,2.0) circle [radius=0.1];
\draw [fill=white] (1.0,2.0) circle [radius=0.1];
\draw [fill=white] (1.5,2.0) circle [radius=0.1];
\draw [fill=white] (2.0,2.0) circle [radius=0.1];
\draw [fill=white] (2.5,2.0) circle [radius=0.1];
\draw [fill=white] (3.0,2.0) circle [radius=0.1];
\draw [fill=white] (3.5,2.0) circle [radius=0.1];
\draw [fill=green!30] (4.0,2.0) circle [radius=0.1];
\draw [fill=white] (4.5,2.0) circle [radius=0.1];
\draw [fill=green!30] (5.0,2.0) circle [radius=0.1];
\draw [fill=white] (5.5,2.0) circle [radius=0.1];
\draw [fill=white] (6.0,2.0) circle [radius=0.1];

\draw [fill=white] (0.0,1.5) circle [radius=0.1];
\draw [fill=white] (0.5,1.5) circle [radius=0.1];
\draw [fill=white] (1.0,1.5) circle [radius=0.1];
\draw [fill=white] (1.5,1.5) circle [radius=0.1];
\draw [fill=white] (2.0,1.5) circle [radius=0.1];
\draw [fill=white] (2.5,1.5) circle [radius=0.1];
\draw [fill=white] (3.0,1.5) circle [radius=0.1];
\draw [fill=white] (3.5,1.5) circle [radius=0.1];
\draw [fill=white] (4.0,1.5) circle [radius=0.1];
\draw [fill=white] (4.5,1.5) circle [radius=0.1];
\draw [fill=white] (5.0,1.5) circle [radius=0.1];
\draw [fill=white] (5.5,1.5) circle [radius=0.1];
\draw [fill=white] (6.0,1.5) circle [radius=0.1];

\node at (-1.0,1.5) {$i+1$};
\node at (-1.0,2.0) {$i$};
\node at (-1.0,2.5) {$i-1$};

\node at (0.65,3.2) {$y-3$};
\node at (2.65,3.2) {$y-2$};
\node at (4.65,3.2) {$y-1$};

\draw [red] (1.9-2,1.0) -- (3.6-2,1.0);
\draw [red] (1.9-2,1.0) -- (1.9-2,1.3);
\draw [red] (3.6-2,1.0) -- (3.6-2,1.3);
\node at (2.8-2,1.15) {\footnotesize \textcolor{red}{$R_{_X}$}};

\draw [red] (1.9,1.0) -- (3.6,1.0);
\draw [red] (1.9,1.0) -- (1.9,1.3);
\draw [red] (3.6,1.0) -- (3.6,1.3);
\node at (2.8,1.15) {\footnotesize \textcolor{red}{$H$}};

\draw [red] (1.9+2,1.0) -- (3.6+2,1.0);
\draw [red] (1.9+2,1.0) -- (1.9+2,1.3);
\draw [red] (3.6+2,1.0) -- (3.6+2,1.3);
\node at (2.8+2,1.15) {\footnotesize \textcolor{red}{$R_{_Z}$}};

\draw [->,red,line width=0.35mm] (7.3,2.1) -- (6.2,2.1);
\node at (6.9,2.27) {\small {Pauli-${X}$}};

\end{tikzpicture}
\caption{\small Pushing a Pauli-$X$ by-product from tape $y$ toward the beginning of the circuit flips the sign of angles associated to green physical qubit. Supposing that angles of green qubits sum up to $\widetilde{\Phi}_i\in\{0,\pi/4,..,7\pi/4\}$, and that the overall rotation angle is equal to ${\Phi}_i\in\{0,\pi/4,..,7\pi/4\}$, logical qubit $i$ is here output in the state $R_Z({\Phi_i}^{\star})R_Z({\Phi}_i)\ket{+}$, where ${\Phi}^{\star}_i=-2\widetilde{\Phi}_i$ is a random angle. By equation \ref{eq:rr}, the probability of detecting the presence of the by-product is 1/2.}
\label{fig:BSflip4}
\end{figure}
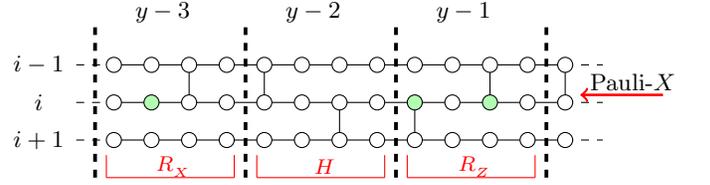

\newpage

\noindent corrupt $ U_{i}^{(y)}$ (remember that the overall operation on each logical qubit is an $R_Z$-gate):

\begin{itemize}
	 \item[(i)] The sign of a set of angles $\{\phi_{i,j}\}$ is flipped. Using $-\phi_{i,j}=\phi_{i,j}-2\phi_{i,j}$ for any $i,j$, it can be seen that logical qubit $i$ is subject to the honest overall rotation modulo a rotation by a \textit{random} angle $\Phi_{i}^{\star}=2\sum_{i,j}\phi_{i,j}\in\{0,\pi/2,\pi,3\pi/2\}$. Thus, equation \ref{eq:rr} guarantees that Bob's success rate is equal to 1/2. An example is provided in Figure \ref{fig:BSflip3}.
 	\item[(ii)] The angles remain unchanged and a by-product of Pauli-$X$ is produced. Commuting the Pauli-$X$ error all the way through to the beginning of the circuit, the sign of a set of angles $\{\phi_{i,j}\}$ is indeed flipped. Thus, by equation \ref{eq:rr}, the probability of detecting the presence of the by-product is 1/2. An example is provided in Figure \ref{fig:BSflip4}.
\item[(iii)] The angles remain unchanged and a by-product of Pauli-$Y$ is produced. Commuting the Pauli-$Y$ error all the way through to the beginning of the circuit, the sign of some set of angles $\{\phi_{i,j}\}$ is indeed flipped. Thus, by equation \ref{eq:rr}, the probability of detecting the presence of the by-product is 1/2.
\item[(iv)] The angles remain unchanged and a by-product of Pauli-$Z$ is produced. Since $R_Z$-gates commute with Pauli-$Z$, the error can be moved to the beginning of the circuit. As a consequence, the input qubit is phase-flipped and the error is detected with probability 1.
\end{itemize}
Overall, any combination of phase-flips affecting physical qubits within a single tape is thus detected with probability $\geq$1/4.\\

\noindent\textbf{Step III: }A similar strategy can be used to prove that phase-flips of physical qubits belonging to different tapes yield either errors that do not affect R-traps (we will deal with them later on using C-traps), or that are detected with probability at least 1/4. 

Consider a combination of phase-flips affecting logical qubit $i$ within two neighbouring tapes $y$ and $y+1$. The unitary implemented within the tapes is equal to 
\begin{eqnarray*}
\begin{tabular}{lll}
&\small $X^{l'_4}Z^{l'_5\oplus l'_3}R_Z(\phi'_3)X^{l'_2}R_X(\phi'_2)Z^{l'_1}R_Z(\phi'_1)$&$\cdot$\\ \cr
$\cdot$&\small $X^{l_4}Z^{l_5\oplus l_3}R_Z(\phi_3)X^{l_2}R_X(\phi_2)Z^{l_1}R_Z(\phi_1)$\cr
\end{tabular}
\end{eqnarray*}
(here, $\phi$ and $l$ label angles and phase-flips of tape $y$, while $\phi'$ and $l'$ label those of tape $y+1$). First, we move by-products affecting tape $y+1$ at the beginning of the tape. The above operator becomes
\begin{eqnarray*}
\begin{tabular}{ccc}
\small $R_Z\big((-1)^{l_4'}\phi'_3\big)R_X\big((-1)^{l'_3\oplus l_5'}\phi'_2\big)R_Z\big((-1)^{l'_2\oplus l_4'}\phi'_1\big)$&\\ \cr
$\cdot\textrm{ }$\small $X^{l'_4\oplus l'_2\oplus l_4}Z^{l'_5\oplus l'_3\oplus l'_1\oplus l_5\oplus l_3}R_Z(\phi_3)X^{l_2}R_X(\phi_2)Z^{l_1}R_Z(\phi_1)$\cr
\end{tabular}
\end{eqnarray*}
Next, we notice (Table \hyperlink{tab:tab1}{1}) that all of the combinations of phase-flips affecting tape $y+1$ generate the same by-products between the two logical rotations. On the other

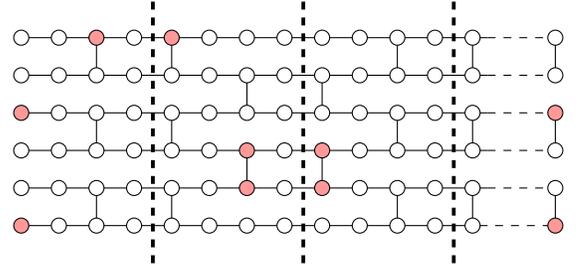
\begin{figure}[H]
\centering
\begin{tikzpicture}
\draw [line width=0.5mm,dashed] (0.0+1.75,-0.5) -- (0.0+1.75,3.0);
\draw [line width=0.5mm,dashed] (2.0+1.75,-0.5) -- (2.0+1.75,3.0);
\draw [line width=0.5mm,dashed] (4.0+1.75,-0.5) -- (4.0+1.75,3.0);

\draw (0.0,0.0) -- (6.2,0.0);
\draw (0.0,0.5) -- (6.2,0.5);
\draw (0.0,1.0) -- (6.2,1.0);
\draw (0.0,1.5) -- (6.2,1.5);
\draw (0.0,2.0) -- (6.2,2.0);
\draw (0.0,2.5) -- (6.2,2.5);

\draw [dashed] (6.3,0.0) -- (7.1,0.0);
\draw [dashed] (6.2,0.5) -- (7.1,0.5);
\draw [dashed] (6.2,1.0) -- (7.1,1.0);
\draw [dashed] (6.2,1.5) -- (7.1,1.5);
\draw [dashed] (6.2,2.0) -- (7.1,2.0);
\draw [dashed] (6.2,2.5) -- (7.1,2.5);

\draw (1.0,0.0) -- (1.0,0.5);
\draw (1.0,1.0) -- (1.0,1.5);
\draw (1.0,2.0) -- (1.0,2.5);
\draw (2.0,0.0) -- (2.0,0.5);
\draw (2.0,1.0) -- (2.0,1.5);
\draw (2.0,2.0) -- (2.0,2.5);

\draw (3.0,0.5) -- (3.0,1.0);
\draw (3.0,1.5) -- (3.0,2.0);
\draw (4.0,0.5) -- (4.0,1.0);
\draw (4.0,1.5) -- (4.0,2.0);

\draw (5.0,0.0) -- (5.0,0.5);
\draw (5.0,1.0) -- (5.0,1.5);
\draw (5.0,2.0) -- (5.0,2.5);
\draw (6.0,0.0) -- (6.0,0.5);
\draw (6.0,1.0) -- (6.0,1.5);
\draw (6.0,2.0) -- (6.0,2.5);

\draw (7.1,0.0) -- (7.1,0.5);
\draw (7.1,1.0) -- (7.1,1.5);
\draw (7.1,2.0) -- (7.1,2.5);

\draw [fill=white] (0.0,2.5) circle [radius=0.1];
\draw [fill=white] (0.5,2.5) circle [radius=0.1];
\draw [fill=red!40] (1.0,2.5) circle [radius=0.1];
\draw [fill=white] (1.5,2.5) circle [radius=0.1];
\draw [fill=red!40] (2.0,2.5) circle [radius=0.1];
\draw [fill=white] (2.5,2.5) circle [radius=0.1];
\draw [fill=white] (3.0,2.5) circle [radius=0.1];
\draw [fill=white] (3.5,2.5) circle [radius=0.1];
\draw [fill=white] (4.0,2.5) circle [radius=0.1];
\draw [fill=white] (4.5,2.5) circle [radius=0.1];
\draw [fill=white] (5.0,2.5) circle [radius=0.1];
\draw [fill=white] (5.5,2.5) circle [radius=0.1];
\draw [fill=white] (6.0,2.5) circle [radius=0.1];
\draw [fill=white] (7.1,2.5) circle [radius=0.1];

\draw [fill=white] (0.0,2.0) circle [radius=0.1];
\draw [fill=white] (0.5,2.0) circle [radius=0.1];
\draw [fill=white] (1.0,2.0) circle [radius=0.1];
\draw [fill=white] (1.5,2.0) circle [radius=0.1];
\draw [fill=white] (2.0,2.0) circle [radius=0.1];
\draw [fill=white] (2.5,2.0) circle [radius=0.1];
\draw [fill=white] (3.0,2.0) circle [radius=0.1];
\draw [fill=white] (3.5,2.0) circle [radius=0.1];
\draw [fill=white] (4.0,2.0) circle [radius=0.1];
\draw [fill=white] (4.5,2.0) circle [radius=0.1];
\draw [fill=white] (5.0,2.0) circle [radius=0.1];
\draw [fill=white] (5.5,2.0) circle [radius=0.1];
\draw [fill=white] (6.0,2.0) circle [radius=0.1];
\draw [fill=white] (7.1,2.0) circle [radius=0.1];

\draw [fill=red!40] (0.0,1.5) circle [radius=0.1];
\draw [fill=white] (0.5,1.5) circle [radius=0.1];
\draw [fill=white] (1.0,1.5) circle [radius=0.1];
\draw [fill=white] (1.5,1.5) circle [radius=0.1];
\draw [fill=white] (2.0,1.5) circle [radius=0.1];
\draw [fill=white] (2.5,1.5) circle [radius=0.1];
\draw [fill=white] (3.0,1.5) circle [radius=0.1];
\draw [fill=white] (3.5,1.5) circle [radius=0.1];
\draw [fill=white] (4.0,1.5) circle [radius=0.1];
\draw [fill=white] (4.5,1.5) circle [radius=0.1];
\draw [fill=white] (5.0,1.5) circle [radius=0.1];
\draw [fill=white] (5.5,1.5) circle [radius=0.1];
\draw [fill=white] (6.0,1.5) circle [radius=0.1];
\draw [fill=red!40] (7.1,1.5) circle [radius=0.1];

\draw [fill=white] (0.0,1.0) circle [radius=0.1];
\draw [fill=white] (0.5,1.0) circle [radius=0.1];
\draw [fill=white] (1.0,1.0) circle [radius=0.1];
\draw [fill=white] (1.5,1.0) circle [radius=0.1];
\draw [fill=white] (2.0,1.0) circle [radius=0.1];
\draw [fill=white] (2.5,1.0) circle [radius=0.1];
\draw [fill=red!40] (3.0,1.0) circle [radius=0.1];
\draw [fill=white] (3.5,1.0) circle [radius=0.1];
\draw [fill=red!40] (4.0,1.0) circle [radius=0.1];
\draw [fill=white] (4.5,1.0) circle [radius=0.1];
\draw [fill=white] (5.0,1.0) circle [radius=0.1];
\draw [fill=white] (5.5,1.0) circle [radius=0.1];
\draw [fill=white] (6.0,1.0) circle [radius=0.1];
\draw [fill=white] (7.1,1.0) circle [radius=0.1];

\draw [fill=white] (0.0,0.5) circle [radius=0.1];
\draw [fill=white] (0.5,0.5) circle [radius=0.1];
\draw [fill=white] (1.0,0.5) circle [radius=0.1];
\draw [fill=white] (1.5,0.5) circle [radius=0.1];
\draw [fill=white] (2.0,0.5) circle [radius=0.1];
\draw [fill=white] (2.5,0.5) circle [radius=0.1];
\draw [fill=red!40] (3.0,0.5) circle [radius=0.1];
\draw [fill=white] (3.5,0.5) circle [radius=0.1];
\draw [fill=red!40] (4.0,0.5) circle [radius=0.1];
\draw [fill=white] (4.5,0.5) circle [radius=0.1];
\draw [fill=white] (5.0,0.5) circle [radius=0.1];
\draw [fill=white] (5.5,0.5) circle [radius=0.1];
\draw [fill=white] (6.0,0.5) circle [radius=0.1];
\draw [fill=white] (7.1,0.5) circle [radius=0.1];

\draw [fill=red!40] (0.0,0.0) circle [radius=0.1];
\draw [fill=white] (0.5,0.0) circle [radius=0.1];
\draw [fill=white] (1.0,0.0) circle [radius=0.1];
\draw [fill=white] (1.5,0.0) circle [radius=0.1];
\draw [fill=white] (2.0,0.0) circle [radius=0.1];
\draw [fill=white] (2.5,0.0) circle [radius=0.1];
\draw [fill=white] (3.0,0.0) circle [radius=0.1];
\draw [fill=white] (3.5,0.0) circle [radius=0.1];
\draw [fill=white] (4.0,0.0) circle [radius=0.1];
\draw [fill=white] (4.5,0.0) circle [radius=0.1];
\draw [fill=white] (5.0,0.0) circle [radius=0.1];
\draw [fill=white] (5.5,0.0) circle [radius=0.1];
\draw [fill=white] (6.0,0.0) circle [radius=0.1];
\draw [fill=red!40] (7.1,0.0) circle [radius=0.1];

\end{tikzpicture}
\caption{\small Example of undetectable errors for R-traps. More generally, R-traps can not detect error generated by (i) phase-flips of physical qubits $(i,j)$ and $(i,j+2)$, where $mod(j,3)=0$, and (ii) phase-flips of physical qubits $(i,1)$ and $(i,m)$. We will refer to errors in the first (respectively second) class as Type-I (respectively Type-II) errors. Type-I and Type-II errors do not affect R-traps, but affect universal quantum computations in general. For this reason, detecting them is imperative.}
\label{fig:BSflip5}
\end{figure}

\noindent hand, we observe that the majority of them generates a different by-product for the logical Hadamard. As an example, the error defined by $l'_2=1$ and the other $l'$ equal to 0 produces an extra Pauli-$X$ if tape $y+1$ implements a logical $R_Z$-gate or $R_X$-gate, while produces an extra Pauli-$Y$ (modulo a global phase) if it implements a logical $H$. Regardless of other effects (such as the flip of some signs), this fact is enough to argue that these errors are detected with probability at least 1/4. Indeed, the probability that the by-products generated within tape $y+1$ cancel out with those generated within tape $y$ is automatically bounded by 1/2 (recall that tape $y+1$ implements a Hadamard with probability 1/2 and a rotation with probability 1/2). Hence, with the same arguments as in the single-tape case, we can see that these errors are detected with probability $\geq$1/4.

The remaining combinations of phase-flips produce the same by-products regardless of the particular logical unitary implemented within tape $y+1$. As an example, the error defined by $l'_1=1$ and the other $l'$ produces a Pauli-$Z$ both in the case where tape $y+1$ implements a rotation or a Hadamard. These phase-flips are detected with probability $\geq$1/4 or are not detectable with R-traps at all. To see this, we consider these particular configurations one by one:
\begin{itemize}
\item $l'_1=1$, other $l'$ equal to 0. This error produces the honest gate modulo a Pauli-$Z$. The extra Pauli can be recovered phase-flipping qubit 3 in tape $y$. Thus, error specified by $l'_1=1$, $l_3=1$, other $l$ and $l'$ equal to 0 is {undetectable} for R-traps. In what follows, we refer to this particular error as Type-I (Figure \ref{fig:BSflip5}).
\item $l_3',l_4'=1$, other $l'$ equal to 0. This error flips the sign of rotations implemented within tape $y+1$, hence is detected with probability 1/4 (equation \ref{eq:rr}) regardless of the unitary implemented within tape $y$. The same happens for the following cases: $l_4',l_5'=1$, other $l'$ equal to 0; $l'_1,l_4',l_5'=1$, other $l'$ equal to 0; $l'_1,l_3',l_4'=1$, other $l'$ equal to 0.

\item $l'_1,l_3',l_5'=1$, other $l'$ equal to 0. This error produces an extra Pauli-$Z$ at the beginning of tape $y+1$ and does not flip any sign. However, flip of qubit 5 is equivalent to flip of qubit 4 of logical qubit $i-1$ or $i+1$, hence can be detected with probability at least $1/4$.
\end{itemize}
Thus, we find the same bound as in step II.\\

\noindent\textbf{Step IV:} We now extend the discussion to combinations of phase-flips affecting two non-neighbouring tapes within row $i$.

If Bob deviates on two non-neighbouring tapes $y$ and $y'>y$, we can move by-products affecting tape $y'$ toward tape $y$ and use a similar argument as the one used in step III: the unitaries implemented within tapes between  $y$ and $y'$ are chosen at random, hence by-products do not cancel out trivially. Indeed, detailed calculations show that they are detected with probability larger than 1/4. The only exception is phase-flipping solely the first and  the last physical qubit within the same row $i$. In this case, the overall unitary acting on qubit $i$ throughout the circuit is known to be a $R_Z$-gate. As a consequence, this combination of phase-flips produces two Pauli-$Z$ errors that cancel out with each other without affecting the overall unitary. This error is undetectable with R-traps, and from now on will be referred to as Type-II (Figure \ref{fig:BSflip5}).

Finally, the same arguments explained above can be used to show that also the combinations of phase-flips affecting more than two tapes are detected with probability larger than 1/4.

\end{proof}

\begin{proof}\textit{(Lemma \ref{lem:c})}
Consider Figure \ref{fig:BSflip6}, which shows the pattern for logical CNOTs. As can be seen, (i) Type-I errors acting on the logical control qubit only (i.e. $l_1,l_2=1$, $l'_1,l'_2=0$) do not corrupt the logical computation, (ii) Type-I errors acting on the logical target qubit only ($l_1,l_2=0$, $l'_1,l'_2=1$) introduce a Pauli-$Z$ error acting on the logical control qubit and (iii) Type-I errors acting on both the control and target qubit ($l_1,l_2,l'_1,l'_2=1$) introduce a Pauli-$Z$ error acting on the control qubit. Thus, in general, Type-I errors produce by-products of Pauli-$Z$ in between the various tapes of the BwS. Similarly, Type-II errors can be seen as by-products of Pauli-$Z$ that first phase-flip some of the inputs to the logical circuit, and subsequently the corresponding outputs.

To prove the lemma we proceed as follows:
\begin{itemize}
	\item[]\textbf{Step I:} We prove that for any combination of Type-I and Type-II errors, there exists some C-trap that can detect the error with probability 1 (meaning, by ``some C-trap'', a specific configuration of target and control qubits for every CNOT in the logical circuit).
	\item[]\textbf{Step II:} We consider the different combinations of Type-I and Type-II errors affecting the C-trap and show that they are detected with probability $\geq$1/2.
\end{itemize}

\begin{figure}[H]
\centering
\begin{tikzpicture}
\draw [dashed] (-1,1.4) -- (0.0,1.4);
\draw [dashed] (5.6,1.4) -- (6.6,1.4);
\draw (0.0,1.4) -- (5.6,1.4);
\draw [dashed] (-1,0.0) -- (0.0,0.0);
\draw [dashed] (5.6,0.0) -- (6.6,0.0);
\draw (0.0,0.0) -- (5.6,0.0);
\draw (2.8,0.0) -- (2.8,1.4);
\draw (5.6,0.0) -- (5.6,1.4);

\draw [dashed,line width=0.5mm] (4.9,-0.7) -- (4.9,2.2);

\draw [fill=white] (0.0,1.4) circle [radius=0.25cm];
\draw [fill=white] (1.4,1.4) circle [radius=0.25cm];
\draw [fill=red!40] (2.8,1.4) circle [radius=0.25cm];
\draw [fill=white] (4.2,1.4) circle [radius=0.25cm];
\draw [fill=red!40] (5.6,1.4) circle [radius=0.25cm];
\draw [fill=white] (0.0,0.0) circle [radius=0.25cm];
\draw [fill=white] (1.4,0.0) circle [radius=0.25cm];
\draw [fill=red!40] (2.8,0.0) circle [radius=0.25cm];
\draw [fill=white] (4.2,0.0) circle [radius=0.25cm];
\draw [fill=red!40] (5.6,0.0) circle [radius=0.25cm];

\node at (2.8,1.4) {$l_{1}$};
\node at (5.6,1.4) {$l_{2}$};
\node at (2.8,0.0) {$l_{1}'$};
\node at (5.6,0.0) {$l_{2}'$};

\node at (0.0,1.85) {$0$};
\node at (1.4,1.85) {$0$};
\node at (2.8,1.85) {$\pi/2$};
\node at (4.2,1.85) {$0$};
\node at (5.6,1.85) {$0$};
\node at (0.0,-0.5) {$0$};
\node at (1.4,-0.5) {$\pi/2$};
\node at (2.8,-0.5) {$0$};
\node at (4.2,-0.5) {$\pi/2$};
\node at (5.6,-0.5) {$0$};

\node at (-0.9,0.3) {$i+1$};
\node at (-0.9,1.7) {$i$};

\node at (2.1,2.3) {\large $y$};

\node at (-0.9,-3.1) {$i+1$};
\node at (-0.9,-1.7) {$i$};

\draw (0.0,-2.0) -- (5.8,-2.0);
\draw [dashed] (-0.9,-2.0) -- (-0.1,-2.0);
\draw [dashed] (5.8,-2.0) -- (6.6,-2.0);
\draw (0.0,-3.4) -- (5.8,-3.4);
\draw [dashed] (-0.9,-3.4) -- (-0.1,-3.4);
\draw [dashed] (5.8,-3.4) -- (6.6,-3.4);

\draw [fill=white] (0.81,-2.4) rectangle (2.45,-1.6);
\node at (1.65,-2.0) {\footnotesize $R_Z(-\frac{\pi}{2})Z^{l_1}$};
\draw [fill=white] (0.81,-3.8) rectangle (2.45,-3.0);
\node at (1.65,-3.4) {\footnotesize $R_X(-\frac{\pi}{2})Z^{l_1'}$};

\draw (2.8,-2.0) -- (2.8,-3.4);
\draw [fill] (2.8,-2.0) circle [radius=0.1cm];
\draw [fill] (2.8,-3.4) circle [radius=0.1cm];
\draw (5.6,-2.0) -- (5.6,-3.4);
\draw [fill] (5.6,-2.0) circle [radius=0.1cm];
\draw [fill] (5.6,-3.4) circle [radius=0.1cm];

\draw [fill=white] (3.8,-2.4) rectangle (4.6,-1.6);
\node at (4.2,-2.0) {\footnotesize $Z^{l_2}$};
\draw [fill=white] (3.43,-3.8) rectangle (4.93,-3.0);
\node at (4.2,-3.4) {\footnotesize $R_X(\frac{\pi}{2})Z^{l_2'}$};

\end{tikzpicture}
\caption{\small Type-I errors affecting the computation at the physical (top figure) and logical (bottom figure) level. Logical qubit $i$ is used as control and logical qubit $i+1$ as the target.}
\label{fig:BSflip6}
\end{figure}
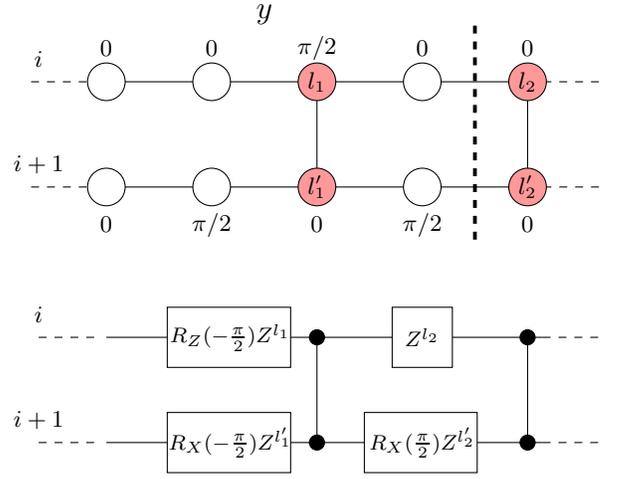

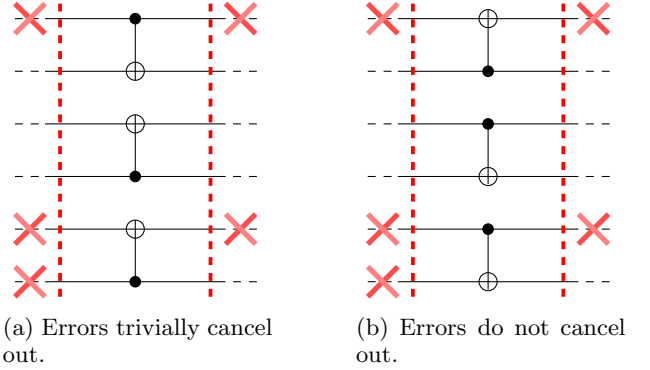
\begin{figure}[H]
\centering
\begin{subfigure}{.2\textwidth}
\centering
\begin{tikzpicture}
\draw [dashed] (-0.4,0.0) -- (0.0,0.0);
\draw [dashed] (2.5,0.0) -- (2.9,0.0);
\draw [dashed] (-0.4,0.7) -- (0.0,0.7);
\draw [dashed] (2.5,0.7) -- (2.9,0.7);
\draw [dashed] (-0.4,1.4) -- (0.0,1.4);
\draw [dashed] (2.5,1.4) -- (2.9,1.4);
\draw [dashed] (-0.4,2.1) -- (0.0,2.1);
\draw [dashed] (2.5,2.1) -- (2.9,2.1);
\draw [dashed] (-0.4,2.8) -- (0.0,2.8);
\draw [dashed] (2.5,2.8) -- (2.9,2.8);
\draw [dashed] (-0.4,3.5) -- (0.0,3.5);
\draw [dashed] (2.5,3.5) -- (2.9,3.5);

\draw (0.0,0.0) -- (2.4,0.0);
\draw (0.0,0.7) -- (2.4,0.7);
\draw (0.0,1.4) -- (2.4,1.4);
\draw (0.0,2.1) -- (2.4,2.1);
\draw (0.0,2.8) -- (2.4,2.8);
\draw (0.0,3.5) -- (2.4,3.5);

\draw [fill] (1.2,0.0) circle [radius=0.7mm]; 
\draw [] (1.2,0.7) circle [radius=1.2mm];
\draw (1.2,0.0) -- (1.2,0.82);
\draw [fill] (1.2,1.4) circle [radius=0.7mm]; 
\draw [] (1.2,2.1) circle [radius=1.2mm];
\draw (1.2,1.4) -- (1.2,2.22);
\draw [fill] (1.2,3.5) circle [radius=0.7mm]; 
\draw [] (1.2,2.8) circle [radius=1.2mm];
\draw (1.2,2.68) -- (1.2,3.5);

\draw [line width=0.5mm,dashed,red] (0.2,-0.2) -- (0.2,3.7);
\draw [line width=0.5mm,dashed,red] (2.2,-0.2) -- (2.2,3.7);

\draw [line width=0.7mm,red!70] (-0.4,-0.2) -- (-0.0,0.2);
\draw [line width=0.7mm,red!50] (-0.0,-0.2) -- (-0.4,0.2);
\draw [line width=0.7mm,red!70] (-0.4,0.5) -- (-0.0,0.9);
\draw [line width=0.7mm,red!50] (-0.0,0.5) -- (-0.4,0.9);
\draw [line width=0.7mm,red!70] (-0.4,3.7) -- (-0.0,3.3);
\draw [line width=0.7mm,red!50] (-0.0,3.7) -- (-0.4,3.3);

\draw [line width=0.7mm,red!70] (2.4,0.5) -- (2.8,0.9);
\draw [line width=0.7mm,red!50] (2.8,0.5) -- (2.4,0.9);
\draw [line width=0.7mm,red!70] (2.4,3.7) -- (2.8,3.3);
\draw [line width=0.7mm,red!50] (2.8,3.7) -- (2.4,3.3);
\end{tikzpicture}
\caption{Errors trivially cancel out.}
\label{fig:proofst2a}
\end{subfigure}%
\hspace{1cm}
\begin{subfigure}{.2\textwidth}
\centering
\begin{tikzpicture}
\draw [dashed] (-0.4,0.0) -- (0.0,0.0);
\draw [dashed] (2.5,0.0) -- (2.9,0.0);
\draw [dashed] (-0.4,0.7) -- (0.0,0.7);
\draw [dashed] (2.5,0.7) -- (2.9,0.7);
\draw [dashed] (-0.4,1.4) -- (0.0,1.4);
\draw [dashed] (2.5,1.4) -- (2.9,1.4);
\draw [dashed] (-0.4,2.1) -- (0.0,2.1);
\draw [dashed] (2.5,2.1) -- (2.9,2.1);
\draw [dashed] (-0.4,2.8) -- (0.0,2.8);
\draw [dashed] (2.5,2.8) -- (2.9,2.8);
\draw [dashed] (-0.4,3.5) -- (0.0,3.5);
\draw [dashed] (2.5,3.5) -- (2.9,3.5);

\draw (0.0,0.0) -- (2.4,0.0);
\draw (0.0,0.7) -- (2.4,0.7);
\draw (0.0,1.4) -- (2.4,1.4);
\draw (0.0,2.1) -- (2.4,2.1);
\draw (0.0,2.8) -- (2.4,2.8);
\draw (0.0,3.5) -- (2.4,3.5);

\draw [fill] (1.2,0.7) circle [radius=0.7mm]; 
\draw [] (1.2,0.0) circle [radius=1.2mm];
\draw (1.2,0.7) -- (1.2,-0.1);

\draw [fill] (1.2,2.1) circle [radius=0.7mm]; 
\draw [] (1.2,1.4) circle [radius=1.2mm];
\draw (1.2,2.1) -- (1.2,1.3);

\draw [fill] (1.2,2.8) circle [radius=0.7mm]; 
\draw [] (1.2,3.5) circle [radius=1.2mm];
\draw (1.2,2.8) -- (1.2,3.62);

\draw [line width=0.5mm,dashed,red] (0.2,-0.2) -- (0.2,3.7);
\draw [line width=0.5mm,dashed,red] (2.2,-0.2) -- (2.2,3.7);

\draw [line width=0.7mm,red!70] (-0.4,-0.2) -- (-0.0,0.2);
\draw [line width=0.7mm,red!50] (-0.0,-0.2) -- (-0.4,0.2);
\draw [line width=0.7mm,red!70] (-0.4,0.5) -- (-0.0,0.9);
\draw [line width=0.7mm,red!50] (-0.0,0.5) -- (-0.4,0.9);
\draw [line width=0.7mm,red!70] (-0.4,3.7) -- (-0.0,3.3);
\draw [line width=0.7mm,red!50] (-0.0,3.7) -- (-0.4,3.3);

\draw [line width=0.7mm,red!70] (2.4,0.5) -- (2.8,0.9);
\draw [line width=0.7mm,red!50] (2.8,0.5) -- (2.4,0.9);
\draw [line width=0.7mm,red!70] (2.4,3.7) -- (2.8,3.3);
\draw [line width=0.7mm,red!50] (2.8,3.7) -- (2.4,3.3);
\end{tikzpicture}
\caption{Errors do not cancel out.}
\label{fig:proofst2b}
\end{subfigure}%
\caption{\small Proof of equation \ref{eq:cc}, an example. Red crosses represent by-products of Pauli-$Z$. Figure \ref{fig:proofst2a} illustrates a combination of CNOTs such that errors cancel out with each other. A configuration of CNOTs that does not let errors cancel out can be found swapping control and target qubits in every brick (Figure \ref{fig:proofst2b}).}
\label{fig:proofst2}
\end{figure}

\begin{itemize}
	 \item[] To be specific, we first show that this is true for errors affecting two nearest-neighbouring tapes. Next, we argue that any combination of Type-I and Type-II errors can be rewritten as a specific error affecting two nearest-neighbouring tapes, and so that the bound holds in general.
\end{itemize}

\noindent In more detail:\\

\noindent \textbf{Step I.} We begin by describing the notation that will be used along this step of the proof. We formally define CNOTs acting on qubits $c$ (control) and $t$ (target) as
\begin{equation}
cX_{c,t}=|0\rangle_c\langle0|\otimes\mathbb{1}_t+|1\rangle_c\langle1|\otimes X_t
\end{equation}
Next, we represent the unitary implemented within tape $y$ with $\mathbf{cX}^{\bar{k}}_y$ (notice the bold font used to distinguish $n$-

\newpage
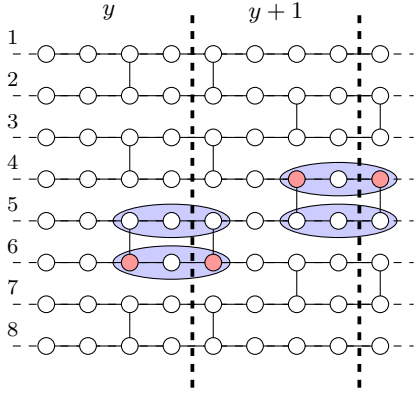
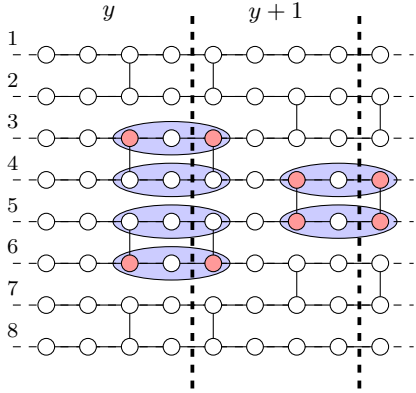
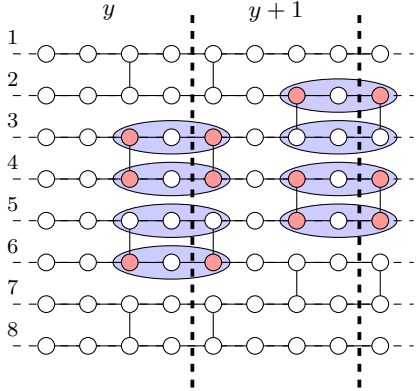
\begin{figure}[H]
\begin{subfigure}{.45\textwidth}
\centering
\begin{tikzpicture}[scale=1.11]
\node at (-0.4,3.7) {\footnotesize $1$};
\node at (-0.4,3.2) {\footnotesize $2$};
\node at (-0.4,2.7) {\footnotesize $3$};
\node at (-0.4,2.2) {\footnotesize $4$};
\node at (-0.4,1.7) {\footnotesize $5$};
\node at (-0.4,1.2) {\footnotesize $6$};
\node at (-0.4,0.7) {\footnotesize $7$};
\node at (-0.4,0.2) {\footnotesize $8$};

\draw [fill=blue!20] (3.5,2.0) ellipse (0.7cm and 0.2cm);
\draw [fill=blue!20] (3.5,1.5) ellipse (0.7cm and 0.2cm);
\draw [fill=blue!20] (1.5,1.5) ellipse (0.7cm and 0.2cm);
\draw [fill=blue!20] (1.5,1.0) ellipse (0.7cm and 0.2cm);

\draw [line width=0.5mm,dashed] (0.0+1.75,-0.5) -- (0.0+1.75,4.0);
\draw [line width=0.5mm,dashed] (2.0+1.75,-0.5) -- (2.0+1.75,4.0);

\draw (0.0,0.0) -- (4.0,0.0);
\draw (0.0,0.5) -- (4.0,0.5);
\draw (0.0,1.0) -- (4.0,1.0);
\draw (0.0,1.5) -- (4.0,1.5);
\draw (0.0,2.0) -- (4.0,2.0);
\draw (0.0,2.5) -- (4.0,2.5);
\draw (0.0,3.0) -- (4.0,3.0);
\draw (0.0,3.5) -- (4.0,3.5);

\draw [dashed] (-0.4,0.0) -- (4.4,0.0);
\draw [dashed] (-0.4,0.5) -- (4.4,0.5);
\draw [dashed] (-0.4,1.0) -- (4.4,1.0);
\draw [dashed] (-0.4,1.5) -- (4.4,1.5);
\draw [dashed] (-0.4,2.0) -- (4.4,2.0);
\draw [dashed] (-0.4,2.5) -- (4.4,2.5);
\draw [dashed] (-0.4,3.0) -- (4.4,3.0);
\draw [dashed] (-0.4,3.5) -- (4.4,3.5);

\draw (1.0,0.0) -- (1.0,0.5);
\draw (1.0,1.0) -- (1.0,1.5);
\draw (1.0,2.0) -- (1.0,2.5);
\draw (1.0,3.0) -- (1.0,3.5);
\draw (2.0,0.0) -- (2.0,0.5);
\draw (2.0,1.0) -- (2.0,1.5);
\draw (2.0,2.0) -- (2.0,2.5);
\draw (2.0,3.0) -- (2.0,3.5);

\draw (3.0,0.5) -- (3.0,1.0);
\draw (3.0,1.5) -- (3.0,2.0);
\draw (3.0,2.5) -- (3.0,3.0);
\draw (4.0,0.5) -- (4.0,1.0);
\draw (4.0,1.5) -- (4.0,2.0);
\draw (4.0,2.5) -- (4.0,3.0);

\draw [fill=white] (0.0,3.5) circle [radius=0.1];
\draw [fill=white] (0.5,3.5) circle [radius=0.1];
\draw [fill=white] (1.0,3.5) circle [radius=0.1];
\draw [fill=white] (1.5,3.5) circle [radius=0.1];
\draw [fill=white] (2.0,3.5) circle [radius=0.1];
\draw [fill=white] (2.5,3.5) circle [radius=0.1];
\draw [fill=white] (3.0,3.5) circle [radius=0.1];
\draw [fill=white] (3.5,3.5) circle [radius=0.1];
\draw [fill=white] (4.0,3.5) circle [radius=0.1];

\draw [fill=white] (0.0,3.0) circle [radius=0.1];
\draw [fill=white] (0.5,3.0) circle [radius=0.1];
\draw [fill=white] (1.0,3.0) circle [radius=0.1];
\draw [fill=white] (1.5,3.0) circle [radius=0.1];
\draw [fill=white] (2.0,3.0) circle [radius=0.1];
\draw [fill=white] (2.5,3.0) circle [radius=0.1];
\draw [fill=white] (3.0,3.0) circle [radius=0.1];
\draw [fill=white] (3.5,3.0) circle [radius=0.1];
\draw [fill=white] (4.0,3.0) circle [radius=0.1];

\draw [fill=white] (0.0,2.5) circle [radius=0.1];
\draw [fill=white] (0.5,2.5) circle [radius=0.1];
\draw [fill=white] (1.0,2.5) circle [radius=0.1];
\draw [fill=white] (1.5,2.5) circle [radius=0.1];
\draw [fill=white] (2.0,2.5) circle [radius=0.1];
\draw [fill=white] (2.5,2.5) circle [radius=0.1];
\draw [fill=white] (3.0,2.5) circle [radius=0.1];
\draw [fill=white] (3.5,2.5) circle [radius=0.1];
\draw [fill=white] (4.0,2.5) circle [radius=0.1];

\draw [fill=white] (0.0,2.0) circle [radius=0.1];
\draw [fill=white] (0.5,2.0) circle [radius=0.1];
\draw [fill=white] (1.0,2.0) circle [radius=0.1];
\draw [fill=white] (1.5,2.0) circle [radius=0.1];
\draw [fill=white] (2.0,2.0) circle [radius=0.1];
\draw [fill=white] (2.5,2.0) circle [radius=0.1];
\draw [fill=red!40] (3.0,2.0) circle [radius=0.1];
\draw [fill=white] (3.5,2.0) circle [radius=0.1];
\draw [fill=red!40] (4.0,2.0) circle [radius=0.1];

\draw [fill=white] (0.0,1.5) circle [radius=0.1];
\draw [fill=white] (0.5,1.5) circle [radius=0.1];
\draw [fill=white] (1.0,1.5) circle [radius=0.1];
\draw [fill=white] (1.5,1.5) circle [radius=0.1];
\draw [fill=white] (2.0,1.5) circle [radius=0.1];
\draw [fill=white] (2.5,1.5) circle [radius=0.1];
\draw [fill=white] (3.0,1.5) circle [radius=0.1];
\draw [fill=white] (3.5,1.5) circle [radius=0.1];
\draw [fill=white] (4.0,1.5) circle [radius=0.1];

\draw [fill=white] (0.0,1.0) circle [radius=0.1];
\draw [fill=white] (0.5,1.0) circle [radius=0.1];
\draw [fill=red!40] (1.0,1.0) circle [radius=0.1];
\draw [fill=white] (1.5,1.0) circle [radius=0.1];
\draw [fill=red!40] (2.0,1.0) circle [radius=0.1];
\draw [fill=white] (2.5,1.0) circle [radius=0.1];
\draw [fill=white] (3.0,1.0) circle [radius=0.1];
\draw [fill=white] (3.5,1.0) circle [radius=0.1];
\draw [fill=white] (4.0,1.0) circle [radius=0.1];

\draw [fill=white] (0.0,0.5) circle [radius=0.1];
\draw [fill=white] (0.5,0.5) circle [radius=0.1];
\draw [fill=white] (1.0,0.5) circle [radius=0.1];
\draw [fill=white] (1.5,0.5) circle [radius=0.1];
\draw [fill=white] (2.0,0.5) circle [radius=0.1];
\draw [fill=white] (2.5,0.5) circle [radius=0.1];
\draw [fill=white] (3.0,0.5) circle [radius=0.1];
\draw [fill=white] (3.5,0.5) circle [radius=0.1];
\draw [fill=white] (4.0,0.5) circle [radius=0.1];

\draw [fill=white] (0.0,0.0) circle [radius=0.1];
\draw [fill=white] (0.5,0.0) circle [radius=0.1];
\draw [fill=white] (1.0,0.0) circle [radius=0.1];
\draw [fill=white] (1.5,0.0) circle [radius=0.1];
\draw [fill=white] (2.0,0.0) circle [radius=0.1];
\draw [fill=white] (2.5,0.0) circle [radius=0.1];
\draw [fill=white] (3.0,0.0) circle [radius=0.1];
\draw [fill=white] (3.5,0.0) circle [radius=0.1];
\draw [fill=white] (4.0,0.0) circle [radius=0.1];

\node at (0.75,4) {$y$};
\node at (2.75,4) {$y+1$};
\end{tikzpicture}
\caption{Type-I errors affecting two nearest-neighbouring and overlapping bricks within nearest-neighbouring tapes. The phase-flip of red physical qubits is detected with probability 1/2, while any other combination of Type-I errors is detected with probability $\geq1/2$.}
\label{fig:int1}
\end{subfigure}%

\begin{subfigure}{.45\textwidth}
\centering
\begin{tikzpicture}[scale=1.11]
\node at (-0.4,3.7) {\footnotesize $1$};
\node at (-0.4,3.2) {\footnotesize $2$};
\node at (-0.4,2.7) {\footnotesize $3$};
\node at (-0.4,2.2) {\footnotesize $4$};
\node at (-0.4,1.7) {\footnotesize $5$};
\node at (-0.4,1.2) {\footnotesize $6$};
\node at (-0.4,0.7) {\footnotesize $7$};
\node at (-0.4,0.2) {\footnotesize $8$};

\draw [fill=blue!20] (1.5,2.5) ellipse (0.7cm and 0.2cm);
\draw [fill=blue!20] (1.5,2.0) ellipse (0.7cm and 0.2cm);
\draw [fill=blue!20] (3.5,2.0) ellipse (0.7cm and 0.2cm);
\draw [fill=blue!20] (3.5,1.5) ellipse (0.7cm and 0.2cm);
\draw [fill=blue!20] (1.5,1.5) ellipse (0.7cm and 0.2cm);
\draw [fill=blue!20] (1.5,1.0) ellipse (0.7cm and 0.2cm);

\draw [line width=0.5mm,dashed] (0.0+1.75,-0.5) -- (0.0+1.75,4.0);
\draw [line width=0.5mm,dashed] (2.0+1.75,-0.5) -- (2.0+1.75,4.0);

\draw (0.0,0.0) -- (4.0,0.0);
\draw (0.0,0.5) -- (4.0,0.5);
\draw (0.0,1.0) -- (4.0,1.0);
\draw (0.0,1.5) -- (4.0,1.5);
\draw (0.0,2.0) -- (4.0,2.0);
\draw (0.0,2.5) -- (4.0,2.5);
\draw (0.0,3.0) -- (4.0,3.0);
\draw (0.0,3.5) -- (4.0,3.5);

\draw [dashed] (-0.4,0.0) -- (4.4,0.0);
\draw [dashed] (-0.4,0.5) -- (4.4,0.5);
\draw [dashed] (-0.4,1.0) -- (4.4,1.0);
\draw [dashed] (-0.4,1.5) -- (4.4,1.5);
\draw [dashed] (-0.4,2.0) -- (4.4,2.0);
\draw [dashed] (-0.4,2.5) -- (4.4,2.5);
\draw [dashed] (-0.4,3.0) -- (4.4,3.0);
\draw [dashed] (-0.4,3.5) -- (4.4,3.5);

\draw (1.0,0.0) -- (1.0,0.5);
\draw (1.0,1.0) -- (1.0,1.5);
\draw (1.0,2.0) -- (1.0,2.5);
\draw (1.0,3.0) -- (1.0,3.5);
\draw (2.0,0.0) -- (2.0,0.5);
\draw (2.0,1.0) -- (2.0,1.5);
\draw (2.0,2.0) -- (2.0,2.5);
\draw (2.0,3.0) -- (2.0,3.5);

\draw (3.0,0.5) -- (3.0,1.0);
\draw (3.0,1.5) -- (3.0,2.0);
\draw (3.0,2.5) -- (3.0,3.0);
\draw (4.0,0.5) -- (4.0,1.0);
\draw (4.0,1.5) -- (4.0,2.0);
\draw (4.0,2.5) -- (4.0,3.0);

\draw [fill=white] (0.0,3.5) circle [radius=0.1];
\draw [fill=white] (0.5,3.5) circle [radius=0.1];
\draw [fill=white] (1.0,3.5) circle [radius=0.1];
\draw [fill=white] (1.5,3.5) circle [radius=0.1];
\draw [fill=white] (2.0,3.5) circle [radius=0.1];
\draw [fill=white] (2.5,3.5) circle [radius=0.1];
\draw [fill=white] (3.0,3.5) circle [radius=0.1];
\draw [fill=white] (3.5,3.5) circle [radius=0.1];
\draw [fill=white] (4.0,3.5) circle [radius=0.1];

\draw [fill=white] (0.0,3.0) circle [radius=0.1];
\draw [fill=white] (0.5,3.0) circle [radius=0.1];
\draw [fill=white] (1.0,3.0) circle [radius=0.1];
\draw [fill=white] (1.5,3.0) circle [radius=0.1];
\draw [fill=white] (2.0,3.0) circle [radius=0.1];
\draw [fill=white] (2.5,3.0) circle [radius=0.1];
\draw [fill=white] (3.0,3.0) circle [radius=0.1];
\draw [fill=white] (3.5,3.0) circle [radius=0.1];
\draw [fill=white] (4.0,3.0) circle [radius=0.1];

\draw [fill=white] (0.0,2.5) circle [radius=0.1];
\draw [fill=white] (0.5,2.5) circle [radius=0.1];
\draw [fill=red!40] (1.0,2.5) circle [radius=0.1];
\draw [fill=white] (1.5,2.5) circle [radius=0.1];
\draw [fill=red!40] (2.0,2.5) circle [radius=0.1];
\draw [fill=white] (2.5,2.5) circle [radius=0.1];
\draw [fill=white] (3.0,2.5) circle [radius=0.1];
\draw [fill=white] (3.5,2.5) circle [radius=0.1];
\draw [fill=white] (4.0,2.5) circle [radius=0.1];

\draw [fill=white] (0.0,2.0) circle [radius=0.1];
\draw [fill=white] (0.5,2.0) circle [radius=0.1];
\draw [fill=white] (1.0,2.0) circle [radius=0.1];
\draw [fill=white] (1.5,2.0) circle [radius=0.1];
\draw [fill=white] (2.0,2.0) circle [radius=0.1];
\draw [fill=white] (2.5,2.0) circle [radius=0.1];
\draw [fill=red!40] (3.0,2.0) circle [radius=0.1];
\draw [fill=white] (3.5,2.0) circle [radius=0.1];
\draw [fill=red!40] (4.0,2.0) circle [radius=0.1];

\draw [fill=white] (0.0,1.5) circle [radius=0.1];
\draw [fill=white] (0.5,1.5) circle [radius=0.1];
\draw [fill=white] (1.0,1.5) circle [radius=0.1];
\draw [fill=white] (1.5,1.5) circle [radius=0.1];
\draw [fill=white] (2.0,1.5) circle [radius=0.1];
\draw [fill=white] (2.5,1.5) circle [radius=0.1];
\draw [fill=red!40] (3.0,1.5) circle [radius=0.1];
\draw [fill=white] (3.5,1.5) circle [radius=0.1];
\draw [fill=red!40] (4.0,1.5) circle [radius=0.1];

\draw [fill=white] (0.0,1.0) circle [radius=0.1];
\draw [fill=white] (0.5,1.0) circle [radius=0.1];
\draw [fill=red!40] (1.0,1.0) circle [radius=0.1];
\draw [fill=white] (1.5,1.0) circle [radius=0.1];
\draw [fill=red!40] (2.0,1.0) circle [radius=0.1];
\draw [fill=white] (2.5,1.0) circle [radius=0.1];
\draw [fill=white] (3.0,1.0) circle [radius=0.1];
\draw [fill=white] (3.5,1.0) circle [radius=0.1];
\draw [fill=white] (4.0,1.0) circle [radius=0.1];

\draw [fill=white] (0.0,0.5) circle [radius=0.1];
\draw [fill=white] (0.5,0.5) circle [radius=0.1];
\draw [fill=white] (1.0,0.5) circle [radius=0.1];
\draw [fill=white] (1.5,0.5) circle [radius=0.1];
\draw [fill=white] (2.0,0.5) circle [radius=0.1];
\draw [fill=white] (2.5,0.5) circle [radius=0.1];
\draw [fill=white] (3.0,0.5) circle [radius=0.1];
\draw [fill=white] (3.5,0.5) circle [radius=0.1];
\draw [fill=white] (4.0,0.5) circle [radius=0.1];

\draw [fill=white] (0.0,0.0) circle [radius=0.1];
\draw [fill=white] (0.5,0.0) circle [radius=0.1];
\draw [fill=white] (1.0,0.0) circle [radius=0.1];
\draw [fill=white] (1.5,0.0) circle [radius=0.1];
\draw [fill=white] (2.0,0.0) circle [radius=0.1];
\draw [fill=white] (2.5,0.0) circle [radius=0.1];
\draw [fill=white] (3.0,0.0) circle [radius=0.1];
\draw [fill=white] (3.5,0.0) circle [radius=0.1];
\draw [fill=white] (4.0,0.0) circle [radius=0.1];

\node at (0.75,4) {$y$};
\node at (2.75,4) {$y+1$};
\end{tikzpicture}
\caption{Type-I errors affecting four nearest-neighbouring and overlapping bricks. The phase-flip of red physical qubits is detected with probability 1/4, while any other combination of Type-I errors is detected with probability $\geq1/4$.}
\label{fig:int2}
\end{subfigure}%

\begin{subfigure}{.45\textwidth}
\centering
\begin{tikzpicture}[scale=1.11]
\node at (-0.4,3.7) {\footnotesize $1$};
\node at (-0.4,3.2) {\footnotesize $2$};
\node at (-0.4,2.7) {\footnotesize $3$};
\node at (-0.4,2.2) {\footnotesize $4$};
\node at (-0.4,1.7) {\footnotesize $5$};
\node at (-0.4,1.2) {\footnotesize $6$};
\node at (-0.4,0.7) {\footnotesize $7$};
\node at (-0.4,0.2) {\footnotesize $8$};

\draw [fill=blue!20] (3.5,3.0) ellipse (0.7cm and 0.2cm);
\draw [fill=blue!20] (3.5,2.5) ellipse (0.7cm and 0.2cm);
\draw [fill=blue!20] (1.5,2.5) ellipse (0.7cm and 0.2cm);
\draw [fill=blue!20] (1.5,2.0) ellipse (0.7cm and 0.2cm);
\draw [fill=blue!20] (3.5,2.0) ellipse (0.7cm and 0.2cm);
\draw [fill=blue!20] (3.5,1.5) ellipse (0.7cm and 0.2cm);
\draw [fill=blue!20] (1.5,1.5) ellipse (0.7cm and 0.2cm);
\draw [fill=blue!20] (1.5,1.0) ellipse (0.7cm and 0.2cm);

\node at (0.75,4) {$y$};
\node at (2.75,4) {$y+1$};

\draw [line width=0.5mm,dashed] (0.0+1.75,-0.5) -- (0.0+1.75,4.0);
\draw [line width=0.5mm,dashed] (2.0+1.75,-0.5) -- (2.0+1.75,4.0);

\draw (0.0,0.0) -- (4.0,0.0);
\draw (0.0,0.5) -- (4.0,0.5);
\draw (0.0,1.0) -- (4.0,1.0);
\draw (0.0,1.5) -- (4.0,1.5);
\draw (0.0,2.0) -- (4.0,2.0);
\draw (0.0,2.5) -- (4.0,2.5);
\draw (0.0,3.0) -- (4.0,3.0);
\draw (0.0,3.5) -- (4.0,3.5);

\draw [dashed] (-0.4,0.0) -- (4.4,0.0);
\draw [dashed] (-0.4,0.5) -- (4.4,0.5);
\draw [dashed] (-0.4,1.0) -- (4.4,1.0);
\draw [dashed] (-0.4,1.5) -- (4.4,1.5);
\draw [dashed] (-0.4,2.0) -- (4.4,2.0);
\draw [dashed] (-0.4,2.5) -- (4.4,2.5);
\draw [dashed] (-0.4,3.0) -- (4.4,3.0);
\draw [dashed] (-0.4,3.5) -- (4.4,3.5);

\draw (1.0,0.0) -- (1.0,0.5);
\draw (1.0,1.0) -- (1.0,1.5);
\draw (1.0,2.0) -- (1.0,2.5);
\draw (1.0,3.0) -- (1.0,3.5);
\draw (2.0,0.0) -- (2.0,0.5);
\draw (2.0,1.0) -- (2.0,1.5);
\draw (2.0,2.0) -- (2.0,2.5);
\draw (2.0,3.0) -- (2.0,3.5);

\draw (3.0,0.5) -- (3.0,1.0);
\draw (3.0,1.5) -- (3.0,2.0);
\draw (3.0,2.5) -- (3.0,3.0);
\draw (4.0,0.5) -- (4.0,1.0);
\draw (4.0,1.5) -- (4.0,2.0);
\draw (4.0,2.5) -- (4.0,3.0);

\draw [fill=white] (0.0,3.5) circle [radius=0.1];
\draw [fill=white] (0.5,3.5) circle [radius=0.1];
\draw [fill=white] (1.0,3.5) circle [radius=0.1];
\draw [fill=white] (1.5,3.5) circle [radius=0.1];
\draw [fill=white] (2.0,3.5) circle [radius=0.1];
\draw [fill=white] (2.5,3.5) circle [radius=0.1];
\draw [fill=white] (3.0,3.5) circle [radius=0.1];
\draw [fill=white] (3.5,3.5) circle [radius=0.1];
\draw [fill=white] (4.0,3.5) circle [radius=0.1];

\draw [fill=white] (0.0,3.0) circle [radius=0.1];
\draw [fill=white] (0.5,3.0) circle [radius=0.1];
\draw [fill=white] (1.0,3.0) circle [radius=0.1];
\draw [fill=white] (1.5,3.0) circle [radius=0.1];
\draw [fill=white] (2.0,3.0) circle [radius=0.1];
\draw [fill=white] (2.5,3.0) circle [radius=0.1];
\draw [fill=red!40] (3.0,3.0) circle [radius=0.1];
\draw [fill=white] (3.5,3.0) circle [radius=0.1];
\draw [fill=red!40] (4.0,3.0) circle [radius=0.1];

\draw [fill=white] (0.0,2.5) circle [radius=0.1];
\draw [fill=white] (0.5,2.5) circle [radius=0.1];
\draw [fill=red!40] (1.0,2.5) circle [radius=0.1];
\draw [fill=white] (1.5,2.5) circle [radius=0.1];
\draw [fill=red!40] (2.0,2.5) circle [radius=0.1];
\draw [fill=white] (2.5,2.5) circle [radius=0.1];
\draw [fill=white] (3.0,2.5) circle [radius=0.1];
\draw [fill=white] (3.5,2.5) circle [radius=0.1];
\draw [fill=white] (4.0,2.5) circle [radius=0.1];

\draw [fill=white] (0.0,2.0) circle [radius=0.1];
\draw [fill=white] (0.5,2.0) circle [radius=0.1];
\draw [fill=red!40] (1.0,2.0) circle [radius=0.1];
\draw [fill=white] (1.5,2.0) circle [radius=0.1];
\draw [fill=red!40] (2.0,2.0) circle [radius=0.1];
\draw [fill=white] (2.5,2.0) circle [radius=0.1];
\draw [fill=red!40] (3.0,2.0) circle [radius=0.1];
\draw [fill=white] (3.5,2.0) circle [radius=0.1];
\draw [fill=red!40] (4.0,2.0) circle [radius=0.1];

\draw [fill=white] (0.0,1.5) circle [radius=0.1];
\draw [fill=white] (0.5,1.5) circle [radius=0.1];
\draw [fill=white] (1.0,1.5) circle [radius=0.1];
\draw [fill=white] (1.5,1.5) circle [radius=0.1];
\draw [fill=white] (2.0,1.5) circle [radius=0.1];
\draw [fill=white] (2.5,1.5) circle [radius=0.1];
\draw [fill=red!40] (3.0,1.5) circle [radius=0.1];
\draw [fill=white] (3.5,1.5) circle [radius=0.1];
\draw [fill=red!40] (4.0,1.5) circle [radius=0.1];

\draw [fill=white] (0.0,1.0) circle [radius=0.1];
\draw [fill=white] (0.5,1.0) circle [radius=0.1];
\draw [fill=red!40] (1.0,1.0) circle [radius=0.1];
\draw [fill=white] (1.5,1.0) circle [radius=0.1];
\draw [fill=red!40] (2.0,1.0) circle [radius=0.1];
\draw [fill=white] (2.5,1.0) circle [radius=0.1];
\draw [fill=white] (3.0,1.0) circle [radius=0.1];
\draw [fill=white] (3.5,1.0) circle [radius=0.1];
\draw [fill=white] (4.0,1.0) circle [radius=0.1];

\draw [fill=white] (0.0,0.5) circle [radius=0.1];
\draw [fill=white] (0.5,0.5) circle [radius=0.1];
\draw [fill=white] (1.0,0.5) circle [radius=0.1];
\draw [fill=white] (1.5,0.5) circle [radius=0.1];
\draw [fill=white] (2.0,0.5) circle [radius=0.1];
\draw [fill=white] (2.5,0.5) circle [radius=0.1];
\draw [fill=white] (3.0,0.5) circle [radius=0.1];
\draw [fill=white] (3.5,0.5) circle [radius=0.1];
\draw [fill=white] (4.0,0.5) circle [radius=0.1];

\draw [fill=white] (0.0,0.0) circle [radius=0.1];
\draw [fill=white] (0.5,0.0) circle [radius=0.1];
\draw [fill=white] (1.0,0.0) circle [radius=0.1];
\draw [fill=white] (1.5,0.0) circle [radius=0.1];
\draw [fill=white] (2.0,0.0) circle [radius=0.1];
\draw [fill=white] (2.5,0.0) circle [radius=0.1];
\draw [fill=white] (3.0,0.0) circle [radius=0.1];
\draw [fill=white] (3.5,0.0) circle [radius=0.1];
\draw [fill=white] (4.0,0.0) circle [radius=0.1];
\end{tikzpicture}
\caption{Type-I errors affecting three nearest-neighbouring and overlapping bricks. The phase-flip of red physical qubits is detected with probability 1/8, while any other combination of Type-I errors is detected with probability $\geq1/8$.}
\label{fig:int3}
\end{subfigure}%

\caption{\small Type-I errors affecting nearest-neighbouring tapes. In any picture, we illustrate the error that has the least probability of being detected. The reader can verify that any other combination of Type-I errors is detected with higher probability.}
\end{figure}

\newpage

\noindent qubit operations from two-qubit ones). Here, $\bar{k}$ is a binary vector composed by $n/2$ elements if $y$ is odd, and by $n/2-1$ if $y$ is even. As a convention, indicating the null-vector as $\bar{0}$, we define
\begin{eqnarray}
\begin{tabular}{lllllll}
$\mathbf{cX}^{\bar{0}}_y$&=&$cX_{1,2}\otimes cX_{3,4}\otimes..\otimes cX_{n-1,n}$&\cr
&&for $y$ odd\\ \cr
&&and\\ \cr
$\mathbf{cX}^{\bar{0}}_y$&=&$\mathbb{1}_1\otimes cX_{2,3}\otimes cX_{4,5}\otimes..\otimes cX_{n-2,n-1}\otimes \mathbb{1}_n$&\cr
&&for $y$ even\cr
\end{tabular}
\end{eqnarray}
Flipping a given element of $\bar{0}$ to 1 is equivalent, in this notation, to exchange the role of target and control qubits in the CNOT operation implemented within the corresponding brick. Finally, we represent by-products of Pauli-$Z$ as $\mathbf{Z}^{\bar{g}}$ (again, notice the bold font), where $\bar{g}$ is a $n$-element binary vector whose $i^{th}$-element is equal to 0 if no error is affecting logical qubit $i$, otherwise is equal to 1. 

With this new notation, we can prove that for any $n$-element binary vectors $\bar{g}$ and $\bar{g}'$ such that they are not simultaneously equal to $\bar{0}$, there exists a $\bar{k}$ such that
\begin{equation}~\label{eq:cc}
\mathbf{Z}^{\bar{g}'}\textbf{cX}^{\bar{k}}\mathbf{Z}^{\bar{g}}\neq\textbf{cX}^{\bar{k}}\textrm{ ,}
\end{equation}
The proof of equation \ref{eq:cc} can be obtained using the relations
\begin{eqnarray}
\begin{tabular}{lllll}
$\big(Z_c\otimes\mathbb{1}_t\big)\textrm{ }cX_{c,t}$&=&$cX_{c,t}\textrm{ }\big(Z_c\otimes\mathbb{1}_t\big)$\\ \cr
$\big(\mathbb{1}_c\otimes Z_t\big)\textrm{ }cX_{c,t}$&=&$cX_{c,t}\textrm{ }\big(Z_c\otimes Z_t\big)$\\ \cr
$\big(Z_c\otimes Z_t\big)\textrm{ }cX_{c,t}$&=&$cX_{c,t}\textrm{ }\big(\mathbb{1}_c\otimes Z_t\big)$\cr
\end{tabular}
\label{eq:ctraps}
\end{eqnarray}
valid for any two qubits $c$ and $t$. In more detail, take a specific $\bar{k}$ such that $\mathbf{Z}^{\bar{g}'}\textbf{cX}^{\bar{k}}\mathbf{Z}^{\bar{g}}=\textbf{cX}^{\bar{k}}$. Using equations \ref{eq:ctraps}, it follows that $\mathbf{Z}^{\bar{g}'}\textbf{cX}^{\bar{k}\textrm{}\oplus\bar{1}}\mathbf{Z}^{\bar{g}}\neq\textbf{cX}^{\bar{k}\textrm{}\oplus\bar{1}}$ (Figure \ref{fig:proofst2}).

Using equation \ref{eq:cc}, we can finally prove that for any error produced by a combination of Type-I and Type-II errors, there exists at least a C-trap that can detect the error with probability 1. Consider the operator
\begin{equation}
U_C^{\bar{k}_w,..,\bar{k}_1}=\textbf{cX}^{\bar{k}_w}_w\cdot..\cdot \textbf{cX}^{\bar{k}_1}_1\textrm{ ,}
\end{equation}
namely the unitary implemented by a C-trap for given vectors $\bar{k}_1,..,\bar{k}_w$. Also, consider a specific set of by-products $\{\mathbf{Z}^{\overline{g}_0},\mathbf{Z}^{\overline{g}_1},..,\mathbf{Z}^{\overline{g}_{w}}\}$ corrupting the logical computation (here, $\mathbf{Z}^{\overline{g}_0}$ represents the phase-flips due to Type-II errors; on the contrary, for any $y=1,..,w$, $\mathbf{Z}^{\overline{g}_y}$ represents the errors affecting the logical circuit in between tapes $y$ and $y+1$). In the presence of this deviation, $U_C$ becomes                                      
\begin{equation*}
\widetilde{U}_C^{\bar{k}_w,..,\bar{k}_1,\bar{g}_w,..,\bar{g}_0}= \mathbf{Z}^{\bar{g}_0\oplus \bar{g}_w}\textbf{cX}^{\bar{k}_w}_w\cdot..\cdot \mathbf{Z}^{\bar{g}_1}\textbf{cX}^{\bar{k}_1}_1\mathbf{Z}^{\bar{g}_0}
\end{equation*}
``Moving'' the by-products toward the end of the circuit (i.e. commuting by-products with the various layers of CNOTs via Equations \ref{eq:ctraps}) we obtain
\begin{equation}
\widetilde{U}_C^{\bar{k}_w,..,\bar{k}_1,\bar{g}_w,..,\bar{g}_0}= \mathbf{Z}^{\bar{g}'}\textbf{cX}^{\bar{k}_w}_w\mathbf{Z}^{\bar{g}}\cdot..\cdot \textbf{cX}^{\bar{k}_1}_1\textrm{ ,}
\end{equation}
where $\bar{g}$ depends on the phase-flips affecting the qubits in the tapes $0,1,..,w-1$ and $\bar{g}'=\bar{g}_0\oplus \bar{g}_w$. (For simplicity, we suppose that $w$ is odd. If $w$ is even, the following arguments hold, provided that the by-products are rewritten as errors affecting the computation before and after tape $w-1$). Before applying Equation \ref{eq:cc} and showing that $\widetilde{U}_C\neq U_C$, we need to show that there is no combination of by-products such that $\overline{g}=\bar{0}$ for every $\bar{k}_1,\bar{k}_2,..,\bar{k}_{w-1}$. To see this, consider the possible different cases:
\begin{itemize}
	\item $\overline{g}_{y_1}\neq\overline{0}$ for a given $y_1\in(0,..,w-1)$, all other $\overline{g}_{y}$ equal to $\bar{0}$. Moving errors toward the end of the circuit, equations \ref{eq:ctraps} guarantee that $\overline{g}\neq\overline{0}$ for any choice of $\bar{k}_1,\bar{k}_2,..,\bar{k}_{w-1}$.
	\item $\overline{g}_{y_1},\overline{g}_{y_2}\neq\overline{0}$ for some $y_1,y_2\in(0,..,w-1)$ ($y_1<y_2$), all other $\overline{g}_y$ equal to $\bar{0}$. Here, we use the same idea as the one in Figure $\ref{fig:proofst2}$: (i) we move errors toward each other, until they are separated by a single tape (say tape $d$) and (ii) we find out (if it exists) a specific $\bar{k_d}$ that let errors cancel out with each other. The C-traps where the configuration of CNOTs in tape $d$ is described by $\bar{k}_d\oplus\bar{1}$ do not let the errors cancel out with each other. Thus, ``merging'' the two errors and moving them toward the end of the circuit, we obtain a $\overline{g}$ that is different from $\bar{0}$. Notice that, in some sense, this procedure allows to ``merge'' two errors into a single one described by $\overline{g}$.
\item $\overline{g}_{y_1},\overline{g}_{y_2},\overline{g}_{y_3}\neq\overline{0}$ for some $y_1,y_2,y_3=0,..,w-1$ $(y_1<y_2<y_3)$, all other $\overline{g}_y$ equal $\bar{0}$. In this case, (i) we move errors affecting tape $y_1$ toward tape $y_2$, (ii) we merge the two errors as explained above and (iii) we do the same with the newly generated error affecting tape $y_2$ and the error affecting tape $y_3$. This way, we reduce three-tape errors to single-tape ones.
	\item We repeat the same procedure as above for deviations affecting more than three tapes.
\end{itemize}
Thus, we can indeed apply equation \ref{eq:cc}. We obtain $\widetilde{U}_C\neq U_C$, meaning that for some of the terms of the summation, errors do not trivially cancel out. The particular sequences of $\{\bar{k}_y\}_{y=0}^w$ defining these terms correspond to logical circuits that are corrupted by errors and output a state which is orthogonal to the expected one.\\

\noindent\textbf{Step II.} We first consider the various combinations of Type-I and Type-II errors affecting a single tape $y_1$ (i.e. $\textbf{Z}^{\bar{g}_{y}}=\mathbb{1}_n$ for all $y=0,..,w$ except from some specific $y_1$).  In this case, there is a probability $\leq1/2$ that some by-products are produced (with equality if the the phase-flips solely affect two physical qubits on the same line, inequality otherwise). Moving the by-products toward the end of the circuit and using equations \ref{eq:ctraps}, it can thus be seen that the probability of detecting the error is $\geq1/2$.

\newpage
\begin{small}

	\noindent\makebox[\linewidth]{\rule{8.8cm}{0.4pt}}
	\textbf{Protocol \hypertarget{pr:pr2}{2}.}\\
	\noindent\makebox[\linewidth]{\rule{8.8cm}{0.4pt}}
	\textbf{Hypothesis:} 
	\begin{itemize}
	\item[] Alice can measure qubits in the set of bases $\{|\pm\rangle_{\tau^{(k)}_{i,j}}\langle\pm|\}$, where $\tau^{(k)}_{i,j}\in\{0,\pi/4,..,7\pi/4\}$.
	\end{itemize}
	\textbf{Input}: 
	\begin{itemize}
	\item[(i) ] the number of computations $v$.
		\item[(ii)] the set of measurement angles $\{\phi_{i,j}\}$ for the target computation.
		\item[(iii)] the sets of random variables $\{r_{i,j}^{(k)}=0,1\}$ and $\{r_{i,j}^{\prime(k)}=0,1\}$ and the set of random angles $\{\theta_{i,j}^{(k)}=0,\pi/4,..,7\pi/4\}$ for any computation $k=1,..,v+1$.
	\end{itemize}
	\noindent\textbf{0. Preliminary operation.}\\
	Alice randomly chooses $v_t\in(1,2,\ldots,v+1)$ and sets $\{\phi^{(v_t)}_{i,j}\}=\{\phi_{i,j}\}$.\\
	
	\noindent For $k=1,..,v+1$:
	\begin{itemize}
		\item[]\textbf{1. Assigning measurement angles.}\\
		{If $k\neq v_t$}, Alice randomly runs Sub-protocol \hyperlink{pr:spr2}{1.1} or Sub-protocol \hyperlink{pr:spr2}{1.2} on input $n\times m$ and obtains the set $\{\phi^{(k)}_{i,j}\}$.
		
		\item[]\textbf{2. State preparation.}
		
		For $i=1,..,n$ and for $j=1,..,m$, Alice asks Bob to create eight qubits in the state $\ket{+}_{\tau^{(k)}_{i,j}}$, $\tau^{(k)}_{i,j}\in\{0,\pi/4,..,7\pi/4\}$, and to send them to her. Next, she measures each qubit $\ket{+}_{\tau^{(k)}_{i,j}}$ it in the basis $\{|\pm\rangle_{\tau^{(k)}_{i,j}}\langle\pm|\}$. If measurements output 0, she sends back to Bob the qubit in the state $R_Z(\theta^{(k)}_{i,j}+\pi\sum_{(i',j')\sim(i,j)}^{(k)}r_{i,j}^{\prime(k)})\ket{+}$ and discards the others. Otherwise, she restarts preparation of qubit $(i,j)$.

		\item[]\textbf{3. Blind Computation.}
		\begin{itemize}
			\item[3.1] Bob entangles the qubits in its memory and creates a $(n\times m)$ BwS.
			\item[3.2] For $j=1,..,m$ and for $i=1,..,n$,
			\begin{itemize}
				\item[-]Bob sends Alice the qubit in position $(i,j)$. Alice measures it with angle $\delta^{(k)}_{i,j}=(-1)^{r^{\prime(k)}_{i,j}}\phi^{(k)}_{i,j}+\theta^{(k)}_{i,j}+r^{(k)}_{i,j}\pi$ and obtains outcome $s^{(k)}_{i,j}$.
				\item[-]Alice recomputes the measurement outcome $s^{(k)}_{i,j}$ as $s^{(k)}_{i,j}\oplus r^{(k)}_{i,j}$. Next, she recomputes measurement angles of yet-to-be-measured qubits as $\{(-1)^{s_X}\phi^{(k)}_{i,j}+s_Z\pi\}$.
			\end{itemize}
		\end{itemize}
		
		\item[]\textbf{4. Verification.}\\
		If $k\neq v_t$ and last-column measurement outcomes $\{s^{(k)}_{i,m}\}_{i=1}^n\neq(0,0,..,0)$ , Alice rejects the whole computation.
		
	\end{itemize}

	\noindent\textbf{Output:} Outcomes $\{s^{(v_t)}_{i,m}\}$ of measurements of last-column qubits of the target computation.\\
	\noindent\makebox[\linewidth]{\rule{8.8cm}{0.4pt}}
\end{small}

\newpage
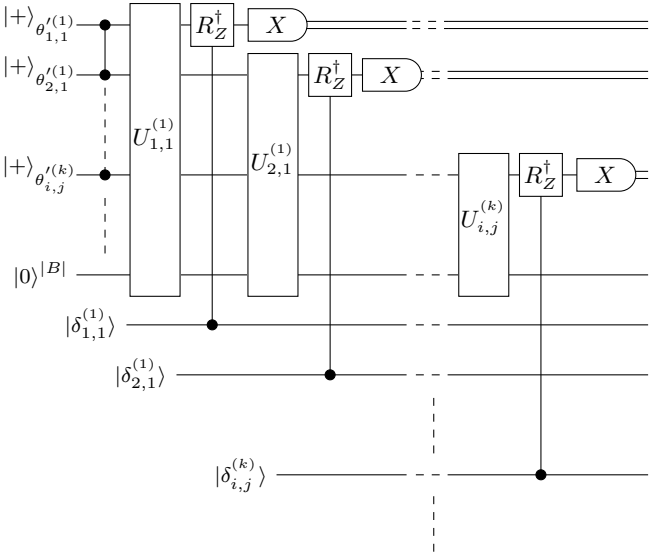
\begin{figure}[H]
	\begin{tikzpicture}[scale=0.95, every node/.style={scale=1}]
	
	\draw (0,4.2) -- (2.7,4.2);
	\draw (2.7,4.25) -- (4.62,4.25);
	\draw (2.7,4.15) -- (4.62,4.15);
	\draw [dashed] (4.7,4.25) -- (5.2,4.25);
	\draw [dashed] (4.7,4.15) -- (5.2,4.15);
	\draw (5.2,4.25) -- (8,4.25);
	\draw (5.2,4.15) -- (8,4.15);
	
	\draw (0,3.5) -- (4.3,3.5);
	\draw (4.3,3.55) -- (4.62,3.55);
	\draw (4.3,3.45) -- (4.62,3.45);
	\draw [dashed] (4.75,3.55) -- (5.2,3.55);
	\draw [dashed] (4.75,3.45) -- (5.2,3.45);
	\draw (5.2,3.55) -- (8,3.55);
	\draw (5.2,3.45) -- (8,3.45);
	
	\draw (0,2.1) -- (4.62,2.1);
	\draw [dashed] (4.75,2.1) -- (5.2,2.1);
	\draw (5.2,2.1) -- (7.8,2.1);
	\draw (7.8,2.15) -- (8,2.15);
	\draw (7.8,2.05) -- (8,2.05);
	
	\draw (0.0,0.7) -- (4.62,0.7);
	\draw [dashed] (4.75,1.4-0.7) -- (5.2,1.4-0.7);
	\draw (5.2,1.4-0.7) -- (8,1.4-0.7);
	
	\draw (0.7,0.7-0.7) -- (4.62,0.7-0.7);
	\draw [dashed] (4.75,0.7-0.7) -- (5.2,0.7-0.7);
	\draw (5.2,0.7-0.7) -- (8,0.7-0.7);
	
	\draw (1.4,-0.7) -- (4.62,-0.7);
	\draw [dashed] (4.75,-0.7) -- (5.2,-0.7);
	\draw (5.2,-0.7) -- (8,-0.7);
	
	\draw (2.8,-2.1) -- (4.62,-2.1);
	\draw [dashed] (4.75,-2.1) -- (5.2,-2.1);
	\draw (5.2,-2.1) -- (8,-2.1);
	
	\draw (0.4,3.5) -- (0.4,4.2);
	\draw [dashed] (0.4,2.1) -- (0.4,3.5);
	\draw [fill] (0.4,4.2) circle [radius=0.7mm];
	\draw [fill] (0.4,3.5) circle [radius=0.7mm];
	\draw [fill] (0.4,2.1) circle [radius=0.7mm];
	
	\node at (-0.5,4.2) {\footnotesize $\ket{+}_{\theta^{^{\prime(1)}}_{1,1}}$};
	\node at (-0.5,3.5) {\footnotesize $\ket{+}_{\theta^{^{\prime(1)}}_{2,1}}$};
	\node at (-0.5,2.1) {\footnotesize $\ket{+}_{\theta^{^{\prime(k)}}_{i,j}}$};
	\node at (0.2,0.0) {\footnotesize $\ket{\delta^{^{(1)}}_{1,1}}$};
	\node at (0.9,-0.7) {\footnotesize $\ket{\delta^{^{(1)}}_{2,1}}$};
	\node at (2.3,-2.1) {\footnotesize $\ket{\delta^{^{(k)}}_{i,j}}$};
	
	\node at (-0.5,0.7) {\footnotesize $\ket{0}^{|B|}$};
	
	\draw (1.9,4.2) -- (1.9,0.0);
	\draw [fill] (1.9,0.0) circle [radius=0.7mm];
	\draw (3.55,3.5) -- (3.55,-0.7);
	\draw [fill] (3.55,-0.7) circle [radius=0.7mm];
	\draw (6.5,2.1) -- (6.5,-2.1);
	\draw [fill] (6.5,-2.1) circle [radius=0.7mm];
	
	\draw [fill=white] (1.6,3.9) rectangle (2.2,4.5);
	\node at (1.9,4.2) {\small $R^{\dagger}_{Z}$};
	\draw [fill=white] (3.25,3.2) rectangle (3.85,3.8);
	\node at (3.55,3.5) {\small $R^{\dagger}_{Z}$};
	\draw [fill=white] (6.2,1.8) rectangle (6.8,2.4);
	\node at (6.5,2.1) {\small $R^{\dagger}_Z$};
	
	\draw [fill=white] (0.75,4.5) rectangle (1.45,0.4);
	\node at (1.1,2.8-0.15) {\small $U^{^{(1)}}_{1,1}$};
	\draw [fill=white] (2.4,3.8) rectangle (3.1,0.4);
	\node at (2.75,2.45-0.15) {\small $U^{^{(1)}}_{2,1}$};
	\draw [fill=white] (5.35,2.4) rectangle (6.05,0.4);
	\node at (5.7,1.65-0.15) {\small $U^{^{(k)}}_{i,j}$};
		
	\draw [dashed] (0.4,1.0) -- (0.4,2.5-0.7);
	\draw [dashed] (5,-1.1-0.7) -- (5,-0.3-0.7);
	\draw [dashed] (5,-1.7-0.7) -- (5,-2.5-0.7);
	
	\node [fill=white,draw,rounded rectangle,rounded rectangle left arc=none,minimum width=1cm] (name) at (2.8,4.2) {$X$};
	\node [fill=white,draw,rounded rectangle,rounded rectangle left arc=none,minimum width=1cm] (name) at (4.4,3.5) {$X$};
	\node [fill=white,draw,rounded rectangle,rounded rectangle left arc=none,minimum width=1cm] (name) at (7.4,2.1) {$X$};
	
	\end{tikzpicture}
	\caption{\small Circuit diagram of a computation on a BwS in the case of trusted measurements. Unitary $U^{^{(k)}}_{i,j}$ represents Bob's deviations before measurement of physical qubit $(i,j)$ along computation $k$. Bob's private register is initialized in the state $\ket{0}^{|B|}$. For simplicity, in the above picture we have rewritten each angle $\theta_{i,j}^{(k)}+\pi\sum_{(i',j')\sim(i,j)}{r}_{i,j}^{\prime(k)}$ as $\theta_{i,j}^{\prime(k)}$.}
	\label{fig:MBQC1circ10}
\end{figure}
\noindent Thus, any single-tape error is detected with probability $\geq$1/2.

The above lower bound remains valid also for errors affecting more tapes. For instance, consider the combinations of Type-I and Type-II errors affecting two nearest-neighbouring tapes $y$ and $y+1$ (i.e. $\textbf{Z}^{\bar{g}_{y}}=\mathbb{1}_n$ for all $y=0,..,w$ except from some specific $y$ and $y+1$). If the phase-flips affect solely two {``overlapping''} bricks (one belonging to tape $y$ and the other to tape $y+1$), the probability of detecting the presence of errors is at least 1/2 (Figure \ref{fig:int1}). Conversely, the combinations of phase-flips that affect more than two overlapping bricks are detected with higher probability (Figures \ref{fig:int2} and \ref{fig:int3}), and the same is true for errors affecting bricks that do not overlap.

The same strategy can be used to show that the probability of detecting combinations affecting two non-nearest-neighbouring tapes is lower bounded by 1/2 (using equations \ref{eq:ctraps}, errors can be moved toward each other and re-written as errors affecting two nearest-neighbouring tapes), as well as for errors affecting more than two tapes (errors can be moved into each other and ``merged'', and so re-written as errors affecting two nearest-neighbouring tapes).
\end{proof}

Summarising, we introduced a verification protocol that makes use of classically efficiently simulable computations to certify the correctness of a universal quantum computation. We described our protocol in the language of cryptographic protocols and subsequently proved that it is correct, blind and verifiable, provided that a trust assumption is made on state preparation. In what follows, we describe another verification protocol where the trust assumption is replaced by another trust assumption, this time made on the measurement device.

\section{Verification for trusted measurements}~\label{sec:verifM}

Suppose that state preparation can not be trusted, while it is possible to reliably measure qubits in eight different bases $\{|\pm\rangle_{\theta}\langle\pm|\}$, $\theta\in\{0,\pi/4,..,7\pi/4\}$ and subsequently reuse them. We now show that in this case, it is possible to use measurements to certify the correctness of state preparation and subsequently verify the quantum computation by means of the scheme described in the previous section.

The idea behind Protocol \hyperlink{pr:pr2}{2} is the following. Suppose that instead of generating a qubit in a state $\ket{\psi}$, a noisy state preparation device generates the state $\E(|\psi\rangle\langle\psi|)$, where $\E$ represents a CPTP-map. Regardless of the particular form of $\E$, after a measurement in the basis $\{|\psi\rangle\langle\psi|, |\psi_\perp\rangle\langle\psi_{\perp}|\}$ (where $\langle\psi|\psi_\perp\rangle=0$) the state collapses either into $\ket{\psi}$ or $\ket{\psi_\perp}$. Thus, despite the potential malfunctioning of the state preparation device, the collapse of the state function into one of the eigenvectors composing the measurement basis guarantees that if the outcome obtained is the expected one, then the state of the qubit after the measurement is itself the expected one. Based on the above idea, we show that Alice can blindly generate the same input state as in Protocol \hyperlink{pr:pr1}{1}, and thus use a similar technique.

\subsection{Description of the Protocol}

State preparation in Protocol \hyperlink{pr:pr2}{2} works as follows. For any qubit $(i,j)$ belonging to any of the graphs (traps and target), the set of single-qubit states $|+\rangle_{\theta}$, $\theta\in\{0,\pi/4,..,7\pi/4\}$, is prepared. Next, each of these eight states is measured in the same basis as prepared. If all of the measurements output 0, then the qubits are indeed in the correct state. A qubit is thus chosen at random and used for the computation, while the others are discarded. In this way, one can produce all of the graphs required by the verification protocol illustrated in the previous section. 

We now provide a more detailed description of our scheme in the language of cryptographic protocols. The roles played by Alice and Bob within the interactive game are described by Protocol \hyperlink{pr:pr2}{2}, that is a protocol for Alice making measurements in the discrete set of bases $\{|\pm\rangle_{\phi}\langle\pm|\}$, $\phi\in\{0,\pi/4,..,7\pi/4\}$. As in Protocol \hyperlink{pr:pr1}{1}, Alice verifies Bob's behaviours by hiding the target computation among $v$ traps of the kind explained in the previous section. 

The main difference with Protocol \hyperlink{pr:pr1}{1} regards the preparation of the physical qubits belonging to the various graphs. Notice that for state preparation Alice needs a short-time single-qubit memory if qubits are measured one after the other, or eight different measurements devices if qubits are measured at the same time.

For any computation $k=1,..,v+1$, the physical qubits belonging to the BwS are prepared as follows. For any physical qubit $(i,j)$, Bob creates eight qubits in the state $\ket{+}_{\tau^{(k)}_{i,j}}$, $\tau^{(k)}_{i,j}\in(0,\pi/4,..,7\pi/4)$, and sends them to Alice. Alice measures each qubit $\ket{+}_{\tau^{(k)}_{i,j}}$ by angle $\tau^{(k)}_{i,j}$ and checks the outcomes. If all of the measurements yield 0, she concludes that she holds the correct eight-qubit state, sends a randomly chosen qubit back to Bob and discards the remaining ones. On the contrary, if some measurement yields 1, she concludes that Bob is providing her with qubits in the wrong state, hence restarts state preparation of qubit $(i,j)$.
 
After receiving every physical qubit belonging to the $k^{th}$ BwS, Bob stores all of them in his register and finally creates the graph. Computational measurements are performed by Alice, who does not disclose any information about computational angles nor about measurement outcomes. 

\subsection{Overheads}
The amount of qubits sent by Alice to Bob equals $N^{\textup{A}}_{\textup{qubit}}=(v+1)nm$, while Bob sends to Alice $N^{\textup{B}}_{\textup{qubit}}=9(v+1)nm$ qubits (namely $8(v+1)nm$ during state preparation and $(v+1)nm$ during the actual computation). No bits are exchanged during the protocol. Although the interaction might potentially be done off-line, it is in fact done on-line. This is because Alice can not store in her register all of the physical qubits, hence she needs Bob to use the communication channel several times.

\subsection{Correctness, Blindness and Verifiability}

We now illustrate the results regarding Protocol \hyperlink{pr:pr2}{2}. First, we formally prove that state preparation outputs the correct state, regardless of Bob's deviations (Lemma \ref{lem:sp}). Next, we show that Protocol \hyperlink{pr:pr2}{2} is correct, blind and verifiable.

\begin{lemma}~\label{lem:sp}
If Alice accepts, then state preparation (namely step 2 of Protocol \hyperlink{pr:pr2}{2}) always produces the correct outcome.
\end{lemma}

\begin{proof}
Consider state preparation of qubit $(i,j)$ along computation $k$. Suppose that Bob sends Alice the eight-qubit state
\begin{equation}
\widetilde{\rho}_{i,j}=\textrm{Tr}_{B}\bigg[\E_{i,j}\bigg(\otimes_{\tau_{i,j}}|+\rangle_{\tau_{i,j}}\langle+|\bigg)\bigg]\textrm{ ,}
\end{equation}
where $\E_{i,j}$ represents a CPTP-map and ${B}$ labels Bob's private system. Consider the case where measurements yield 0 (otherwise the whole procedure is restarted). In this case, the measurement operation maps $\widetilde{\rho}_{i,j}$ into $\bigotimes_{l=0}^7|+\rangle_{l\pi/4}\langle+|$, i.e. the correct eight-qubit state.
\end{proof}

\begin{theorem}
\textup{\textbf{[Correctness]}} Protocol \hyperlink{pr:pr2}{2} is correct.
\end{theorem}
\begin{proof}
Correctness can be proven with the same arguments as those used for Protocol \hyperlink{pr:pr1}{1}, even though Alice makes the measurements instead of Bob. Indeed, the main difference between the two protocols consists in state preparation, which is correct by Lemma \ref{lem:sp}.
\end{proof}

\begin{theorem}
\textup{\textbf{[Blindness]}} Protocol \hyperlink{pr:pr2}{2} is blind.
\end{theorem}
\begin{proof}
Blindness is guaranteed by the no-communication theorem \cite{NC00}. The no-communication theorem states that Alice can not send any kind information to Bob by solely making measurements on her part of the system, even if her system is entangled with Bob's one. Thus, Bob can not retrieve any information about measurement angles.\\
\end{proof}

\begin{theorem}~\label{theorem:ver2}
\textup{\textbf{[Verifiability] }}For any $v\geq7$, Protocol \hyperlink{pr:pr2}{2} is $\varepsilon$-verifiable with soundness
\begin{equation}~\label{eq:oneround}
\varepsilon=\frac{7}{v+1}\bigg(\frac{7}{8}\bigg)^{6}\cong\frac{3.14}{v+1}\textrm{ ,}
\end{equation}
where $v$ represents the number of trap computations.
\end{theorem}

\begin{proof}
Completeness can be proven with the same arguments used for Protocol \hyperlink{pr:pr1}{1}, together with correctness of state preparation (Lemma \ref{lem:sp}). 

To prove soundness, suppose that for any qubit $(i,j)$ of any computation $k\in(1,..,v+1)$, the measurements performed during state preparation yield outcome 0. As a consequence, the qubits sent by Alice to Bob are in the correct state (Lemma \ref{lem:sp}) and the computation can be described by the circuit in Figure \ref{fig:MBQC1circ10}. It is possible to simplify the circuit by (i) moving deviations toward the end of the circuit and merging them into a single unitary $U_B$, (ii) expressing controlled-$R_Z$ gates as uncontrolled rotations and (iii) simplifying pre-rotations in the states with $\theta_{i,j}^{(k)}$  in the angles. This yields the circuit in Figure \ref{fig:MBQC1circ7}, and the proof of the theorem follows from the proof of verifiability of Protocol \hyperlink{pr:pr1}{1}. 
\end{proof}
Notice that the deviations in circuit in Figure \ref{fig:MBQC1circ10} do not affect the measurement angles $\delta_{i,j}^{(k)}$. This is due to the fact that Alice performs all of the measurements and Bob can not guess any measurement angle. Thus, one may wonder whether verifiability of Protocol \hyperlink{pr:pr2}{2} indeed requires the encryption of the initial BwS by means of the eight angles $\theta_{i,j}^{(k)}\in\{0,\pi/4,..,7\pi/4\}$, and not by angles chosen uniformly at random from a smaller set (such as, for instance, $\theta_{i,j}^{(k)}\in\{0,\pi/2,\pi,3\pi/2\}$, which on its own would be enough to prevent Bob from getting any information about the initial state of the BwS). The answer is yes, and the intuitive reason is the following. Because of the presence of controlled rotations, the ``overall'' deviation $U_B$ affects measurement angles as well (Figure \ref{fig:MBQC1circ7}). If measurement angles are not ``one-time-padded'' by the angles $\theta_{i,j}^{(k)}$, $U_B$ carries a dependency on the angles $\delta_{i,j}^{(k)}$. In turn, this does not allow to sum over the random parameters $r_{i,j}^{(k)}$ and $r_{i,j}^{\prime {(k)}}$, which is fundamental to reduce deviations to a convex combination of Pauli by-products. Thus, the encryption of the initial state of the BwS by means of angles $\theta_{i,j}^{(k)}$ chosen at random from a eight-element set is crucial to verifiability of Protocol \hyperlink{pr:pr2}{2}.

\section{Conclusion and further steps}~\label{sec:conclusions}
We illustrated a verification protocol that certifies the correctness of a quantum computation by means of several other classically efficiently simulable computations. We proved that our technique is valid provided that a trust assumption is made on state preparation. Next, we adapted it to the case where measurements are trusted and qubits can be reused after the measurements. Our protocols restricts the trust assumptions to operations on the $XY$-plane of the Bloch sphere while maintaining a linear overhead.

An experimentally relevant open question regarding our trap computation technique is whether it can be adapted to a scenario where measurements are indeed trusted, but measured qubits can not be reused. Reusing already measured qubits is practically impossible in many experimental implementations of quantum computing due to practical issues, thus qubits are often re-initialized into the desired state after the measurement. In this case, trusting measurements implicitly means trusting state preparation as well, which defeats the whole purpose of our protocol. Another open question regards the possibility of making our protocols ``device independent'', namely security in a scenario where state preparation and measurement devices are both untrusted, but a space-like separation prevents them from communicating \cite{GKW15}. Finally, as a future step, we intend to exploit the ``symmetry'' characterising our protocols (namely the fact that state preparation or measurements are performed in the $XY$-plane of the BwS, and we make no use of qubits in the states $\ket{0}$ and $\ket{1}$, nor of measurements in Pauli-$Z$ basis) to adapt it to a multi-party scenario, such as \cite{KP16}. 

{\section*{Acknowledgements}}
\noindent{This research was supported by the UK EPSRC (EP/K04057X/2) and the UK Networked Quantum Information Technologies (NQIT) Hub (EP/M013243/1). We acknowledge helpful discussions with Elham Kashefi and Dominic Branford.}

\bibliography{verification}
\bibliographystyle{ieeetr}

\begin{appendix}

\onecolumngrid
\newpage
\section{Correctness of Protocol 1}~\label{app:corr}
Here we give a proof of the correctness of Protocol \hyperlink{pr:pr1}{1}. If Bob follows Alice's instructions, the state of the system immediately before the measurements is of the form (Figure \ref{fig:MBQC1circ6}, where the deviations are set to the identity)
\begin{align}
{\sigma}_{\textup{out}}&=\sum_{\overline{\theta},\overline{r},\overline{r}'}\frac{{cR}_{n,m}^{\dagger(v+1)}..{cR}_{1,1}^{\dagger(1)}E\cdot(\sigma^{\overline{\theta},\overline{r},\overline{r}'}_{\textup{in}})}{2^{2nm(v+1)}8^{nm(v+1)}}
\end{align}
In the above expression, we use the notation $O\cdot\sigma=O\sigma O^{\dagger}$ for an operator $O$ and a state $\sigma$. Each operator ${cR}_{i,j}^{\dagger(k)}$ represents the rotation acting on qubit $(i,j)$ of computation $k$ and controlled by angle $\delta_{i,j}^{(k)}$, $E$ represents the entangling operation and the summation is made over all the combinations $\overline{\theta}=\{\theta_{i,j}^{(k)}\}$, $\overline{r}=\{r_{i,j}^{(k)}\}$ and $\overline{r}'=\{r_{i,j}^{\prime(k)}\}$. For any fixed $\overline{\theta}$, $\overline{r}$ and $\overline{r}'$, the state $\sigma^{\overline{\theta},\overline{r},\overline{r}'}_{\textup{in}}$ is equal to
\begin{align}
\sigma_{\textup{in}}^{\overline{\theta},\overline{r},\overline{r}'}&=\bigotimes_{i,j,k}|+\rangle_{\theta_{i,j}^{^{\prime(k)}}}\langle+|\otimes|\delta_{i,j}^{(k)}\rangle\langle\delta_{i,j}^{(k)}|=\bigotimes_{i,j,k}\bigg[R_{Z}\bigg(\theta_{i,j}^{(k)}+\pi\sum_{(i',j')\sim(i,j)}^{(k)}r_{i',j'}^{\prime(k)}\bigg)\cdot|+\rangle_{i,j}^{(k)}\langle+| \bigg]\otimes|\delta_{i,j}^{(k)}\rangle\langle\delta_{i,j}^{(k)}|\textrm{ ,}
\end{align}
After rewriting the controlled rotations as un-controlled $R_Z$-gates, the classical registers containing the information about the measurement angles can be traced out. Denoting as $\rho^{\overline{\theta},\overline{0},\overline{r}'}_{\textup{in}}$ the input state obtained by tracing out the angles (notice that it has no dependency on $\overline{r}$), ${\sigma}_{\textup{out}}$ can be rewritten as
\begin{align*}
{\rho}_{\textup{out}}&=\sum_{\overline{\theta},\overline{r},\overline{r}'}\frac{\big[{R_{Z}^{\dagger}}(\delta_{n,m}^{(v+1)})\otimes..\otimes{R_{Z}^{\dagger}}(\delta_{1,1}^{(1)})\big]E\cdot(\rho^{\overline{\theta},\overline{0},\overline{r}'}_{\textup{in}})}{2^{2nm(v+1)}8^{nm(v+1)}}\cr
	&=\sum_{\overline{\theta},\overline{r},\overline{r}'}\frac{\big[ Z^{r_{n,m}^{(v+1)}}R_{Z}^{\dagger}\big((-1)^{r_{n,m}^{\prime(v+1)}}\phi_{n,m}^{(v+1)}\big){R^{\dagger}_{Z}}(\theta_{n,m}^{(v+1)})\otimes..\otimes Z^{r_{1,1}^{(1)}}{R_{Z}^{\dagger}}\big((-1)^{r_{1,1}^{\prime(1)}}\phi_{1,1}^{(1)}\big){R_{Z}^{\dagger}}(\theta_{1,1}^{(1)})\big]E\cdot(\rho^{\overline{\theta},\overline{0},\overline{r}'}_{\textup{in}})}{2^{2nm(v+1)}8^{nm(v+1)}}
\end{align*}
Since the rotations in the $XY$-plane commute with $E$, we can commute the rotations by angle $\theta_{i,j}^{(k)}$ with $E$ and cancel them out. Thus, summing over $\overline{\theta}$,
\begin{align}
{\rho}_{\textup{out}}&=\sum_{\overline{r},\overline{r}'}\frac{\big[ Z^{r_{n,m}^{(v+1)}}X^{r_{n,m}^{\prime(v+1)}}{R^{\dagger}_{Z}}(\phi_{n,m}^{(v+1)})X^{r_{n,m}^{\prime(v+1)}}\otimes..\otimes Z^{r_{1,1}^{(1)}}X^{r_{1,1}^{\prime(1)}}{R^{\dagger}_{Z}}(\phi_{1,1}^{(1)})X^{r_{1,1}^{\prime(1)}}\big]E\cdot(\rho^{\overline{0},\overline{0},\overline{r}'}_{\textup{in}})}{2^{2nm(v+1)}}\textrm{ ,}
\end{align}
where we also rewrote each $R_Z$-gate by angle $-\phi$ as $XR_Z(\phi)X$. Commuting the Pauli-$X$ on the right-hand side of the remaining rotations with $E$, ${\rho}_{\textup{out}}$ becomes
\begin{align}
{\rho}_{\textup{out}}&=\sum_{\overline{r},\overline{r}'}\frac{\big[ Z^{r_{n,m}^{(v+1)}}X^{r_{n,m}^{\prime(v+1)}}{R^{\dagger}_{Z}}(\phi_{n,m}^{(v+1)})\otimes..\otimes Z^{r_{1,1}^{(1)}}X^{r_{1,1}^{\prime(1)}}{R^{\dagger}_{Z}}(\phi_{1,1}^{(1)})\big]E\cdot(\rho^{\overline{0},\overline{0},\overline{0}}_{\textup{in}})}{2^{2nm(v+1)}}
\end{align}
After the measurement of the first qubit, we obtain
\begin{align*}
{\rho'}_{\textup{out}}=&\sum_{s_{1,1}^{(1)}}\bigg[\bigg({}_{1,1}^{(1)}\langle+|Z^{s_{1,1}^{(1)}}\otimes\mathbb{1}\bigg)\bigg(\sum_{\overline{r},\overline{r}'}\frac{\big[ Z^{r_{n,m}^{(v+1)}}X^{r_{n,m}^{\prime(v+1)}}{R^{\dagger}_{Z}}(\phi_{n,m}^{(v+1)})\otimes..\otimes Z^{r_{1,1}^{(1)}}X^{r_{1,1}^{\prime(1)}}{R^{\dagger}_{Z}}(\phi_{1,1}^{(1)})\big]E\cdot(\rho^{\overline{0},\overline{0},\overline{0}}_{\textup{in}})}{2^{2nm(v+1)}}\bigg)\bigg(Z^{s_{1,1}^{(1)}}|+\rangle_{1,1}^{(1)}\otimes\mathbb{1}\bigg)\bigg]\cr
	&\otimes Z^{s_{1,1}^{(1)}}|+\rangle\langle+|Z^{s_{1,1}^{(1)}}\cr
	=&\sum_{s_{1,1}^{(1)}, r_{1,1}^{(1)}}\bigg[\bigg({}_{1,1}^{(1)}\langle+|Z^{s_{1,1}^{(1)}\oplus r_{1,1}^{(1)}}\otimes\mathbb{1}\bigg)\bigg(\sum_{r_{2,1}^{(1)},..}\frac{\big[ Z^{r_{n,m}^{(v+1)}}X^{r_{n,m}^{\prime(v+1)}}{R^{\dagger}_{Z}}(\phi_{n,m}^{(v+1)})\otimes..\otimes {R^{\dagger}_{Z}}(\phi_{1,1}^{(1)})\big]E\cdot(\rho^{\overline{0},\overline{0},\overline{0}}_{\textup{in}})}{2^{nm(v+1)}}\bigg)\bigg(Z^{s_{1,1}^{(1)}\oplus r_{1,1}^{(1)}}|+\rangle_{1,1}^{(1)}\otimes\mathbb{1}\bigg)\bigg]\cr
	&\otimes Z^{s_{1,1}^{(1)}}|+\rangle\langle+|Z^{s_{1,1}^{(1)}}\textrm{ ,}
\end{align*}
where $\mathbb{1}$ is the identity on the rest of the system. In the last line we used the fact that Pauli-$X$ operators stabilize qubits in the $\ket{+}$ state to cancel out the residual Pauli-$X$ acting on the measured qubit.
When Alice recomputes the measurement outcome $s_{1,1}^{(1)}$ as $s_{1,1}^{(1)}\oplus r_{1,1}^{(1)}$ (step 3.2 of the Protocol), we obtain
\begin{align*}
{\rho'}_{\textup{out}}=&\sum_{s_{1,1}^{(1)}, r_{1,1}^{(1)}}\bigg[\bigg({}_{1,1}^{(1)}\langle+|Z^{s_{1,1}^{(1)}\oplus r_{1,1}^{(1)}}\otimes\mathbb{1}\bigg)\bigg(\sum_{r_{2,1}^{(1)},..}\frac{\big[ Z^{r_{n,m}^{(v+1)}}X^{r_{n,m}^{\prime(v+1)}}{R^{\dagger}_{Z}}(\phi_{n,m}^{(v+1)})\otimes..\otimes {R^{\dagger}_{Z}}(\phi_{1,1}^{(1)})\big]E\cdot(\rho^{\overline{0},\overline{0},\overline{0}}_{\textup{in}})}{2^{nm(v+1)}}\bigg)\bigg(Z^{s_{1,1}^{(1)}\oplus r_{1,1}^{(1)}}|+\rangle_{1,1}^{(1)}\otimes\mathbb{1}\bigg)\bigg]\cr
	&\otimes Z^{s_{1,1}^{(1)}\oplus r_{1,1}^{(1)}}|+\rangle_{1,1}^{(1)}\langle+|Z^{s_{1,1}^{(1)}\oplus r_{1,1}^{(1)}}\cr
	=&\sum_{s_{1,1}^{(1)}}\bigg[\bigg({}_{1,1}^{(1)}\langle+|Z^{s_{1,1}^{(1)}}\otimes\mathbb{1}\bigg)\bigg(\sum_{r_{2,1}^{(1)},..}\frac{\big[ Z^{r_{n,m}^{(v+1)}}X^{r_{n,m}^{\prime(v+1)}}{R^{\dagger}_{Z}}(\phi_{n,m}^{(v+1)})\otimes..\otimes {R^{\dagger}_{Z}}(\phi_{1,1}^{(1)})\big]E\cdot(\rho^{\overline{0},\overline{0},\overline{0}}_{\textup{in}})}{2^{nm(v+1)}}\bigg)\bigg(Z^{s_{1,1}^{(1)}}|+\rangle_{1,1}^{(1)}\otimes\mathbb{1}\bigg)\bigg]\cr
	&\otimes Z^{s_{1,1}^{(1)}}|+\rangle_{1,1}^{(1)}\langle+|Z^{s_{1,1}^{(1)}}
\end{align*}
To obtain the second equality, one needs to define a new variable $s_{1,1}^{\prime(1)} = s_{1,1}^{(1)}+r_{1,1}^{(1)}$ and make the change of variable $s_{1,1}^{(1)}\rightarrow s_{1,1}^{\prime(1)}$ (which is possible because $s_{1,1}^{(1)}$ and $r_{1,1}^{(1)}$ are independent parameters). Importantly, this change of variable removes all dependence of the angles $\phi_{i,j}^{(k)}$ of the other qubits from $r_{1,1}^{(1)}$. 

Overall, the state ${\rho'}_{\textup{out}}$ represents the state of the system after the first qubit has been measured. As it can be seen, it does not depend on any of the random variables associated to the first qubit, namely $\theta_{1,1}^{(1)}$, $r_{1,1}^{(1)}$ and $r_{1,1}^{\prime(1)}$. Repeating the same calculation as the above one for every other qubits (keeping in mind that the measurement angles of yet-to-be-measured qubits have to be recomputed as in step 3.2 of the protocol, and namely as $\{(-1)^{s_X}+s_Z\pi\}$), it can be shown that the dependency on the remaining random variables vanishes as well. The correctness of Protocol \hyperlink{pr:pr1}{1} can finally be proven based on the correctness of MBQC. 

\section{Blindness of Protocol 1}~\label{app:oldproofs}
Here, we give a proof of Theorem \ref{th:bl1}, which states that Protocol \hyperlink{pr:pr1}{1} is blind. Blindness of Protocol \hyperlink{pr:pr1}{1} can be proven with similar arguments as for other protocols in the prepare-and-send class  \cite{D12,BFK09}.

After Bob receives all of the physical qubits and the measurement angles, he holds the state
\begin{align}
\rho_B &=\frac{1}{2^{2(v+1)nm}}\frac{1}{8^{(v+1)nm}}\sum_{\overline{\phi},\overline{\theta},\overline{r},\overline{r^{\prime}}}p_{\overline{\phi}}\bigotimes_{i,j,k}|+\rangle_{\theta^{\prime(k)}_{i,j}}\langle+|\otimes|\delta_{i,j}^{(k)}\rangle\langle\delta_{i,j}^{(k)}|\cr
	&=\frac{1}{2^{2(v+1)nm}}\frac{1}{8^{(v+1)nm}}\sum_{\overline{\phi},\overline{\theta},\overline{r},\overline{r^{\prime}}}p_{\overline{\phi}}\bigotimes_{i,j,k}\bigg[R_Z\big({\theta}^{\prime(k)}_{i,j}\big)\cdot|+\rangle_{i,j}^{(k)}\langle+|\bigg]\otimes|\delta_{i,j}^{(k)}\rangle\langle\delta_{i,j}^{(k)}|\textrm{ ,}
\end{align} 
where, for simplicity, we have used the notation $O\cdot\sigma$ to indicate the action $O\sigma O^{\dagger}$ of an operator $O$ on a state $\sigma$. In the above expression, $\theta_{i,j}^{^{\prime(k)}}$ is equal to $\theta_{i,j}^{^{(k)}}+\pi\sum_{(i',j')\sim(i,j)}^{(k)}r_{i',j'}^{\prime(k)}$, while $\overline{\phi}$ represents the set of computational angles $\{\phi_{i,j}^{(k)}\}$ for any computation $k$, $\overline{\theta}$ represents the set of angles $\{\theta_{i,j}^{(k)}\}$ for any computation $k$, $\overline{r}$ and $\overline{r^{\prime}}$ represent the sets of random parameters $\{r_{i,j}^{(k)}\}$ and $\{r_{i,j}^{\prime(k)}\}$ and $p(\overline{\phi})$ is the probability of the specific set of angles $\overline{\phi}$ representing Bob's prior knowledge (if Bob has no has no prior knowledge of Alice's desired computation, then $p(\overline{\phi})=1/8^{nm(v+1)}$ for any $\overline{\phi}$). The state $\rho_B$ can be rewritten as
\begin{align*}
\rho_B=\frac{1}{2^{2(v+1)nm}}&\frac{1}{8^{(v+1)nm}}\sum_{\overline{\phi},\overline{\theta},\overline{r},\overline{r^{\prime}}}p_{\overline{\phi}}\bigotimes_{i,j,k}\bigg[R_Z\bigg(\delta^{(k)}_{i,j}-(-1)^{r_{i,j}^{\prime(k)}}\phi_{i,j}^{(k)}+\pi
{r}^{(k)}_{i,j}+\pi\sum_{(i',j')\sim(i,j)}^{(k)}r_{i',j'}^{\prime(k)}\bigg)\cdot|+\rangle_{i,j}^{(k)}\langle+|\bigg]\otimes|\delta_{i,j}^{(k)}\rangle\langle\delta_{i,j}^{(k)}|
\end{align*}
Suppose that we act on the above state with the classically controlled unitary $CCU$, whose action on a state $\sigma$ and on an angle $\delta$ is defined as
\begin{equation}
CCU\big(\sigma\otimes|{\delta}\rangle\langle\delta|\big)CCU^{\dagger}=R_Z{(-\delta)}\sigma R_Z^{\dagger}(-\delta)\otimes|{\delta}\rangle\langle\delta|
\end{equation}
Thus, $\rho_B$ becomes
\begin{align*}
\rho'_B=CUU\cdot\rho_B=\frac{1}{2^{2(v+1)nm}}\frac{1}{8^{(v+1)nm}}\sum_{\overline{\phi},\overline{\theta},\overline{r},\overline{r^{\prime}}}p_{\overline{\phi}}&\bigotimes_{i,j,k}\bigg[R_Z\bigg(-(-1)^{r_{i,j}^{\prime(k)}}\phi_{i,j}^{(k)}+\pi{r}^{(k)}_{i,j}+\pi\sum_{(i',j')\sim(i,j)}^{(k)}r_{i',j'}^{\prime(k)}\bigg)\cdot|+\rangle_{i,j}^{(k)}\langle+|\bigg]\cr
&\otimes|\delta_{i,j}^{(k)}\rangle\langle\delta_{i,j}^{(k)}|
\end{align*}
Summing over all possible $\overline{\theta}$ (which are now only contained in the angles $\delta_{i,j}^{(k)}$ in the classical register), we obtain
\begin{equation*}
\rho'_B=\frac{1}{2^{2(v+1)nm}}\sum_{\overline{\phi},\overline{r},\overline{r^{}}^{\prime}}p_{\overline{\phi}}\bigotimes_{i,j,k}\bigg[R_Z\bigg(-(-1)^{r_{i,j}^{\prime(k)}}\phi_{i,j}^{(k)}+\pi{r}^{(k)}_{i,j}+\pi\sum_{(i',j')\sim(i,j)}^{(k)}r_{i',j'}^{\prime(k)}\bigg)\cdot|+\rangle_{i,j}^{(k)}\langle+|\bigg]\otimes\mathbb{1}^{\otimes3(v+1)nm}
\end{equation*}
where $\mathbb{1}^{\otimes t}$ represents the identity on $t$ qubits. Finally, we can sum over the random parameter $r_{i,j}$. Notice that we can not sum over all of them at the same time. Intuitively, since Alice recomputes every measurement outcome $s_{i,j}^{(k)}$ as $s_{i,j}^{(k)}\oplus r_{i,j}^{(k)}$ and subsequently redefines the set of measurement angles based on the recomputed outcome (step 3.2 of the protocol), the angles $\phi_{i,j}^{(k)}$ depend on the parameters $r_{i,j}^{(k)}$ of previous qubit. Therefore, we start from the summations over $r_{i,m}^{(k)}$ (i.e. from the $r_{i,j}^{(k)}$ of last-column qubits). We obtain
\begin{align*}
\rho'_B=\frac{1}{2^{2(v+1)n(m-1)}}\sum_{\overline{\phi}}p_{\overline{\phi}}&\bigotimes_{i,j<m,k}\sum_{r_{i,j}^{(k)},r_{i,j}^{\prime(k)}}\bigg[R_Z\bigg(-(-1)^{r_{i,j}^{\prime(k)}}\phi_{i,j}^{(k)}+\pi{r}^{(k)}_{i,j}+\pi\sum_{(i',j')\sim(i,j)}^{(k)}r_{i',j'}^{\prime(k)}\bigg)\cdot|+\rangle_{i,j}^{(k)}\langle+|\otimes\mathbb{1}^{m(v+1)}\bigg]\cr
&\otimes\mathbb{1}^{\otimes3(v+1)nm}
\end{align*}
Similarly, summing right-to-left (i.e. column per column, from $j=m-1$ to $j=1$) over the remaining random parameters , we obtain
\begin{equation*}
\rho'_B=\sum_{\overline{\phi}}p_{\overline{\phi}}\textrm{ }\mathbb{1}^{\otimes(v+1)nm}\otimes\mathbb{1}^{\otimes3(v+1)nm}\textrm{ = }\mathbb{1}^{\otimes(v+1)nm}\otimes\mathbb{1}^{\otimes3(v+1)nm}
\end{equation*}
The above state is completely mixed. It is equivalent to the one in Bob's register up to the unitary $CUU$. Thus, Bob holds in his register the completely mixed state, and can not retrieve any information about the computational angles other than he originally had before the protocol.

\end{appendix}

\end{document}